\documentclass[11pt]{article}
\usepackage{graphicx}
\graphicspath{{Figures/}}
\usepackage[sort&compress, numbers]{natbib}
\usepackage{amsmath}
\usepackage{url}
\usepackage{amssymb}
\usepackage{stmaryrd}
\usepackage{algorithm}
 \usepackage{amsthm}
\usepackage{tabu}
\usepackage{xcolor}
\usepackage{colortbl}
 \usepackage{fullpage}
\usepackage{xspace}
\usepackage[footnotesize]{caption}
\usepackage[caption=false,font=footnotesize]{subfig}

\usepackage{tikz}
\usetikzlibrary{external}
\usetikzlibrary{arrows}

\usepackage{pgfplots}
\usepackage{multirow}
\usepackage{rotating}
\definecolor{mycolor1}{rgb}{0.83,0.83,1}

\newtheorem{thm}{Theorem}
\newtheorem{lem}{Lemma}
\newtheorem{proposition}{Proposition}
\newtheorem{cor}{Corollary}
\newtheorem{definition}{Definition}
\newtheorem{remark}{Remark}
\usepackage{color}
\usepackage{hyperref}
\usepackage{cleveref}
\newcommand{\CI}{{\rm ci}}
\newcommand{\SI}{{\rm si}}
\newcommand{\Opt}{{\rm opt}}
\newcommand{\Nyq}{{\rm nyq}}
\newcommand{\range}[1]{\llbracket #1 \rrbracket}
\newcommand{\Id}{\bs I}

\newcommand{\inv}[1]{\frac{1}{#1}}

\newcommand{\tinv}[1]{{\textstyle\frac{1}{#1}}}

\newcommand{\ud}{\mathrm{d}} 
\newcommand{\Nhs}{N_{\rm hs}}
\newcommand{\Nhsexp}{{}^{\hspace{-0.1mm}N_{\hspace{-0.1mm}\scalebox{.4}{\rm hs}}}}

\DeclareMathOperator{\st}{s.t.}

\DeclareMathOperator{\ve}{vec}

\DeclareMathOperator{\diag}{diag}

\DeclareMathOperator*{\argmin}{arg\,min}

\newcommand{\bs}{\boldsymbol}
\newcommand{\bb}{\mathbb}
\newcommand{\cl}{\mathcal}

\newcommand{\ts}{\textstyle}
\newcommand{\ie}{\emph{i.e.},\xspace}
\newcommand{\eg}{\emph{e.g.},\xspace}
\newcommand{\etal}{\emph{et al.}}

\newcommand{\nell}{{l}}

\newcommand{\iid}{%
  \ifmmode
  \mathrm{i.i.d.}%
  \else%
  i.i.d.\@\xspace%
  \fi%
}
\newcommand{\rv}{\mbox{r.v.}\xspace}
\newcommand{\rvs}{\mbox{r.v.s}\xspace}

\begin{document}

\title{A Variable Density Sampling Scheme\\ for Compressive Fourier Transform Interferometry}
\author{A.~Moshtaghpour\thanks{ISPGroup, ICTEAM/ELEN, UCLouvain, Belgium 
    (\url{{amirafshar.moshtaghpour,valerio.cambareri,laurent.jacques}@uclouvain.be}). AM is funded by the FRIA/FNRS. LJ and VC is funded by the F.R.S.-FNRS. Part of this study is funded by the project
    AlterSense (MIS-FNRS).}\and L.~Jacques\footnotemark[1] \and V.~Cambareri\footnotemark[1] \and P.~Antoine\thanks{Lambda-X, Nivelles, Belgium (\url{{pantoine,mroblin}@lambda-x.com})} \and~M.~Roblin\footnotemark[2]%
}

\markboth{}%
{}

\maketitle

\begin{abstract} 
Fourier Transform Interferometry (FTI) is an appealing Hyperspectral (HS) imaging modality for many applications demanding high spectral resolution, \eg in fluorescence microscopy. However, the effective resolution of FTI is limited by the durability of biological elements when exposed to illuminating light. Over-exposed elements are indeed subject to \emph{photo-bleaching} and become unable to fluoresce. In this context, the acquisition of biological HS volumes based on sampling the Optical Path Difference (OPD) axis at Nyquist rate leads to unpleasant trade-offs between spectral resolution, quality of the HS volume, and light exposure intensity. In this paper we propose two variants of the FTI imager, \ie Coded Illumination-FTI (CI-FTI) and Structured Illumination FTI (SI-FTI), based on the theory of compressive sensing (CS). These schemes efficiently modulate light exposure temporally (in CI-FTI) or spatiotemporally (in SI-FTI). Leveraging a variable density sampling strategy recently introduced in CS, we provide near-optimal illumination strategies, so that the light exposure imposed on a biological specimen is minimized while the  spectral resolution is preserved. Our theoretical analysis focuses on two criteria: \emph{(i)} a trade-off between exposure intensity and the quality of the reconstructed HS volume for a given spectral resolution; \emph{(ii)} maximizing HS volume quality for a fixed spectral resolution and constrained light exposure budget. Our contributions can be adapted to an FTI imager without hardware modifications. The reconstruction of HS volumes from compressive sensing FTI measurements relies on an $\ell_1$-norm minimization problem promoting a spatiospectral sparsity prior. Numerically, we support the proposed methods on synthetic data and simulated compressive sensing measurements (from actual FTI measurements) under various scenarios. In particular, the biological HS volumes considered in this work can be reconstructed with a three-to-ten-fold reduction in the light exposure.
\end{abstract}

\section{Introduction}\label{sec:intro} 
A HyperSpectral (HS) volume refers to a three-dimensional (3D) data cube associated with stacks of 2D spatial images along the spectral axis, \ie one monochromatic image for each wavenumber. Each spatial location records the variation of transmitted or reflected light intensity along a large number of spectral bands from that location in a scene.

Since chemical elements have unique spectral signatures, referred to as their fingerprints, observing the dense spectral content of an HS volume, instead of an RGB image, provides accurate information about the constituents of a scene. This has made HS imaging appealing in a wide range of applications, \eg agriculture~\cite{lacar2001use}, food processing~\cite{GOWEN2007590}, chemical imaging~\cite{farley2007chemical}, and biology~\cite{Leonard2015,Yudovsky2010,lu2014medical}.

Due to the extremely large amount of data stored in an HS volume, also referred to as HS cube, the acquisition and processing of these volumes encounters significant challenges. For example, in satellite or airborne applications it is common to transmit volumes with size of several gigabytes to a ground station. Additionally, in biological applications, designers of HS imagers face a myriad of trade-offs related to photon efficiency of the sensors, acquisition time, achievable spatiospectral resolution, and cost. 

Compared to other techniques that often reach lower spectral resolutions\footnote{See,~\eg \cite{Sellar:05} for a comprehensive classification of different HS acquisition methods.}, Fourier Transform Interferometry (FTI) is an HS imaging modality that has shown promising results in the acquisition of high spectral resolution biological HS volumes\footnote{They contain hundreds of spectral samples in the range of visible (400\,nm-700\,nm) and/or near-infrared (700\,nm-1000\,nm) wavelengths.}, specially in fluorescence spectroscopy. In this application a biological specimen is stained with dyes that fluoresce with a unique spectrum when they are exposed to light. Since each fluorescent dye binds to a certain protein, this technique is highly used for determining the localization, concentration, and growth of certain target cells (\eg malignant tumors), whose automatic characterizations are improved at higher spectral resolution.

As described in  Sec.~\ref{sec:fti}, FTI consists in a Michelson interferometer~\cite{michelson1887relative} with one moving mirror controlling the optical path difference of the light. Besides, the raw FTI measurements are intermediate data, \ie the Fourier coefficients of the actual HS volume with respect to the spectral domain. Shannon-Nyquist sampling theory states that for a fixed wavenumber range, the spectral resolution of an HS volume is proportional to the number of FTI measurements (see Sec.~\ref{sec:fti}). This fact is simultaneously a blessing and a curse in fluorescence spectroscopy, where high spectral resolution is highly desirable for biologists. One may record more FTI measurements, thereby over-exposing fluorescent dyes. As a consequence, the ability of the dyes to fluoresce fades during the experiments. This phenomenon, that is due to the photochemical alteration of the dyes, is called \textit{photo-bleaching}~\cite{Diaspro2006}. Therefore, the resolution of an HS volume is limited by the durability of the fluorescent dyes when exposed to the illumination. An alternative is to reduce light intensity in a fashion that is inversely proportional to the increase in the number of FTI measurements, resulting in a low input Signal-to-Noise Ratio (SNR). Considering all these facts, our main motivation in this paper is to reduce photo-bleaching in fluorescence spectroscopy using FTI. However, this requires delicate compromises between the goals of achieving high spectral resolution and high SNR while minimizing the light exposure imposed on the specimen.

As described in Sec.~\ref{sec:Sec_03} and~\ref{sec:Sec_04}, we propose two compressive sensing-FTI imagers, \ie \textit{Coded Illumination-FTI} (CI-FTI) and \textit{Structured Illumination-FTI} (SI-FTI), that are near-optimal in the sense of minimizing the light exposure on the biological specimen. In CI-FTI scheme the light source is randomly activated at few time intervals (or slots) during the motion of the FTI moving mirror, while in SI-FTI, at each time slot, the illumination is activated on a programmed set of spatial locations (\eg using a spatial light modulator). Both models amount to an incomplete set of FTI measurements that must be processed to recover the biological HS volume. While a common, but inaccurate, recovery procedure consists in applying the inverse Fourier transform on the (properly zero-padded) incomplete measurements~\cite{candes2006robust} --- with undesirable side-lobes artifacts limiting the usability of the reconstructed HS volume for biomedical purposes --- we resort here to the theory of Compressive Sensing (CS) introduced by Donoho~\cite{donoho2006compressed}, Cand\`es and Tao~\cite{candes2006near} for reconstructing the HS volume. In CS, the recovery of a signal from a small set of (random) linear measurements is ensured if the number of measurements is higher than the intrinsic dimension of the signal (\eg its sparsity level \cite{mallat2008wavelet,foucart2013mathematical}). 
A critical aspect of the adaptation of CI-FTI and SI-FTI to CS theory concerns the way FTI measurements are sub-sampled. Classical results of CS theory, \ie based on Uniform Density Sampling (UDS) of the signal frequency domain, are inefficient in the setting of FTI; the incompatibility between the Fourier sensing and the regularity of the HS volume (\eg in a wavelet domain) imposes limited sensing compression ratio. However, as most informative FTI measurements are concentrated around low frequencies \cite{moshtaghpour2016}, we here consider the sampling strategy as a degree of freedom in the system model; we leverage the Variable Density Sampling (VDS) procedure introduced by Krahmer and Ward~\cite{krahmer2014stable} to optimize the illumination coding strategy of CI-FTI and SI-FTI.

Although VDS-based methods are new in the field of HS imaging~\cite{roman2014asymptotic,studer2012compressive}, they have received considerable interest in Magnetic Resonance Imaging (MRI) applications where the recorded data is the 2D Fourier transform of an image. Our work differs in the sense that the practical FTI devices entail 3D signals whose 1D Fourier transform (with respect to spectral domain) is recorded. 

\subsection{Contributions}
The main objective of this paper is to establish a compressive sensing scheme for FTI, \ie reducing the total number of measurement in the optical path difference (OPD) domain compared to common (Nyquist rate) FTI and still allowing reliable and robust estimation of HS data volumes. We show that this compressive FTI enables HS imaging when total light exposure is a critical parameter, such as in biological confocal microscopy where special fluorescent dyes inserted in a biological specimen suffer from photo-bleaching. In fact, our work also studies this context when light exposure is considered as a fixed resource to be distributed equally over all measurements.        

More precisely, our work is driven by the following questions: \emph{(i)} For a fixed spectral resolution, to which extent can we reduce the amount of light exposure on the observed specimen and still allow robust reconstruction of HS data? \emph{(ii)} For a fixed exposure intensity budget and spectral resolution, what is the best illumination allocation per OPD sample (or time interval) and per pixel? These questions are answered through these two main contributions.
\medskip

\noindent \emph{(i)~Coded/Structured Illumination:} We propose two novel compressive sensing-FTI frameworks, based on coded and structured illumination techniques. The first system codes globally the light source, \ie whether the full biological specimen is globally highlighted or not per time slot, whilst in the second system we allow the illumination of each spatial location of the specimen to be independently coded over time, \eg thanks to a spatial light modulator \emph{structuring} the illumination of the specimen, before entering the FTI. The reconstruction of the HS volume resorts to the theory of CS, by leveraging an HS low-complexity prior, \ie spectral or joint spatiospectral sparsity models. By using a particular VDS scheme~\cite{krahmer2014stable} we derive uniform (\ie valid for all biological HS volumes) near-optimal reconstruction bounds in the sense of minimum amount of light exposure and its distribution for both CI-FTI and SI-FTI models. Table~\ref{tab:related works}, as described in Sec.~\ref{sec:related works}, summarizes our first contribution with respect to the state-of-the-art HS acquisition techniques.
\ \\[-2mm]

\noindent \emph{(ii)~Biologically-friendly constrained light exposure coding:} In fluorescence spectroscopy the tolerance of the fluorescent dyes, in terms of the light exposure budget, can be fixed by the biologists, \eg according to the specification of fluorescent dyes. In this case, the total light exposure is a fixed resource that must be consumed exactly over the compressive FTI measurements, while ensuring good reconstruction performance. We propose illumination strategies, both for CI-FTI and SI-FTI, that meet this context. This is done by adapting the intensity delivered by the light source according to the sample complexity --- the minimum number of measurements required for a successful HS data recovery --- of these schemes. Interestingly, in this scenario, we observe that full (Nyquist rate) sampling does not achieve the best HS reconstruction quality.
\ \\[-2mm]

\noindent \emph{(iii)~Noise level estimation in VDS:} In CS, the reconstruction quality of an HS volume, \eg from an $\ell_1$-minimization constrained by an $\ell_2$-fidelity term (see \eqref{eq:BPDN}), critically depends on an accurate bound on the $\ell_2$-norm of the noise.  In VDS, this fidelity term integrates a random weighting matrix --- with entries directly connected to the frequency sampling density --- whose presence prevents the use of common noise power estimator (such as a $\chi^2$-bound for Gaussian noise energy). As a complementary contribution, we propose in Sec.~\ref{sec:vds-noise} a new estimator that only 
depends on the unweighted $\ell_2$- and $\ell_\infty$-norms of the noise (Thm.~\ref{prop:gen-bound-noise-vds}). Moreover, this bound is instantiated to the case of an additive Gaussian noise in Cor.~\ref{cor:noise-level-estim-Gauss-noise}. 
\medskip

To the best of our knowledge, this paper provides the first theoretical analysis of \emph{(i)}~a compressive FTI scheme, \emph{(ii)}~an application of VDS of the Fourier coefficients outside the context of MRI (see Sec.~\ref{sec:related works}), and \emph{(iii)}~an estimation of noise power in VDS framework. In the sequel, beside the synthetic experiments, we validate our contribution by performing the proposed compressive sensing approaches on the practical FTI measurements recorded with respect to the Nyquist sampling rate.

\subsection{Related works}\label{sec:related works}

Hyperspectral imaging is an appealing area of research for the application of CS theory, in terms of acquisition~\cite{gehm2007single,wagadarikar2008single,arguello2011code,golbabaee2012joint,golbabaee2012hyperspectral,studer2012compressive,roman2014asymptotic,moshtaghpour2016,Woringer:17,moshtaghpourcoded,moshtaghpour2017a} and processing~\cite{zhang2011joint,li2012compressive,martin2015hyca,golbabaee2013compressive}. For instance, prior works~\cite{gehm2007single,wagadarikar2008single,arguello2011code} have proposed three types of Coded Aperture Snapshot Imager (CASSI) whose working principles are affiliated to our structured illumination scheme. The imager in~\cite{gehm2007single}, termed as Dual Disperser CASSI (DD CASSI), contains two dispersive elements and an aperture in between. The input scene is measured as a 2D multiplexed spectral projection. The corresponding sensing operator is equivalent to a cyclic S-matrix (\ie an alteration of the Hadamard matrix) and the solution of the system is estimated from under-determined measurements using the expectation maximization algorithm. A variant of CASSI, called Single Disperser CASSI (SD CASSI), that includes only one dispersive element is presented in~\cite{wagadarikar2008single}, wherein a 2D multiplexed spatiospectral projection of the scene is measured and the HS volume is recovered via an $\ell_{1}$ optimization problem. In~\cite{arguello2011code}, the idea of single-shot CASSI in~\cite{gehm2007single} and~\cite{wagadarikar2008single} is extended to an imager that collects multiple shot measurements. 

Besides CASSI systems, \cite{golbabaee2012joint} and~\cite{golbabaee2012hyperspectral} described a theoretical CS scheme for HS data acquisition based on random Gaussian projections. 
They provided a reconstruction method based on a convex optimization constrained by an $\ell_2$-fidelity term with respect to the noisy compressive measurements. The convex objective function in~\cite{golbabaee2012joint} penalizes both the trace norm and the (spatial) Total Variation (TV) norm of the HS volume; whereas in~\cite{golbabaee2012hyperspectral} the TV norm is replaced by an $\ell_{2,1}$ mixed-norm of the data matrix inferring a joint-sparse structure in the HS volume. 

Recently, the feasibility of the theory of CS in fluorescence spectroscopy has been demonstrated in~\cite{studer2012compressive,roman2014asymptotic,moshtaghpour2016}. In~\cite{studer2012compressive} and~\cite{roman2014asymptotic}, the proposed microscope directly scans the spectral dimension (as opposed to the interferometric methods of this paper), meaning that the resolution of the HS volume depends on the resolution of the photo-detector. Both works target the problem of acquiring high spatial resolution HS volume from compressive measurements using Hadamard sensing model, \eg by structuring the wide-field illumination~\cite{studer2012compressive} as in SI-FTI. The authors in~\cite{studer2012compressive} compared the uniform and half-half density samplings in their implementation while~\cite{roman2014asymptotic} employs a multilevel sampling strategy~\cite{adcock2013breaking,adcock2015quest} for the same microscope in~\cite{studer2012compressive} in addition to taking into account the photonic noise as well as the lens Point Spread Function~(PSF). 

Compressive sensing-FTI was introduced in~\cite{moshtaghpour2016} by sub-sampling the Optical Path Difference (OPD) dimension of the interferometric signals. Preliminary work by the same authors in \cite{moshtaghpour2016,moshtaghpourcoded,moshtaghpour2017a} showed that variable density sampling, when promoting low-frequencies, captures more informative FTI samples. Recently, a versatile scheme for fluorescence microscopy has appeared in the literature~\cite{Woringer:17}. It applies when the acquisition of the HS volume is done consecutively in the axial domain such that at every axial point the light source is modulated by a waveform. The CS measurements are essentially formed by recording deterministically the primitive part of axial domain.  A post-processing approach to increase the spectral resolution of an FTI, \ie a Sagnac interferometer, coupled with structured light source is demonstrated in~\cite{choi2014depth}. The objective of~\cite{choi2014depth} is to mitigate the physical limitation for achieving a high resolution HS volume; nonetheless this approach significantly increases photo-bleaching. The authors in \cite{jin2017hyperspectral} have recently advocated a compressive sensing FTI system based on single pixel imaging technique \cite{duarte2008single}. The structured illumination coding in \cite{jin2017hyperspectral} was later improved in \cite{moshtaghpour2018compressive}, using VDS scheme of \cite{krahmer2014stable}, and in \cite{moshtaghpour2018compressive}, using sparsity structure of fluorochrome spectra and VDS scheme of \cite{adcock2016note}. On the other hand, the works~\cite{zhang2011joint,li2012compressive,martin2015hyca,golbabaee2013compressive} tackle the problem of computing the abundance fraction of endmembers in the context of blind~\cite{zhang2011joint} or non-blind unmixing~\cite{martin2015hyca} and~\cite{golbabaee2013compressive}. 

Our work provides a clear, fivefold contribution compared to the aforementioned studies. First, we target the problem of compressive HS volume acquisition; second, the physical model of the FTI device imposes the Fourier sensing operator (since it captures interferometric information) that brings about challenges different from ideal Gaussian sensing operator; third, we study the trade-off between the spectral resolution of the HS volume and photo-bleaching deterioration of the fluorescent dyes; forth, we apply VDS scheme to the field of interferometric HS imaging\footnote{Parallel to this study, we have presented an improved coding scheme for CI-FTI in \cite{moshtaghpour2018multilevel} by leveraging the sparsity structure of the fluorochrome spectra.}; and fifth, unlike single pixel FTI models \cite{jin2017hyperspectral,moshtaghpour2018compressive,moshtaghpour2018compressive2}, the proposed compressive FTI systems are based on 2D imaging sensors.

Let us finally mention that prior studies on variable density sampling in CS focus on specific signals, \eg 1D signals~\cite{puy2011variable}, 2D spatial images~\cite{adcock2013breaking,krahmer2014stable,boyer2015compressed,bigot2016analysis,candes2007sparsity,wang2010variable}. A VDS is computed in~\cite{puy2011variable} through a convex optimization problem whose performance depends on the parameter tuning. In \cite{wang2010variable}, a VDS scheme is adapted to statistical models of natural images. Based on the RIP-less analysis in~\cite{candes2011probabilistic}, a block sampling CS model is introduced in~\cite{bigot2016analysis,boyer2015compressed}, whose stable and robust recovery guarantee was later proved in~\cite{adcock2018oracle}. A novel multi-level sampling scheme is proposed in~\cite{adcock2013breaking,roman2014asymptotic} in order to exploit the level-sparsity feature of some natural images, \eg in MRI. A RIP-based guarantee is presented in~\cite{krahmer2014stable} based on the notion of \textit{random sampling in bounded orthonormal systems} introduced in \cite{rauhut2010compressive,foucart2013mathematical}. The sampling density in \cite{krahmer2014stable} is controlled by the bound of local coherence between sensing and sparsity bases. Besides, a number of empirical VDS methods exist in the literature including~\cite{lustig2008compressed,moshtaghpour2016,studer2012compressive}. An example of a 3D HS acquisition is considered in~\cite{roman2014asymptotic} using filter banks for scanning spectral domain. Our work explores another VDS-based acquisition; we study compressive FTI methods where the spectral information is multiplexed through a Fourier (interferometric) transform (see Sec.~\ref{sec:fti}).

\begin{table}[t]
  {\footnotesize
    \begin{center}
      \noindent
      \scalebox{0.85}{\begin{tabular}{| l | m{9mm} | m{9mm} | c| c |c|c|c|c|c|c|}
                        \hline
                        & \multicolumn{2}{l|}{\textbf{Spectrometry}} & \multicolumn{3}{l|}{\textbf{Study type}} & \multicolumn{3}{l|}{\textbf{Prior model}} & \multicolumn{2}{l|}{\textbf{Sampling}}\\
                        \cline{2-11}
                        & \centering \rotatebox[origin=c]{90}{~Interferometric~} & \centering\rotatebox[origin=c]{90}{Direct} & \rotatebox[origin=c]{90}{Hardware} & \rotatebox[origin=c]{90}{Numerical} & \rotatebox[origin=c]{90}{Analytical} & \rotatebox[origin=c]{90}{3D low-rank} & \rotatebox[origin=c]{90}{~Spectral sparsity~} & \rotatebox[origin=c]{90}{Spatial sparsity} & \rotatebox[origin=c]{90}{Operator}& \rotatebox[origin=c]{90}{Strategy}\\
                        \hline
                        M. Gehm \etal~\cite{gehm2007single} &  & \centering $\checkmark$ &  $\checkmark$ &  &  &  & DW  & DW & RSE & UDS\\
                        \hline
                        A. Wagadarikar \etal~\cite{wagadarikar2008single} &  & \centering $\checkmark$ & $\checkmark$ &  &  &  & Dirac & DW & RSE & UDS\\
                        \hline
                        H. Arguello \etal~\cite{arguello2011code} &  & \centering $\checkmark$ &  $\checkmark$ &  &  &  & DC & DW & RSE & UDS\\
                        \hline
                        M. Golbabaee \etal~\cite{golbabaee2012joint} &  & \centering $\checkmark$ &  &  $\checkmark$ &  & $\checkmark$ & Dirac & TV & RGM &  --- \\
                        \hline
                        M. Golbabaee \etal~\cite{golbabaee2012hyperspectral} &  & \centering $\checkmark$ &  & $\checkmark$ &  &  $\checkmark$ & Dirac & DW & RGM &  --- \\
                        \hline
                        V. Studer \etal~\cite{studer2012compressive} &  & \centering $\checkmark$ & $\checkmark$ &  &  &  & Dirac & DW & RHE & Half-half\\
                        \hline
                        B. Roman \etal~\cite{roman2014asymptotic} &  & \centering $\checkmark$ &  & $\checkmark$ &  &  & Dirac & DW & RHE & Multi-level\\
                        \hline
                        A. Moshtaghpour \etal~\cite{moshtaghpour2016} & \centering $\checkmark$ &  &  & $\checkmark$ &  &  & DW & DW & RFE & VDS\\
                        \hline
                        M. Woringer \etal~\cite{Woringer:17} &  & \centering $\checkmark$ & $\checkmark$ &  $\checkmark$ &  &  & Dirac & Dirac & LFE & Deterministic\\
                        \hline
                        \rowcolor{mycolor1}
                        CI-FTI & \centering $\checkmark$ &  &  & $\checkmark$ & $\checkmark$ &  & Dirac & DW & RFE & VDS\\
                        \hline
                        \rowcolor{mycolor1}
                        SI-FTI & \centering $\checkmark$ &  &  & $\checkmark$ & $\checkmark$ &  & DW & DW & RFE & VDS\\
                        \hline\hline
                        \multicolumn{5}{|l||}{RSE: Random S-matrix Ensembles} & \multicolumn{6}{l|}{RGM: Random Gaussian Mixtures}
                        \\
                        \multicolumn{5}{|l||}{RHE: Random Hadamard Ensembles} & \multicolumn{6}{l|}{RFE: Random Fourier Ensembles}
                        \\
                        \multicolumn{5}{|l||}{LFE: Low-pass Fourier Ensembles} & \multicolumn{6}{l|}{TV: Total Variation}
                        \\
                        \multicolumn{5}{|l||}{DC: Discrete Cosine} & \multicolumn{6}{l|}{DW: Discrete Wavelet}
                        \\
                        \multicolumn{5}{|l||}{UDS: Uniform Density Sampling} & \multicolumn{6}{l|}{VDS: Variable Density Sampling}
                        \\
                        \hline
                      \end{tabular}
                   }
                  \end{center}
               }
                \caption{Comparison of the related acquisition modalities in HS imaging. In direct spectrometry the spectral dimension is recorded by either wavelength filtering, grating, or prism. Half-half sampling strategy means half of the $M$ measurements is taken from the lowest frequency components and the other half is taken by randomly sub-sampling the rest of the remaining samples. Note that compressive FTI systems designed based on single pixel imaging technique, \eg \cite{jin2017hyperspectral,moshtaghpour2018compressive,moshtaghpour2018compressive2}, are not reported in this table.}
                \label{tab:related works}
              \end{table}
              
\subsection{Paper Organization}
This paper is organized as follows. We start by providing preliminary introduction to the theory of CS, with an emphasis on sensing with random ensembles of orthonormal systems, and the acquisition model of a conventional FTI imager. The problem formulation and theoretical analysis of CI-FTI and SI-FTI frameworks are presented in Sec.~\ref{sec:Sec_03} and Sec.~\ref{sec:Sec_04}, respectively. The concept of constrained light exposure budget for both CI-FTI and SI-FTI is revealed in Sec.~\ref{sec:Sec_05}. Sec.~\ref{sec:Sec_06} focuses on a comprehensive synthetic and experimental tests of our schemes demonstrating the  efficiency of the proposed methods when the fluorochromes in a biological specimen undergone photo-bleaching. Finally, we conclude with some remarks and perspectives in Sec.~\ref{sec:Sec_07}.
\subsection{Notations}
\label{sec:notations}

Domain dimensions are represented by capital letters, \eg $K,M,N$. Vectors, matrices, and data cubes are denoted by bold symbols. For a cube $\bs{\mathcal{U}} \in \bb{C}^{N_3 \times N_1 \times N_2}$,  $\bs{U} = (\bs{u}_1, \cdots, \bs{u}_{N_1N_2}) \in \bb{C}^{N_3 \times N_1N_2}$ corresponds to the unfolded matrix representation of $\bs{\cl U}$, \ie compacting its two last dimensions, while the vectorization of $\bs{\cl U}$ reads $\bs{u} = {\rm vec}(\bs{U}) := (\bs{u}_1^\top, \cdots, \bs{u}_{N_1N_2}^\top)^\top \in \bb{C}^{N_1N_2N_3}$. When this is clear from the context, we assimilate 3D data in $\bb{C}^{N_1\times N_2 \times N_3}$ with their vector and matrix representations, \eg identifying $\bs{\cl U}$ with $\bs{U}$. For any matrix (or vector) $\bs V \in \bb C^{M\times N}$, $\bs V^\top$ and $\bs V^*$ represent the transposed and the conjugate transpose of $\bs V$, respectively, and $\bs{V} \otimes \bs{W}$ denotes the Kronecker product of two matrices $\bs V$ and $\bs W$. The $\ell_p$-norm of $\bs{u}$ reads $\|\bs{u}\|_p := (\sum_i |u_i|^p)^{1/p}$, for $p\geq 1$, with $\|\bs u\| :=\|\bs u\|_2$. The identity matrix of dimension $N$ is represented as $\bs{I}_N$. The set of indices ranging between 1 and $N$ is denoted $\range{N}:=\{1, \cdots, N\}$. By an abuse of convention, except if expressed differently, we consider that a set (or a subset) of indices is actually a \emph{multiset}, \ie we allow for multiple instances for each of its elements; the set cardinality thus considers the total number of (non-unique) multiset elements. For a subset $\Omega \subset \range{N}:=\{1, \cdots, N\}$ of cardinality $|\Omega|$, the restriction of a vector $\bs{u} \in \bb{C}^N$ (or a matrix $\bs{\Phi}\in \bb{C}^{M\times N}$) to the components (or the columns) indexed in $\Omega$ is denoted by $\bs{u}_\Omega$ (or $\bs{\Phi}_\Omega$). We always consider that the restriction operator has a higher precedence than the adjoint or transposition operators, \eg $\bs{\Phi}^*_\Omega:=(\bs{\Phi}_\Omega)^*$. Finally, given a random event $\cl E$, $\bb{P}(\cl E)$ denotes the probability that it occurs, and we use the asymptotic relations $f \lesssim g$ (or $f \gtrsim g$), also denoted by $f = O(g)$ (resp. $f = \Omega(g)$), if $f \le c\,g$ (resp. $g \le c\, f$) for two functions $f$ and $g$ and some value $c>0$ independent of their parameters.

\section{Preliminaries on Compressive Sensing}             
\label{sec:Sec_02}

This section provides the main tools from CS theory and its recent extension to variable density sampling of an orthogonal sensing basis (\eg Fourier). 

\subsection{From Uniform to Variable Density Sampling}
\label{sec:compr-uds-vds}

When restricted to signal sensing with random orthonormal basis ensembles~\cite{donoho2006compressed,candes2006near}  (\eg random Fourier or Hadamard ensembles), CS theory targets the recovery of a ``low-complexity'' signal $\bs x \in \bb C^N$ from a vector of noisy measurements $\bs{y} = \bs{\Phi}^*_\Omega \bs{x} + \bs n$, where $\bs{\Phi} \in \bb{C}^{N \times N}$ is an orthonormal sensing system and $\Omega \subset \range{N}$ is a set of indices chosen at random with $|\Omega| =M \ll N$, and $\bs n$ accounts for some additive observation noise. The low-complexity nature of $\bs x$ generally amounts to assuming it $K$-sparse, \ie with $\|\bs x\|_0 := |{\rm supp}(\bs{x})|\le K$, or at least well approximated by a sparse signal, \ie compressible.

In order to estimate $\bs x$ from $\bs y$, the \textit{restricted isometry property} (RIP)~\cite{candes2005decoding} is a sufficient condition on the matrix $\bs{\Phi}^*_\Omega$ for most classes of algorithms, \eg convex optimization~\cite{candes2005decoding}, greedy \cite{Zhangopm}, and thresholding \cite{blumensath2009iterative} strategies.
\begin{definition}\label{def:RIP}
  Given $K \le N$, the restricted isometry constant $\delta_K$ associated with $\bs{A} \in \bb{C}^{M \times N}$ is the smallest number $\delta$ for which 
  \begin{equation}
    (1-\delta)\|\bs{x}\|^2 \leq \|\bs{A} \bs{x}\|^2 \le (1+\delta)\|\bs{x}\|^2\label{eq:RIP}
  \end{equation}
  holds for all $K$-sparse vectors $\bs{x}\in \bb{C}^N$. Alternatively, if~\eqref{eq:RIP} holds for  $\delta=\delta_K$, we say that  $\bs{A}$ satisfies the RIP of order $K$ and constant~$\delta_K$. 
\end{definition}
When the RIP holds and the observation noise is bounded, \ie $\|\bs n\| \leq \varepsilon$ for some $\varepsilon > 0$, the Basis Pursuit DeNoise (BPDN) program expressed as 
\begin{equation}
  \hat{\bs{x}}\ =\ \argmin_{\bs{u} \in \bb{C}^{N}} \|\bs{u}\|_1\ \st\ \| \bs{y}-\bs{A} \bs{u}\| \le \varepsilon, \label{eq:BPDN old}
\end{equation}
provides an accurate estimate of the original signal. 
\begin{proposition}[Prop. 4.2~\cite{krahmer2014stable}]\label{prop:rip old}
  Assume that the restricted isometry constant $\delta_{5K}$ of $\bs{A} \in \bb{C}^{M \times N}$ satisfies\footnote{There exist alternative versions of this condition on restricted isometry constant, \eg in \cite{cai2014sparse} and \cite[Theorem 5]{foucart2016flavors}, that provide sharper results.} $\delta_{5K} < \frac{1}{3}$. Then, for all $\bs{x} \in \bb{C}^N$ observed through the noisy CS model $\bs{y} = \bs{A} \bs{x}+\bs{n}$ with  $\|\bs{n}\| \le\varepsilon$, the solution $\hat{\bs x}$ of~\eqref{eq:BPDN old} satisfies 
  \begin{equation*}
    \|\bs{x}-\hat{\bs{x}}\| \le \textstyle \frac{2\sigma_K(\bs{x})_1}{\sqrt{K}} + \varepsilon,
  \end{equation*}
  where $\sigma_K(\bs{x})_1 :=\|\bs{x}-\bb H_K(\bs{x})\|_1$ is the best $K$-term approximation error (in the $\ell_1$ sense), and $\bb H_K$ is the hard thresholding operator that maps all but the $K$ largest-magnitude entries of the argument to zero. In particular, the reconstruction is exact, \ie $\hat{\bs x}=\bs x$, if $\bs x$ is $K$-sparse and $\varepsilon$ = 0.
\end{proposition}

More generally, if the signal $\bs x$ has a sparse or compressible representation in a general orthonormal basis $\bs \Psi \in \bb{C}^{N \times N}$ (\eg in a wavelet basis), \ie $\bs x = \bs \Psi \bs s$ where the vector of coefficients $\bs s \in \bb{C}^{N}$ is $K$-sparse or compressible, one must ensure that $\bs{\Phi}^*_\Omega\bs{\Psi}$ respects the RIP. 

Interestingly, according to~\cite{rauhut2010compressive}, choosing each of the $M$ elements of $\Omega$ uniformly at random in $\range{N}$, \ie according to a Uniform Density Sampling (UDS) gives that, with probability exceeding $1-N^{-c\log^3(K)}$, the restricted isometry constant of the matrix $\bs A = \sqrt{\frac{N}{M}} \bs{\Phi}^*_\Omega\bs{\Psi}$ satisfies $\delta_K \le\delta$, if 
\begin{equation}
  M \gtrsim \delta^{-2}\mu^2 K \log^3(K)\log(N),\label{eq:rip_for_onb}
\end{equation} 
where $\mu/ \sqrt{N}$ is an upper bound on the \textit{mutual coherence}\footnote{In some references, \eg~\cite{roman2014asymptotic}, it is defined as $\mu^{\rm mut}(\bs{\Phi},\bs{\Psi}) :=\max_{\nell,j} |\langle \bs{\phi}_\nell,\bs{\psi}_j\rangle|$.} of $\bs{\Phi}^*\bs{\Psi}$, \ie 
$$
\mu^{\rm mut}(\bs{\Phi}^*\bs{\Psi}) := \max_{\nell,j}|(\bs{\Phi}^*\bs{\Psi})_{\nell,j} |\le \mu /\sqrt{N}.
$$ 

However, the impact of this result is unfortunately limited in the case where, for instance, $\bs{\Phi}$ and $\bs{\Psi}$ are the discrete Fourier and wavelet bases, respectively. This is in fact an important special combination appearing in, \eg Magnetic Resonance Imaging~\cite{adcock2016note}, radio-interferometry~\cite{wiaux2009compressed} and in the compressive FTI scheme considered in this paper (see Sec.~\ref{sec:Sec_03} and Sec.~\ref{sec:Sec_04}). The coherence in this case is $\mu^{\rm mut}(\bs{\Phi}^*\bs{\Psi}) = 1$ and~\eqref{eq:rip_for_onb} states that all samples are required for reconstruction ($M \gtrsim N$), although the signal can be highly sparse in the wavelet basis. 

Fortunately, this limitation is actually induced by the way $\Omega$ is built, \ie according to a UDS of the columns of $\bs \Phi$. Adopting a Variable Density Sampling (VDS) scheme allows us to break this ``coherence barrier''~\cite{puy2011variable,adcock2013breaking,roman2014asymptotic,wang2010variable, krahmer2014stable}. Of interest for our analysis,~\cite{krahmer2014stable} shows in particular that the number of measurements in~\eqref{eq:rip_for_onb} can be reduced by assigning higher sampling probability to the columns of $\bs{\Phi}$ that are highly coherent with the columns the of sparsity basis $\bs{\Psi}$. This probability is determined by the \emph{local coherence} of the sensing basis over the sparsity basis, \ie the quantity $\mu^{{\rm loc}}_\nell(\bs \Phi, \bs \Psi) := \mu^{{\rm loc}}_\nell(\bs \Phi^* \bs \Psi)$ 
with 

\begin{equation}
  \label{eq:loc-coh}
  \mu^{{\rm loc}}_\nell(\bs{A}) := \max_{1 \le j\le N} |a_{\nell,j}|,\quad \bs A \in \bb C^{N \times N}. 
\end{equation}

\begin{proposition}[RIP for VDS~{\cite[Thm. 5.2]{krahmer2014stable}}]\label{prop:variable_density_sampling}
  Let $\bs{\Phi} \in \bb{C}^{N \times N}$ and $\bs{\Psi} \in \bb{C}^{N \times N}$ be orthonormal sensing and sparsity bases, respectively, with $\mu^{{\rm loc}}_\nell(\bs \Phi, \bs \Psi) \le \kappa_\nell$ for some values $\kappa_\nell \in \bb R_+$. Let us define $\bs \kappa := (\kappa_1,\cdots, \kappa_N)^\top$. Suppose $K \gtrsim \log(N)$,
  \begin{equation}
    M \gtrsim \delta^{-2}\|\bs{\kappa}\|^2 K \log^3(K) \log(N),\label{eq:rip_for_onb_vds}
  \end{equation}
  and choose $M$ (possibly not distinct) indices $\nell \in \Omega \subset \range{N}$ \iid with respect to the probability distribution $p$ on $\range{N} $ given by 
  \begin{equation}
 \ts p(\nell) := \frac{\kappa_\nell^2}{\|\bs{\kappa}\|^2}.\label{eq:vds pmf}
  \end{equation}
  Consider the diagonal matrix $\bs{D}={\rm diag}(\bs{d}) \in \bb{R}^{M \times M}$ with $d_j = {1}/{\sqrt{p(\Omega_j)}}$, $j \in \range{M}$. Then with probability exceeding $1-N^{-c\log^3(K)}$, the restricted isometry constant $\delta_K$ of the preconditioned matrix $\frac{1}{\sqrt{M}}\bs{D}\bs{\Phi}^*_\Omega\bs\Psi$ satisfies $\delta_K \le\delta$.
\end{proposition}

We will see later that this VDS offers new means for compressive FTI. Note that, in practice, although $\bs\Phi^*_\Omega$ models the analog/optical sensing procedure of a growing number of CS applications (\eg in compressive MRI or radio-interferometry), the conditioning by $\bs{D}$ proposed in Prop.~\ref{prop:variable_density_sampling} is rather achieved by post-processing (digitally) the acquired data~\cite{krahmer2014stable}. Therefore, our estimation of the HS data will rather follow this straightforward adaptation of Prop.~\ref{prop:rip old}, also tuned to an analysis-based sparsity framework~\cite{elad2007analysis} that better suits the rest of our developments.
\\
\begin{proposition}\label{prop:RIP}
  Assume that the restricted isometry constant $\delta_{5K}$ of $\frac{1}{\sqrt{M}}\bs{D}\bs{\Phi}^*_\Omega \bs{\Psi}$ $\in \bb{C}^{M \times N}$ satisfies $\delta_{5K} < \frac{1}{3}$. Let $\bs{x} \in \bb{C}^N$ be a signal observed by the noisy sensing model $\bs{y} = \bs{\Phi}^*_\Omega \bs{x}+\bs{n}$ with $\|\frac{1}{\sqrt{M}}\bs{D} \bs{n}\| \le\varepsilon$. Then
  \begin{equation}
    \ts \hat{\bs{x}} = \argmin_{\bs{u} \in \bb{C}^{N}} \|\bs{\Psi}^*\bs{u}\|_1\ \st \ \|\frac{1}{\sqrt{M}}\bs{D}( \bs{y}-\bs{\Phi}^*_\Omega \bs{u})\| \le \varepsilon, \label{eq:BPDN}
  \end{equation}
  satisfies 
  \begin{equation}
    \label{eq:l2-l1-inst-optimality}
    \ts \|\bs{x}-\hat{\bs{x}}\| \le \frac{2\sigma_K(\bs{\Psi}^*\bs{x})_1}{\sqrt{K}} + \varepsilon.
  \end{equation}
\end{proposition}

\begin{proof}
  The proof involves three steps: \emph{(i)} in Prop.~\ref{prop:rip old} let $\bs A = \bs{\Phi}^*_\Omega \bs \Psi$; \emph{(ii)} precondition the  matrix $\bs A$ according to Prop.~\ref{prop:variable_density_sampling} so that it satisfies RIP; \emph{(iii)} apply a change of variable $\bs u \rightarrow \bs \Psi^* \bs u$ in \eqref{eq:BPDN old} using the fact that $\bs \Psi$ is an orthonormal basis.
\end{proof}

In Sec.~\ref{sec:Sec_04} and Sec.~\ref{sec:Sec_05}, we will leverage Prop.~\ref{prop:variable_density_sampling} and Prop.~\ref{prop:RIP} in order to characterize the sample complexities and the stability of our two compressive FTI strategies, namely coded illumination FTI (Thm~\ref{thm:cifti}) and structured illumination FTI (Thm~\ref{thm:sifti}).
\medskip

\subsection{Noise level estimation in Variable Density Sampling}
\label{sec:vds-noise}

We study now an important aspect of the VDS scheme that seems not covered in the literature: estimating the noise level $\varepsilon$ in \eqref{eq:BPDN} by integrating the influence of the weighting (random) matrix $\bs D$ introduced in Prop.~\ref{prop:variable_density_sampling}. This estimation is indeed critical for reaching robust reconstruction of HS volumes from the two compressive FTI schemes studied in Sec.~\ref{sec:Sec_04} and Sec.~\ref{sec:Sec_05}. The content of this section is voluntarily written in general notations, as it could also be of general interest for any VDS-based compressive sensing application corrupted by some additive measurement noise.   

The next theorem bounds with controlled probability the (weighted) $\ell_2$-norm of any vector $\bs n \in \bb C^M$, \eg the fixed realization of a noise vector corrupting the sensing model $\bs y = \bs \Phi_{\Omega}^* \bs x + \bs n$, when this norm is weighted by the random diagonal matrix $\bs D$ associated with $\Omega$. The bound only depends on the unweighted $\ell_2$- and $\ell_\infty$-norms of $\bs n$, and on a quantity $\rho>1$ fixed by the density defining the VDS.     

\begin{thm}
\label{prop:gen-bound-noise-vds}
Given two integers $M<N$, let us consider a discrete random variable (\rv) $\beta \in \range{N}$ associated with the probability mass function (pmf) $\eta(\nell):= \bb P[\beta = \nell]$ for $l\in [N]$, and assume there exists a $\rho \in [1, \infty)$ such that
\begin{equation}
  \label{eq:pmf-moment}
\ts \tinv{N}\,\sup_{q\geq 1} \frac{1}{q} (\bb E_\beta\,\eta(\beta)^{-q})^{1/q} \leq\ \rho.  
\end{equation}
We define a random index set $\Omega = \{\Omega_j\}_{j=1}^M \subset \range{N}$ made of $M$ (possibly non-distinct) indices $\Omega_j \sim_{\iid} \beta$, and a random diagonal matrix $\bs D \in \bb R^{M\times M}$ such that $D_{jj} = 1/\sqrt{\eta(\Omega_j)}$.~Given $s>0$ and $\bs n \in \bb C^M$, we have  
\begin{align}
  \label{eq:gen-bound-noise-vds-ext-noise}
\ts \frac{1}{M}\,\|\bs D \bs n\|^2&\ts \leq \frac{N}{M}\|\bs n\|^2 + 4 e \max\big\{\frac{s}{M}, \frac{\sqrt s}{\sqrt M}\big\} \rho\,\|\bs n\|_\infty^{2} N,
\end{align} 
with probability exceeding $1-e^{-s/2}$ (\eg for $s=6$, this probability exceeds $0.95$).
\end{thm}
For the proof see App.~\ref{sec:proof_prop_gen-bound-noise-vds}. The fact that $\rho \geq 1$ is a simple consequence of $\bb E (\eta(\beta))^{-1} = \sum_l 1 = N$.

\begin{remark}
\label{remark:simple_bound_for_rho}
Up to the normalization in $1/N$, the parameter $\rho$ in \eqref{eq:pmf-moment} is actually a bound on the sub-exponential norm $\|1/\eta(\beta)\|_{\psi_1} := \sup_{q\geq 1} \frac{1}{q} (\bb E_\beta\,\eta(\beta)^{-q})^{1/q}$ of the \rv $1/\eta(\beta)$ \cite{V2012,vershynin_2018}. From this definition, and since $\|X\|_{\psi_1} \leq L$ for any bounded \rv $X$ with $|X| \leq L$ and $L>0$, we can directly identify that in the case of UDS, $\eta(l) = 1/N$ implies $\|1/\eta(\beta)\|_{\psi_1} \leq N$, \ie we can set $\rho = 1$. 

More generally, for any VDS scheme, \eg see Prop.~\ref{prop:cifti local coherence}, where $\eta(l) = c_\eta\,\min\{1,|l-l_0|^{-\alpha}\}$ for some exponent $\alpha > 0$, an offset parameter $l_0 \in \range{N}$ centering $\eta$ on $l_0$, and $c_\eta>0$, we find 
$$
\ts \|1/\eta(\beta)\|_{\psi_1} = c_\eta^{-1} \|\max\big\{1,|\beta-l_0|^\alpha\big\}\|_{\psi_1} \leq c_\eta^{-1} (N - l_0)^\alpha,
$$
which provides $\rho = c_\eta^{-1} |N-l_0|^\alpha/N$. In Sec.~\ref{sec:Sec_03} and Sec.~\ref{sec:Sec_04}, the parameter $l_0$ denotes the index of the OPD origin. In particular, for $\alpha = 1$, $\rho$ mainly depends on the pmf normalization constant $c_\eta$, \ie $c_\eta^{-1} \simeq \log N$, as computed in Prop.~\ref{prop:cifti local coherence}.

In conclusion, for $\alpha = 0$ (UDS) and for $\alpha = 1$, we can expect that $\rho$ is either constant or that it grows slowly (logarithmically) when $N$ increases, hence ensuring a good control over the bound \eqref{eq:gen-bound-noise-vds-ext-noise}. 
\end{remark}
The previous theorem is expected to provide useful bounds if the leading terms in the right-hand side of \eqref{eq:gen-bound-noise-vds-ext-noise} is $\frac{N}{M} \|\bs n\|^2$, \ie if $\rho \|\bs n\|_\infty^2 N \lesssim \frac{N}{M} \|\bs n\|^2$. This happens for $\rho$ slowly growing when $N$ increases, \eg $\rho = O(\log N)$, (see Remark~\ref{remark:simple_bound_for_rho}) and for any vector $\bs n$ that is not too sparse, \ie such that $\|\bs n\|^2_\infty \lesssim \|\bs n\|^2/M$. As shown in the next corollary, a Gaussian random noise respects this requirement with high probability\footnote{While Cor.~\ref{cor:noise-level-estim-Gauss-noise} can be easily extended to the complex field, we restrict it anyway to a real Gaussian noise, \ie the noise of interest in the next sections.}. 
\begin{cor}
\label{cor:noise-level-estim-Gauss-noise}
In the context of Thm.~\ref{prop:gen-bound-noise-vds}, let us consider the Gaussian random vector $\bs n \in \bb R^M$ with $n_k \sim_{\iid} \cl N(0,\sigma^2)$, $k\in\range{M}$. Given some $s>0$, 
in the UDS case, \ie $\eta(\nell) = 1/N$ and $\bs D^2 = N \bs I_N$, we have 
\begin{align}
\label{eq:noise-level-estim-Gauss-noise-ext-noise-uds}
\ts \frac{1}{\sqrt M}\,\|\bs D \bs n\|\ &\ts \leq\ \varepsilon^0_{\sigma,s}(N,M) := \sigma \sqrt N\, (1 + \frac{1}{\sqrt 2}\frac{\sqrt{s}}{\sqrt{M}} + \frac{s}{M})^{1/2},
\end{align}
with probability exceeding $1-\exp(-\frac{s}{2})$. More generally, in a VDS context with arbitrary pmf $\eta$, 
\begin{align}
\label{eq:noise-level-estim-Gauss-noise-ext-noise}
\ts \frac{1}{\sqrt M}\,\|\bs D \bs n\|\ &\leq\ \varepsilon_{\sigma,s}(N,M,\rho)\\
&\ts := \sigma \sqrt N\,\big[ (1 + \frac{1}{\sqrt 2}\frac{\sqrt{s}}{\sqrt{M}} + \frac{s}{M}) + 4e (2\log M +s) \max\big\{\frac{s}{M}, \frac{\sqrt s}{\sqrt M}\big\} \rho \big]^{1/2},\nonumber
\end{align}
with probability exceeding $1-3\exp(-\frac{s}{2})$ (\eg greater than $0.95$ for $s=8.2$). 
\end{cor}
We postpone the proof of this corollary to App.~\ref{sec:proof_cor_noise-level-estim-Gauss-noise}.\medskip

  \begin{remark} By comparing \eqref{eq:noise-level-estim-Gauss-noise-ext-noise} to \eqref{eq:noise-level-estim-Gauss-noise-ext-noise-uds}, we observe that the variability of $\bs D$ induces a bias behaving like $O( \rho \log M \max\{\frac{1}{M}, \frac{\sqrt 1}{\sqrt M}\})$ in $\varepsilon_{\sigma,s}(N,M,\rho)$ compared to the simpler bound $\varepsilon^0_{\sigma,s}$ reached by UDS. 
Therefore, if $\rho$ is slowly growing when $N$ increases (see Rem.~\ref{remark:simple_bound_for_rho}), and for $M$ large, this bias is small and the two noise levels, for the VDS and UDS schemes, are thus comparable.   
\end{remark}
The previous remark will be used in conjunction with Cor.~\ref{cor:noise-level-estim-Gauss-noise} in Sec.~\ref{sec:cifti_HS Reconstruction Method and Guarantee} and Sec.~\ref{sec:sifti reconstruction} in order to bound the noise in our compressive FTI under a Gaussian noise assumption.

\section{Acquisition Model in Conventional FTI}
\label{sec:fti}

This section describes the principles of FTI acquisition, its continuous observation model and its crucial discretization into a forward model that can be inverted numerically. We also discuss the main noise sources  that corrupt the FTI.  

\subsection{Continuous Observation Model}
\label{sec:FTI continuous model}

	In a nutshell, the operating principle of a conventional FTI, called hereafter \textit{Nyquist FTI}, is based on Michelson interferometry~\cite{michelson1887relative} (see, \eg~\cite{bell2012introductory} for a detailed survey on FTI). As shown in Fig.~\ref{fig:FTI_scheme}~(top), a parallel beam of coherent light obtained from the 2D image of a thin, still, and almost transparent biological specimen magnified by a (confocal) microscope (not represented here), is divided into two beams by a Beam-Splitter (BS), \eg made of two birefringent prisms. The two beams are reflected back by a fixed and a moving mirror, and interfere after being recombined by the BS. The intensity of the resulting beam is next acquired by a standard imaging sensor (camera); each image captured by this camera corresponds to a given position of the moving mirror, and each pixel records temporally a specific interference pattern. 

\begin{figure}[!t]
  \centering
  \includegraphics[width=0.9\textwidth]{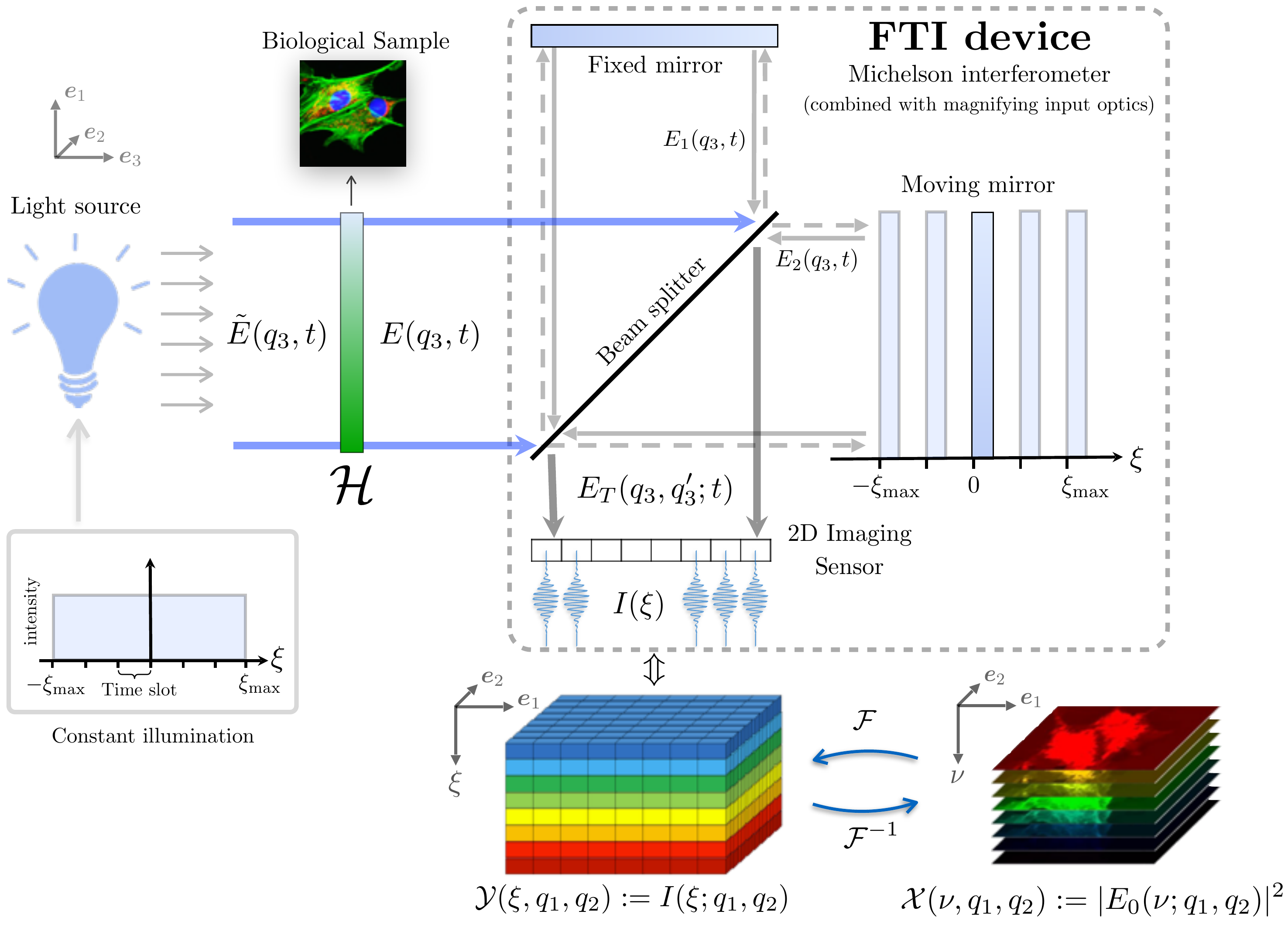}
  \caption{Operating principle of Fourier transform interferometry.}
  \label{fig:FTI_scheme}
  \vspace{-5mm}
\end{figure}

In order to understand how an HS volume of the observed biological specimen can be deduced from the recorded interferences, let us develop an idealized optical model of this FTI. We simplify our description by omitting the magnifying optics that are required to observe the spatial details of the specimen\footnote{Such an optical setup (\eg a confocal microscope) can be easily inserted in our scheme by the introduction of a scaling factor between the coordinate systems of the biological specimen plane and the 2D imaging sensor plane.}. Given an orthonormal coordinate system $\{\bs e_1, \bs e_2, \bs e_3\} \subset \bb R^3$, a coherent wide-band plane wave emitted by the light source and traveling along $\bs e_3$-direction may be formulated at some point $\bs q := (q_1,q_2,q_3)^\top \in \bb R^3$ and time $t$ as 
\begin{equation}
  \label{eq:wide-band-wave}
  \ts \tilde{E}(\bs q;t) = \int_{0}^{+\infty}\tilde{E}_0(q_1, q_2;\nu, t) e^{i(2\pi q_3\nu-\omega t)}\,\ud\nu,  
\end{equation}
where $\nu$ is the wavenumber\footnote{Recall that in the application of FTI in fluorescence spectroscopy the light frequency $\nu {\sf c}_{\rm v} \in $ [400 THz, 770 THz], with ${\sf c}_{\rm v}$ the speed of light in the vacuum.}, $\omega = 2\pi \nu {\sf c}_{\rm v}$ is the angular frequency (with ${\sf c}_{\rm v}$ the light velocity in the medium) and $\tilde{E}_0(q_1,q_2;\nu, t)$ represents the amplitude of light at~$(q_1,q_2;\nu)$; see, \eg \cite{young2007university} for the principles of electromagnetic wave propagation.

In this work, we consider that the illumination ensures that $\tilde{E}_0(q_1,q_2;\nu, t)$ is constant in time and with respect to~$(q_1,q_2)$, \ie $\tilde{E}_0(q_1,q_2;\nu, t)=\tilde{E}_0(\nu)$, so that the light source intensity per second and per unit area
\begin{equation}
  \label{eq:intensity-def}
\ts  \bar I_0 := \int_0^{+\infty} |\tilde{E}_0(q_1,q_2;\nu, t)|^2\ \ud \nu = \int_0^{+\infty} |\tilde{E}_0(\nu)|^2\ \ud \nu,
\end{equation}
is constant in the same way. Therefore, $c_0 \bar I_0$ is the total light exposure received per unit of time and per unit of surface on the biological specimen, with $c_0>0$ depending on the speed of light ${\sf c}_{\rm v}$ in the vacuum, the medium refractive index and the vacuum permittivity. Hereafter, for simplicity and up to a general rescaling of the intensities, we consider that $c_0=1$. 

Neglecting the refraction of light induced by the transparent biological specimen, we can forget the representation of $(q_1,q_2)$ in our next developments as well as in Fig.~\ref{fig:FTI_scheme}. This is equivalent to considering a single spatial location in the following equations. Moreover, the beam folding action of the optical elements in the Michelson interferometer (\ie the mirrors and the BS) are still compatible with the $q_3$-parametrization, provided that $q_3$ stands for the path length of the light propagation till the imaging sensor. Note that, from a one-to-one correspondence between the specimen plane and the imaging sensor, the coordinates $(q_1,q_2)$ are associated with a 2D pixel location on the imaging plane.     

Since a coherent light is a collection of monochromatic waves, we proceed by considering a monochromatic plane wave $\tilde{E}(q_3;\nu;t)=\tilde{E}_0(\nu) e^{i(2\pi q_3\nu-\omega t)}$ of wavenumber $\nu>0$ that is incident to the biological specimen. After having traveled through the specimen, this beam then reads
\begin{equation}
  \label{eq:outbeam}
  E(q_3;\nu;t) = E_0(\nu)e^{i(2\pi q_3\nu-\omega t)}.
\end{equation}

In non-fluorescent biological applications, the field intensity $|E_0(\nu)|^2$ is the result of the multiplication of $|\tilde E_0(\nu)|^2$ by the absorption $H(\nu)>0$ of the constituents. In fluorescent microscopy, however, the two field intensities are related through a more general model: in addition to absorption phenomena, a fluorescent material re emits light at different wavelengths compared to the one of the incident light beam. Consequently, we consider here that $E_0$ and $\tilde E_0$ are related through a general transfer operator $\cl H$, \ie 
\begin{equation}
  \label{eq:transfer-conversion}
  E_0(\nu) = \cl H[\tilde{E}_0](\nu),  
\end{equation}
which is linear in the total light intensity, \ie for any $\lambda >0$,  
\begin{equation}
  \label{eq:linearity-transfer-intensity}
\ts  \bar I_0 \to \lambda \bar I_0 \quad \Rightarrow\quad \bar I := \int_0^{+\infty} |\cl H[\tilde{E}_0](\nu)|^2 \ud \nu\  \to\ \lambda \bar I.  
\end{equation}
In words, if the illumination intensity is increased by a factor of $\lambda$, the same amplification is observed for the intensity of the light that has traveled though the specimen. As will be clear later, this assumption actually sustains a compressive FTI mode where the total light exposure is kept constant with respect to a Nyquist FTI acquisition (see Sec.~\ref{sec:Sec_05}).  

Next, from the field amplitude conversion \eqref{eq:transfer-conversion} occurring in the specimen, the beam~\eqref{eq:outbeam} is divided into two beams by the BS. After (perfect) reflexion on their respective mirror, and assuming perfect reflection and transmission in the BS, the incident beams at the BS read
$$
  E_1(q_3;\nu;t) = E_0(\nu) e^{i(2\pi q_3\nu-\omega t)} \quad\text{and}\quad E_2(q'_3;\nu;t) = E_0(\nu) e^{i(2\pi q'_3\nu-\omega t)},
$$
where $q_3$ and $q'_3$ are the distances that the two beams have traveled so far. Therefore, neglecting the equal additional paths traveled in the BS, the recombined beam is $E_T(q_3, q'_3;\nu;t) = E_1(q_3;\nu;t)+E_1(q'_3;\nu;t)$. 

Finally, the intensity $I_T(\xi;\nu) = |E_T(q_3, q'_3;\nu;t)|^2$ of this beam is measured in one pixel of the imaging sensor according to the rule    
\begin{equation*}
  I_T(\xi;\nu) = 2\,|E_0(\nu)|^2(1+\cos(2 \pi \nu\xi)),
\end{equation*}
where $\xi := q'_3-q_3$ is the OPD parameter. Here, $\xi = 0$ means that the moving mirror position makes the two arms of the Michelson interferometer, and thus the two light paths, of equal lengths. Hence, as illustrated in Fig.~\ref{fig:FTI_scheme}, there is a one-to-one correspondence between the time domain (or time slot), mirror position and OPD value. Moreover, to fix the ideas, we consider throughout this paper that each time slot (or OPD sample) integrates light over a constant time duration $\tau_{\xi} > 0$. This time is actually associated with the frame-per-second (fps) rate of the imaging sensor. 

Coming back to the general wide-band plane wave model \eqref{eq:wide-band-wave} with constant amplitude and assuming that $\tau_\xi$ is much larger than the temporal range\footnote{Which is clearly the case for a visible light and near-infrared HS system with $\nu {\sf c}_{\rm v}$ in the range of 400 THz to 770 THz, and $\tau_\xi$ of the order of a few hundredths of a second. } $1/(\nu {\sf c}_{\rm v})$, we quickly verify that the total temporal-averaged intensity recorded by the detector can be written~as
\begin{equation*}
  \ts I_T(\xi) = \int_{0}^{+\infty} I_T(\xi;\nu) \,\ud\nu=I_{{\rm DC}}+I(\xi),
\end{equation*}
with $I_{\rm DC} = 2\, \int_{0}^{+\infty} |E_0(\nu)|^2 \,\ud \nu = \frac{1}{2} I_T(0)$. After removing the DC (or mean) component, it is easy to show that the AC (or zero-mean) part~$I$, termed as \textit{interferogram}, is the Fourier transform of $|E_0(|\nu|)|^2$, \ie the symmetrization of $\nu\in\bb R_+ \mapsto |E_0(\nu)|^2 \in \bb R_+$ around $\nu=0$:
\begin{equation}
  \ts I(\xi)\ =\ I_T(\xi) - \frac{1}{2} I_T(0)\ =\ \int_{-\infty}^{+\infty}|E_0(|\nu|)|^2e^{-i2\pi \nu \xi} \,\ud\nu. \label{eq:continuous FTI model}
\end{equation}

Consequently, computing the inverse Fourier transform of the interferogram and restricting it to the positive spectral axis yields the spectrum\footnote{This inverse Fourier transform is actually proportional to the spectrum, with a multiplicative constant depending on, \eg the attenuation and reflection coefficients of the mirrors and the BS, and the sensor pixel efficiency.} $|E_0(\nu)|^2$ of the observed specimen on a given spatial location $(q_1,q_2)$.

Overall, reinserting these spatial coordinates, the continuous sensing model relating the HS volume $\cl X(\nu, q_1,q_2) := \cl M[|E_0(q_1, q_2; \nu)|^2]$, where $\cl M: f(q_1,q_2,\nu) \to f(q_1,q_2,|\nu|)$ is the symmetrization operator, to the volume of interferograms $\cl Y(\xi, q_1,q_2) = I(\xi) = I(\xi; q_1, q_2)$ reads
\begin{equation}
  \label{eq:full-continuous-sensing-model}
  \cl Y(\xi, q_1,q_2) = \cl F(\cl X)(q_1,q_2,\xi),
\end{equation}
with $\cl F : f(q_1,q_2,\nu) \to \int_{-\infty}^{+\infty} f(q_1,q_2, \nu) e^{-i2\pi \nu \xi} \,\ud\nu$ representing the 1D Fourier transform in the $\nu$ domain. 

\begin{remark} As represented through the action of the transfer operator $\cl H$ above, in biological imaging, the spectrum $|E_0(\nu)|^2$ is thus a combined signature of the absorption spectrum of the biological specimen, the emission spectra of the fluorescent dyes and the spectrum of the light source. In this work, we will not consider the question of unmixing these three parts and we thus focus on the acquisition of $|E_0(\nu)|^2$.    
\end{remark}

\begin{remark}
In practice, the location of $\xi = 0$, as required by~\eqref{eq:continuous FTI model}, can be obtained through calibration with two methods: \emph{(i)} by looking for the OPD point where the intensity of the recorded signal is twice the value of the empirical mean of the recorded signal, or  \emph{(ii)} by identifying the interferogram maximum, which occurs at the OPD origin from~\eqref{eq:continuous FTI model}. Note that the knowledge of the OPD origin is also critical to define our two VDS FTI schemes, as will be clear in Secs~\ref{sec:Sec_03} and~\ref{sec:Sec_04}.
\end{remark}

\subsection{Discrete Sensing Model} 
\label{sec:discrete-model}

FTI actually proceeds by first capturing discrete samples of the continuous volume $\mathcal{Y}$ in \eqref{eq:full-continuous-sensing-model}, both in the spatial and in the time domain according to the pixel grid and the number of frames per second recorded by the imaging sensor (see Fig.~\ref{fig:FTI_scheme}), and by processing them in order to reconstruct a discretized version of $\cl X$ compatible with the Shannon-Nyquist sampling theorem. 

To fix the ideas, we consider that \emph{(i)} the moving mirror gives access to an OPD domain $(-\xi_{\max}, \xi_{\max}]$  (for some range $\xi_{\max} >0$) that is evenly discretized over $N_\xi \in \bb N$ samples with an OPD step size $\Delta_\xi>0$ (\ie $2\xi_{\max} = N_\xi \Delta_\xi$), and \emph{(ii)} the spatial domain is sampled according to $\bar N_{\rm p} \times \bar N_{\rm p}$ square pixel grid with $N_{\rm p} = \bar N_{\rm p}^2$ pixels, \ie the grid of the same sensor, each pixel being square with side length $\Delta>0$. Note that, accordingly, the time domain is thus regularly discretized with $N_\xi$ samples related to a \emph{time slot} of $\tau_\xi = T/N_\xi$, given the total acquisition time\footnote{In practice, the FTI system is equivalently characterized from the frame-per-second rate of the imaging sensor, which determines $\tau_\xi$, and from the speed and the extent of the mirror motion, $2\xi_{\max}$.} $T>0$.

Mathematically, the discrete FTI measurements are gathered in a cube $\bs{\mathcal{Y}} \in \bb{R}^{N_\xi\times \bar N_{\rm p} \times \bar N_{\rm p}}$ approximating $\mathcal{Y}$ over $N_{\rm p} = \bar N_{\rm p}^2$ pixels and $N_\xi$ OPD points, \ie over $\Nhs = N_\xi N_{\rm p}$ voxels. Similarly, the discrete HS volume is represented by a data cube $\bs{\mathcal{X}} \in \bb{R}^{N_\nu\times \bar N_{\rm p} \times \bar N_{\rm p}}$ approximating $\mathcal{X}$ over $N_{\rm p}$ pixels and $N_\nu = N_\xi$ wavenumber samples, and thus also over $\Nhs$ voxels. Therefore, according to Shannon-Nyquist theorem \cite[Page 374]{Haykin:1998:SS:552094}, assuming $N_\xi$ even and positioning the OPD origin on the spectral index $\nell = N_\xi/2$, the sampling rules are thus 
\begin{align}
  \ts \bs{\mathcal{Y}}_{{\nell},{j_1},{j_2}}&\ts :=\mathcal{Y}((\nell-\inv{2} N_\xi)\Delta_\xi;j_1 \Delta_{\rm p},j_2 \Delta_{\rm p}),\nonumber\\
  \ts \bs{\mathcal{X}}_{{\nell'},{j_1},{j_2}}&\ts :=\mathcal{X}((\nell'-\inv{2} N_\nu)\Delta_\nu;j_1 \Delta_{\rm p},j_2 \Delta_{\rm p}),\label{eq:discrete model for intensities}
\end{align}
with $\nell, \nell' \in \range{N_\xi}$, $j_1,j_2 \in \range{\bar N_{\rm p}}$, $\Delta_\nu =1/(N_\xi\Delta_\xi)$ and $\nu_{\max}:=1/(2\Delta_\xi)$. Throughout this paper we call $\nell$ (and $\nell'$) as OPD index (resp. wavenumber index). Note that $\bs{\mathcal{X}}_{{\nell'},{j_1},{j_2}}=\bs{\mathcal{X}}_{N_\nu - {\nell'},{j_1},{j_2}}$ from the action of $\cl M$ in \eqref{eq:full-continuous-sensing-model}. This symmetric construction will be assumed throughout the paper, for the Nyquist FTI model above as well as for the compressive sensing schemes presented in the next sections. Consequently, despite the complex nature of the Fourier transform, all entries of $\cl Y$, partially observed or not, are real, in agreement with the measurement process. 

As a result, by increasing the number of OPD samples the final HS volume contains spectral information with higher details. However, since for a constant light source the total light exposure of an observed biological specimen is proportional to $N_\xi$, this increase is limited in practice by the risk of photo-bleaching~\cite{Diaspro2006}. Thus, a trade-off between spectral resolution and photo-bleaching must be found. 

Finally, if $\bs{Y}^{\Nyq} \in \bb R^{N_\xi \times N_{\rm p}}$ and $\bs{X} \in \bb R^{N_\xi \times N_{\rm p}}$ denote the matrix unfolding (see Sec.~\ref{sec:notations}) of the cube $\bs{\cl Y}$ and $\bs{\cl X}$, respectively, the acquisition process of Nyquist-FTI can be formulated in matrix form as
\begin{equation}
  \bs{Y}^{\Nyq} = \bs{F}^*\bs{X}+\bs{N}^{\Nyq},\label{eq:nyquist-fti}
\end{equation}
where $X_{l,j} = X_{N_\xi -l,j}$ for all $l \in \range{N_\xi/2}$ and $j\in\range{N_{\rm p}}$, $\bs{N}^{\Nyq}=[\bs{n}^{\Nyq}_1,\cdots,\bs{n}^{\Nyq}_{N_{\rm p}}] \in \bb{R}^{N_\xi \times N_{\rm p}}$ models an additive noise, and $\bs{F}  \in \bb{C}^{N_\xi \times N_\xi}$ is the 1D Discrete Fourier Transform (DFT) basis. Equivalently, this description can also be arranged into a vector form, \ie
\begin{equation}
\label{eq:nyquist-fti vector form}
  \bs{y}^{\Nyq} = \bs{\Phi}^*\bs{x}+\bs{n}^{\Nyq},
\end{equation}
with $\bs{\Phi}:=\bs{I}_{N_{\rm p}}\otimes\bs{F}$, $\bs{y}^{\Nyq}:={\rm vec}(\bs{Y}^{\Nyq})$, $\bs{x}:={\rm vec}(\bs{X})$, and $\bs{n}^{\Nyq}:={\rm vec}(\bs{N}^{\Nyq})$.

\subsection{Noise Model Identification and Estimation}
\label{sec:noise}

In practice, FTI experiments are of course corrupted by different noise sources. Let us list and assess the impact of the main ones. 

First, there is the \textit{observation noise} that mainly results from the combination of camera's electronic noise (such as thermal noise), quantization noise and photon noise. Electronic noise is commonly assumed additive, independent of the observations, and distributed as Gaussian and homoscedastic noise, \ie with constant variance for all measurements. Quantization noise is induced by the digitization of the observations. Strictly speaking, this is not noise but a deterministic distortion. However, when the bit-depth is large (\ie under the high-resolution assumption), each measurement is quantized over a large number of bits and the quantization distortion is also well modeled by an additive, heteroscedastic Gaussian noise~\cite{gray1998quantization}. Finally, the photon noise (or \emph{shot} noise) is signal dependent and associated with the quantized nature of light, \ie camera sensors record light intensity by integrating individual photons' energies over a given time interval. Photon noise is modeled by a Poisson distribution, whose variance is equal to its mean. Therefore, for high-photon counting rates, the ratio between the standard deviation and the mean of the light intensity vanishes, so that the impact of photon noise is limited before electronic and quantization noises. We will assume this high-photon counting regime in all our experiments (see Sec.~\ref{sec:Sec_06}).

Second, a \textit{modeling noise} is induced by the discretization of the sensing model \eqref{eq:full-continuous-sensing-model}, \ie by the discrepancy between the output of \eqref{eq:nyquist-fti}, given the discretization $\bs X$ of the continuous HS volume $\cl X$, and the true FTI observations in the absence of observation noises. In this work, we suppose that this modeling noise is limited by considering large resolutions $N_\xi$ and $N_{\rm p}$. Note that in \cite{roman2014asymptotic,adcock2013breaking}, the authors directly solve the infinite dimensional inverse problem posed by the reconstruction of a continuous function, sparsely representable in a continuous basis, and observed from finite or countable linear observations. We leave for a future study the application of this framework to our work. 

Third, there exist other \textit{instrument noises} corrupting the acquired data. For instance, bad instrument calibration induces an error in the OPD origin, supposed to lie on $\xi = 0$ in \eqref{eq:full-continuous-sensing-model}. We will see in the following that this point is important for applying VDS strategies. We omit this error by assuming an accurate calibration process. Additionally, Sec.~\ref{sec:FTI continuous model} does not consider light diffraction through the different optical elements of the Michelson interferometer and through the transparent biological specimen itself. This diffraction induces, however, spatial mixing of the recorded interferences, that would be modeled by convolving uncorrupted observation by an instrumental \emph{point spread function}.

\medskip
Consequently, since these noise sources are independent and the interferograms in~\eqref{eq:nyquist-fti} are already DC-free (with mean removed), we will suppose that the discrete FTI observations are corrupted by a global additive, zero-mean, homoscedastic Gaussian noise $\bs{N}^\Nyq = (\bs n^\Nyq_1, \cdots, \bs n^\Nyq_{N_{\rm p}}) \in \bb{R}^{N_\xi\times N_{\rm p}}$, with $N_{j,\nell}^\Nyq\sim_{\iid} \mathcal{N}(0,\sigma^2_\Nyq)$ for all OPDs $j \in \range{N_\xi}$ and pixels $\nell \in \range{N_{\rm p}}$. Given the recorded discrete and noisy FTI measurements $\bs Y^\Nyq$, the variance $\sigma^2_\Nyq$ can be estimated using, \eg the Robust Median (RM) estimator~\cite{donoho1994ideal,taswell2000and} (see Sec.~\ref{subsec:Simulated experimental data}). 

As an important aspect in the analysis of the two compressive FTI schemes developed in Sec.~\ref{sec:Sec_04} and Sec.~\ref{sec:Sec_05}, we will consider that the noise corrupting these compressive observations is also an additive homoscedastic Gaussian noise with the same variance, \ie with the variance $\sigma^2_\Nyq$ estimated from the Nyquist FTI acquisition.   

\section{Sparse Models for a Biological HS Volume}
\label{sec:sparsity models}

To solve problems of the form~\eqref{eq:BPDN}, it is crucial to select a sparsity basis $\bs \Psi$ that captures key information about the signal of interest with a minimum of basis elements: we must ensure that the first error term in the right-hand side of \eqref{eq:l2-l1-inst-optimality} decreases rapidly when the sparsity level increases.
\begin{remark}
Beyond the sparsity models used in this work, one can exploit other HS priors. A specimen is comprised of few biological elements whose spectral signatures share common support set in a sparsity basis. Moreover, smoothness of the spectral signatures (see, \eg Fig.~\ref{fig:gt}) results in highly correlated spectral bands. These priors can be imposed on the reconstruction problem using nuclear-norm, mixed-norm, or both~\cite{jacques2011panorama,golbabaee2012hyperspectral,golbabaee2012joint}.
\end{remark}
In this work, we focus on two HS sparsity priors both leveraging the Haar wavelet transform in the spectral or in the spatiospectral domains \cite{mallat2008wavelet}. While these choices are perfectible, \eg by selecting other wavelet schemes (such as the Daubechies wavelets) or redundant systems~\cite{mallat2008wavelet}, the two selected priors yield computable bounds on the local coherence of the corresponding sensing and sparsity bases. Following Sec.~\ref{sec:compr-uds-vds}, they thus provide guarantees on the quality of the reconstructed HS volume.
 
\paragraph{(i) Spectral sparsity} In the context of fluorescence microscopy of still biological specimen, each spatial location of an HS volume (\ie $\bs{\mathcal{X}}$ in Sec.~\ref{sec:discrete-model}) corresponds to the spectrum of the specimen on this position, \ie the mixture of the spectral signatures of the fluorescent dyes. In general, these spectral signatures are (piecewise) smooth (see, \eg \cite{fluorochromes,moshtaghpour2018multilevel}, and Fig.~\ref{fig:gt} for the spectra of 3 common fluorochromes); we expect them to be sparsely approximable in a wavelet basis, such as the 1D Haar wavelet system~\cite{mallat2008wavelet,jacques2011panorama}.

Our first possible prior will thus leverage this fact and be purely spectral. This will be crucial in Sec.~\ref{sec:Sec_03} for our first compressive FTI system, the Coded Illumination FTI, which subsamples the OPD domain of interferometric data according to an identical sampling pattern for all spatial locations; this scheme can thus be seen as the replication of $N_{\rm p}$ 1D CS systems that can only be controlled by the HS spectral sparsity.

Mathematically, our purely spectral sparsity basis is defined by $\bs \Psi := \bs I _{N_{\rm p}} \otimes \bs \Psi_{\rm 1D}$. This basis is associated with the following representations of the HS volume $\bs X$:
$$
  \bs X =\bs \Psi_{\rm 1D} \bs{S},\  \bs S =\bs \Psi^*_{\rm 1D} \bs{X}, \quad {\rm or}\quad \bs x = \bs \Psi \bs s,\ \bs s = \bs \Psi^* \bs x, 
$$
with $\bs s := {\rm vec}(\bs S)$, and where $\bs \Psi_{\rm 1D} \in \bb R^{N_\xi \times N_\xi}$ is the 1D discrete Haar wavelet basis. From the above considerations, the columns of $\bs S=(\bs s_1, \cdots, \bs s_{N_{\rm p}})$ (and the vector $\bs s$) are expected to have small best $K$-term approximation errors $\sigma_K(\bs s_i)_1$ (resp. $\sigma_K(\bs s)_1$) for a relatively small $K \ll N_\xi$ (resp. $K \ll \Nhs = N_\xi N_{\rm p}$). Sec.~\ref{sec:Sec_03} leverages this model for our analysis.\\[-2mm] 

\paragraph{(ii) Joint spatiospectral sparsity} The HS volume of a still specimen observed by an FTI system is made of a collection of monochromatic images describing the spatial configuration of fixed materials at a given wavenumber index. These images are expected to be sparsely approximable in a convenient 2D basis, \eg a 2D Haar wavelet basis $\bs \Psi_{\rm 2D} \in \bb R^{N_{\rm p} \times N_{\rm p}}$ \cite{mallat2008wavelet,jacques2011panorama}. 

Therefore, combining this new representation with the spectral sparsity model above, our second sparsity prior is defined by the sparsity basis $\bs \Psi :=\bs \Psi_{\rm 2D}\otimes\bs \Psi_{\rm 1D}$, where we focus on the common 2D Haar wavelet transform $\bs \Psi_{\rm 2D}$ with \emph{isotropic levels}, \ie obtained by multiplying 1D wavelets with identical scales for both the horizontal and vertical directions. However, our results extend to the case of \emph{anisotropic levels} where the identification of the scale of each 1D wavelet in each spatial direction is removed, \ie when we can write $\bs \Psi_{\rm 2D} = \bs \Psi_{\rm 1D} \otimes \bs \Psi_{\rm 1D}$~\cite{nowak1999wavelet} (see App.~\ref{sec:proof sifti}).
 
Thanks to this sparsity basis $\bs \Psi$,  we expect a small best $K$-term approximation error for the vectorization $\bs s$ of the matrix $\bs S \in \bb R^{N_\xi \times N_{\rm p}}$ in the representation $\bs X = \bs \Psi_{\rm 1D} \bs S \bs{\Psi}_{\rm 2D}^\top$, or $\bs x = \bs \Psi \bs s$. We will use this second model for SI-FTI scheme in Sec.~\ref{sec:Sec_04}, which allows for an easy characterization of this stronger sparsity prior.

\section{Coded Illumination-FTI (CI-FTI)}          
\label{sec:Sec_03}

In this section we focus on a simple modification of the FTI system. We introduce a temporal coding of the specimen illumination, \eg by acting on the global activation of the light source (see Fig.~\ref{fig:CIFTI_SIFTI_scheme}~(top)). This is potentially an easy adaptation, specially when the FTI module is combined with, \eg a confocal microscope where the light illumination can be programmed. We support this modification with the VDS scheme introduced in \cite{krahmer2014stable} (see Sec.~\ref{sec:compr-uds-vds}). We first describe the principles of a CI-FTI system before explaining how to reconstruct the observed HS volume. 

\subsection{Acquisition Strategy}
\label{sec:acquisition-strategy}

We here refer to \textit{coded illumination} as a technique that activates the light source in a portion of time slots, \ie coding temporally the light distribution. Recall that in Nyquist FTI, the light source is illuminating the biological specimen during $N_\xi$ time slots. As will be clear below, from the one-to-one correspondence between the time domain (or time slot), mirror position in the FTI system and OPD value (see Sec.~\ref{sec:FTI continuous model}), this temporal illumination coding amounts thus to sub-sample the OPD domain. 

Pursuing the context developed in Sec.~\ref{sec:FTI continuous model}, we assume that the light source delivers constant, but controllable, intensity $\bar I_0 := \int_0^{+\infty}|\tilde{E}_0(\nu)|^2\, \ud \nu$ per second and per unit area on any location of the biological specimen. If each OPD sample is associated with a time slot of $\tau_\xi$ second, then the total light exposure per unit area on each location of the biological specimen is equal to $N_\xi \tau_\xi \bar I_0$. The idea here is to reduce this exposure to $M_\xi \tau_\xi \bar I_0$ by activating the light source only over $M_\xi<N_\xi$ time slots; these being associated with a subset $\Omega^\xi := \{\beta_1,\, \cdots, \beta_{M_\xi}\}$ of $M_\xi$ (possibly non-unique) OPD samples $\beta_r \in \range{N_\xi}$ with $r \in \range{M_\xi}$.

For this purpose, we decide to follow the VDS scheme introduced in Sec.~\ref{sec:compr-uds-vds}; we generate $\Omega^\xi$ by randomly drawing its elements, \ie $\beta_r \sim_{\iid} \beta$, for $r\in\range{M_\xi}$, and $\beta \in \range{N_\xi}$ is a \rv whose pmf $p(\nell):=\bb{P}[\beta = \nell]$ is optimized in Sec.~\ref{sec:cifti_HS Reconstruction Method and Guarantee} according to the FTI sensing and the HS sparsity prior. 

According to the discrete FTI sensing model (see Sec.~\ref{sec:discrete-model}), this selection of OPD samples results in recording the collection of interferograms $\mathcal{Y}\big((\nell-N_\xi/2)\Delta_\xi; j_1,j_2\big)$ at OPD indices $\nell \in \Omega^{\xi}$. The acquisition model of CI-FTI then reads 
\begin{figure}[!t]
  \centering
  \includegraphics[width=0.8\textwidth]{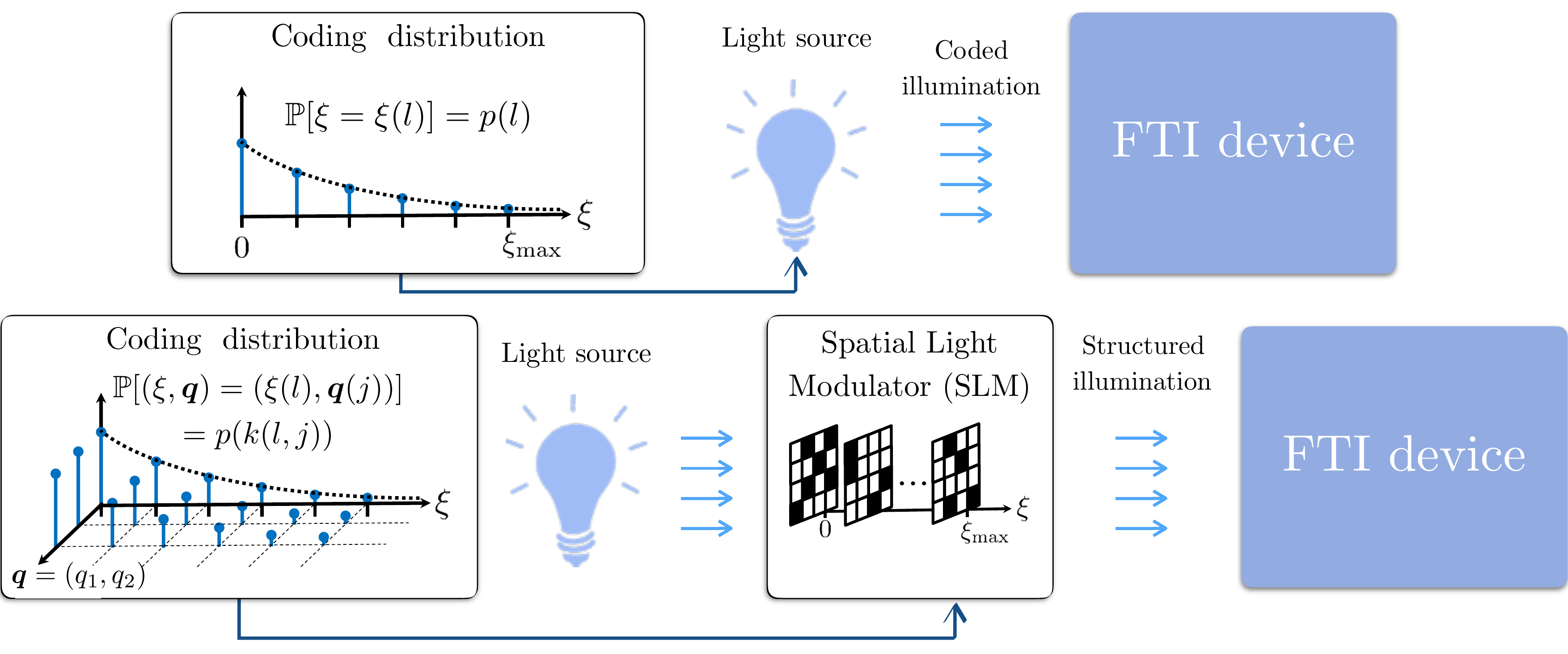}
  \caption{Illustration of (top) coded illumination-FTI and (bottom) structured illumination-FTI systems according to the coding distributions (\ie the pmfs of VDS) established in Sec.~\ref{sec:Sec_03} and Sec.~\ref{sec:Sec_04}. For the sake of simplicity only the positive part of OPD axis is shown. In these figures, $\xi(\nell)$ represents the $\nell^{\rm th}$ OPD location, \ie $\xi(l) := (l - N_\xi/2)\Delta_\xi$, and $\bs q(j)$ the $j^{\rm th}$ spatial location in the ordering of the vectorized spatial coordinates.}
  \label{fig:CIFTI_SIFTI_scheme}
\end{figure}
\begin{equation}
  \bs{Y}^{\CI} = \bs{F}^*_{\Omega^{\xi}}\bs{X}+\bs{N},~ {\rm or}~  \bs{y}^{\CI} = (
  \bs{I}_{N_{\rm p}}\otimes\bs{F}^*_{\Omega^{\xi}})\bs{x}+\bs{n} = \bs y^\Nyq_\Omega,\label{eq:cifti acquisition}
\end{equation}
where $\Omega= \bigcup_{j=1}^{N_{\rm p}} \{N_\xi(j-1)+\Omega^\xi\}$ contains the indices associated with all selected entries of $\bs{y}^{\Nyq}$, $\bs{n} = {\rm vec}(\bs{N})$, and $\bs{N} \in \bb{R}^{M_\xi \times N_{\rm p}}$ models an additive measurement~noise. 

Following the considerations of Sec.~\ref{sec:noise}, we assume 
that $\bs N$ is a random Gaussian noise with variance $\sigma^2$, \ie $N_{r, j} \sim_{\iid} \cl N(0, \sigma^2)$ for $r\in \range{M_\xi}$ and $j \in \range{N_{\rm p}}$. Moreover, we suppose that $\sigma$ is independent of both the number of measurements and the observed HS volume. The validity of this assumption will be confirmed in Sec.~\ref{subsec:Simulated experimental data} by showing that $\sigma$ only moderately grows when the intensity of the HS volume strongly increases. Consequently, we set hereafter $\sigma^2 = \sigma^2_\Nyq$, with $\sigma^2_\Nyq$ being the variance of the measurement noise on each Nyquist FTI observation, \eg estimated from an RM estimator~\cite{donoho1994ideal,taswell2000and}, or by the calibration of the FTI system.

\subsection{HS Reconstruction Method and Guarantee}\label{sec:cifti_HS Reconstruction Method and Guarantee}
Given the noisy CI-FTI measurements in~\eqref{eq:cifti acquisition}, we leverage the VDS scheme supported by Prop.~\ref{prop:variable_density_sampling} and Prop.~\ref{prop:RIP}, and propose the following convex optimization problem for recovering the $\Nhs$-voxel HS volume, \ie
\begin{equation}
  \label{eq:cifti main reconstruction}
  \hat{\bs{x}} = \underset{\bs{u}\in \bb{C}^{\Nhsexp}}{\argmin} \|\bs{\Psi}^\top\bs{u}\|_1\ \st\ \|\bs{D}^\xi(\bs{y}^{\CI}_j-\bs{F}^*_{\Omega^\xi}\bs{u}_j)\|\le \varepsilon_\CI \sqrt{M_\xi},~\forall j \in \range{N_{\rm p}},
\end{equation}
where $\Omega^\xi = \{\beta^1,\,\cdots,\beta^{M_\xi}\}$ is randomly generated as described in Sec.~\ref{sec:acquisition-strategy}, $\bs{D}^\xi = \diag(\bs d^\xi) \in \bb{R}^{M_\xi \times M_\xi}$ is a diagonal matrix such that $d^\xi_{r} = 1/(p(\beta^r))^{1/2}$, and $\varepsilon_\CI$ is a bound on the measurement noise level such that (with high probability) $\|\bs{D}^{\xi}\bs{n}_j\|\le\varepsilon_\CI\sqrt{M}$ for all $j \in \range{N_{\rm p}}$. We will characterize it momentarily.

\begin{remark}
For the sake of consistency with the result of Krahmer and Ward in Prop.~\ref{prop:variable_density_sampling}, we do not force the estimation of \eqref{eq:cifti main reconstruction} to be real-valued, even, or non-negative. Adding those extra prior information do not worsen the theoretical guarantee in \eqref{eq:cifti main reconstruction}, while, numerically, they improve the quality of the HS data recovery, as they reduce the set of feasible solutions in~\eqref{eq:cifti main reconstruction}.
\end{remark}

In the optimization \eqref{eq:cifti main reconstruction}, while a global fidelity constraint could have been set between the partial interferometric observations $\bs y^\CI$ and the candidate HS volume $\bs u$, we rather impose $N_{\rm p}$ individual fidelity constraints, one per pixel index $j\in \range{N_{\rm p}}$. Moreover, concerning the regularizer of \eqref{eq:cifti main reconstruction}, we consider only a spectral sparsity prior described in Sec.~\ref{sec:sparsity models}, \ie $\bs \Psi =\bs I_{N_{\rm p}} \otimes \bs \Psi_{\rm 1D}$. The compressibility of the acquisition is indeed mainly brought on the spectral domain, with an OPD subsampling pattern shared for all spatial locations. As a result (expressed in the next lemma), the choice of these fidelity constraints together with this specific prior allows for the decomposition of \eqref{eq:cifti main reconstruction} into small-size problems, for which we can find individual recovery guarantee.
\begin{lem}\label{lemma:cifti}
  Problem~\eqref{eq:cifti main reconstruction} can be decoupled into $N_{\rm p}$ subproblems
  \begin{equation}
    \hat{\bs{x}}_j = \underset{\bs{u}\in \bb{C}^{N_\xi}}{\argmin} \|\bs{\Psi}_{\rm 1D}^\top\bs{u}\|_1\ \st\ \|\bs{D}^\xi(\bs{y}^{\CI}_j-\bs{F}^*_{\Omega^\xi}\bs{u}_j)\|\le \varepsilon_\CI \sqrt{M_\xi},\label{eq:cifti sub reconstruction}
  \end{equation}
  with $1\le j \le N_{\rm p}$.
\end{lem}
\begin{proof}
  From the separability of the $\ell_1$-prior $\|\bs{\Psi}^\top\bs{u}\|_1 = \sum_{j=1}^{N_{\rm p}}\|\bs{\Psi}_{\rm 1D}^\top\bs{u}_j\|_1$ and according to the argument in~\cite[pp. 337]{rockafellar2015convex}, the proof is straightforward.
\end{proof}

From $N_{\rm p}$ sub-problems in \eqref{eq:cifti sub reconstruction}, we can develop a recovery guarantee for the reconstruction of any HS volume $\bs X$ from \eqref{eq:cifti main reconstruction} according to the specified sparsity basis. According to Prop.~\ref{prop:variable_density_sampling}, given the vector of bounds $\bs \kappa$ on the local coherence between $\bs \Psi_{1D}$ and $\bs F$, \ie with $\mu^{\rm loc}_\nell(\bs F, \bs \Psi_{\rm 1D}) \leq \kappa_\nell$, by selecting $M_\xi \gtrsim \delta^{-2}\|\bs \kappa\|^2 K_\xi \log^3(K_\xi) \log(N_\xi)$ OPD indices with respect to the pmf defined in~\eqref{eq:vds pmf}, the RIP of order $K_\xi$ holds for the matrix $(M_\xi)^{-1/2}\bs{D}^\xi\bs{F}_{\Omega^{\xi}}\bs{\Psi}_{\rm 1D}$ with probability exceeding $1-N_\xi^{-c\log^3(K_\xi)}$. 
Since Prop.~\ref{prop:RIP} and Prop.~\ref{prop:variable_density_sampling} provide uniform guarantee in the sense that the recovery is ensured for all possible signals, the reconstruction of each column of the HS volume from \eqref{eq:cifti sub reconstruction} satisfies 
\begin{equation*}
  \ts    \|\bs{x}_j-\hat{\bs{x}}_j\| \le \frac{2\sigma_{K_\xi}(\bs{\Psi}_{\rm 1D}^\top\bs{x}_j)_1}{\sqrt{K_\xi}}+\varepsilon_\CI,\quad \forall j\in\range{N_{\rm p}},
\end{equation*}
with probability exceeding $1-N_\xi^{-c\log^3(K_\xi)}$. Consequently, since $(a+b)^2 \le 2(a^2+b^2)\le 2(a+b)^2$ for all $a, b > 0$, we can bound the estimation error of the whole HS image as follows
\begin{equation}
  \ts    \|\bs{x}-\hat{\bs{x}}\| \le \frac{2\sqrt{2}}{\sqrt{K_\xi}}\left(\sum_{j=1}^{N_{\rm p}}\left(\sigma_{K_\xi}(\bs{\Psi}_{\rm 1D}^\top\bs{x}_j)_1\right)^{2}\right)^{1/2}+\sqrt{2N_{\rm p}}\varepsilon_\CI.\label{eq:cifti error}
\end{equation}

In order to adjust the noise level $\varepsilon_\CI$ as a function of known/estimable parameters, \eg $M_\xi$, $N_\xi$ and $\sigma_\Nyq$, we need to characterize the pmf responsible for the selection of the OPD indices.

\begin{proposition}
\label{prop:cifti local coherence}
In the framework of Prop.~\ref{prop:variable_density_sampling}, set $K = K_\xi$, $N = N_\xi$, $M = M_\xi$, $\bs{\Phi} = \bs{F}$ and $\bs \Psi = \bs{\Psi}_{\rm 1D}$. From this choice, we have 
$$
\ts \mu_\nell(\bs F, \bs{\Psi}_{\rm 1D})\ \leq\ \kappa_\nell := \sqrt 2 \min\big\{1, |\nell - \frac{N_\xi}{2}|^{-1/2}\big\},\quad \nell \in \range{N_\xi},
$$
with $\|\bs \kappa\|^2 \leq 8 + 4\log(\frac{N_\xi}{2}) \lesssim \log(N_\xi)$. Moreover, the corresponding pmf in \eqref{eq:vds pmf} is given by 
  \begin{equation}
\label{eq:pmf-ci-FTI}
    \ts p(\nell) =C_{N_\xi}\,\min\{ 1, |\nell-\frac{N_\xi}{2}|^{-1}\}, \quad \nell \in \range{N_\xi},
  \end{equation}
  where the normalization constants $C_{N_\xi}$ respects $2\log(N_{\xi}/2) < C_{N_\xi}^{-1} < 4+2\log(N_{\xi}/2)$.
\end{proposition}
For the proof see App.~\ref{sec:proof cifti}. Note that we computed a bound (not the exact values) for the local coherence. Even though this tight bound is not optimal, it defines a meaningful sampling pmf \eqref{eq:pmf-ci-FTI} and sample complexity bound \eqref{eq:cifti m}. We will see in Sec.~\ref{sec:Sec_06} that the performance of the proposed pmf is very close to the performance of the optimal pmf computed numerically. The same argument holds for Prop.~\ref{prop:sifti local coherence} related to the SI-FTI system.

\begin{remark}
Following this proposition, if the random quantities 
$\Omega = \Omega^\xi$ and $\bs D = \bs D^\xi$ defined above are specified by the pmf of Prop.~\ref{prop:cifti local coherence}, we can estimate the level of noise in CI-FTI at all pixels $j \in \range{N_{\rm p}}$. 
Since the pmf $p$ in \eqref{eq:pmf-ci-FTI} corresponds to a VDS scheme of exponent $\alpha=1$, offset $l_0 = N_\xi/2$ and constant $c_\eta = C_{N_\xi}$ in Rem.~\ref{remark:simple_bound_for_rho}, we find $\rho = C_{N_\xi}^{-1} (N_\xi - (N_\xi/2))/ N_\xi \leq 2+\log(N_{\xi}/2)$. Therefore, setting there $\sigma = \sigma_\Nyq$ and $\eta(l) = p(l)$ in 
Cor.~\ref{cor:noise-level-estim-Gauss-noise} with that bound for $\rho$, and using a union bound over all pixels, we get
\begin{equation}
\ts\frac{1}{\sqrt M_\xi} \|\bs{D}^\xi(\bs{y}^{\CI}_j-\bs{F}^*_{\Omega^\xi}\bs{x}_j)\|\ \leq\ \varepsilon_\CI(s) := \varepsilon_{\sigma_\Nyq,s}(N_\xi, M_\xi, 2+\log(N_\xi/2)),~~\forall j \in \range{N_{\rm p}},\label{eq:cifti_final_noise_estimate}
\end{equation}
with probability exceeding $1-3N_{\rm p} e^{-s/2}$ and $\varepsilon_{\sigma,s}$ defined in \eqref{eq:noise-level-estim-Gauss-noise-ext-noise}.
\end{remark}

We are now ready to summarize all the analysis in this section and to provide the main result for CI-FTI in the following theorem.

\begin{thm}[Coded illumination-FTI]\label{thm:cifti}
  Given $s>0$, fix integers $K_\xi$, $N_\xi$, $N_{\rm p}$, $M_\xi$ such that $K_\xi \gtrsim \log(N_\xi)$ and 
  \begin{equation}
    M_\xi \gtrsim K_\xi \log^3(K_\xi) \log^2(N_\xi).\label{eq:cifti m}
  \end{equation}
Generate $M_\xi$ (possibly non-unique) OPD indices $\Omega^{\xi}=\{\beta^1,\,\cdots,\beta^{M_\xi}\} \subset \range{N_\xi}$ such that $\beta^r \sim_\iid \beta$ for $r \in \range{M_\xi}$, with $\beta$ a \rv with the pmf \eqref{eq:pmf-ci-FTI}. Then, given the corresponding noisy CI-FTI measurements $\bs y^{\CI}$ in~\eqref{eq:cifti acquisition}, the HS volume $\bs x$ can be approximated by solving~\eqref{eq:cifti main reconstruction} with the bound $\varepsilon_\CI(s)$ in \eqref{eq:cifti_final_noise_estimate}, up to an error 
  \begin{equation}
    \ts    \|\bs{x}-\hat{\bs{x}}\| \le \frac{2\sqrt{2}}{\sqrt{K_\xi}}\left(\sum_{j=1}^{N_{\rm p}}\left(\sigma_{K_\xi}(\bs{\Psi}_{\rm 1D}^\top\bs{x}_j)_1\right)^{2}\right)^{1/2}+\sqrt{2N_{\rm p}}\varepsilon_\CI(s),\label{eq:cifti error final}
  \end{equation}
  and with probability exceeding $1-N_\xi^{-c\log^3(K_\xi)}-3N_{\rm p}e^{-s/2}$.
\end{thm}

\begin{proof}
  A combination of Prop.~\ref{prop:variable_density_sampling}, Prop.~\ref{prop:RIP}, Lem.~\ref{lemma:cifti}, and Prop.~\ref{prop:cifti local coherence} with \eqref{eq:cifti error} completes the proof.
\end{proof}
In words, this theorem states that the pmf in (5.5) associated with a near-optimal sampling strategy, follows a two-sided power-law decay centered at zero OPD point $\xi=0$.

As discussed in Sec.~\ref{sec:fti}, zero OPD point can be determined prior to FTI acquisition process and be used for developing the above-mentioned sampling strategy. Moreover, by leveraging this VDS strategy, the amount of required light exposure for successful HS recovery behaves of the order of $K_\xi$, \ie for typical HS volumes $K_\xi \ll N_\xi$. However, following the calculation of $\bs \kappa$ in App.~\ref{sec:proof cifti}, it is seen that the use of a UDS strategy gives $M_\xi \gtrsim N_\xi K_\xi \log^3(K_\xi) \log(N_\xi)$, meaning that the reduction in light exposure is not possible in this context.

\section{Structured Illumination-FTI (SI-FTI)}             
\label{sec:Sec_04}

Structured Illumination-FTI is an alternative approach to further reduce the light exposure on a biological specimen compared to CI-FTI. In short, the light distribution is here coded (or \emph{structured}) in both spatial and OPD (or time) domains. We explain hereafter the SI-FTI acquisition procedure, the associated HS volume estimation problem, and we deduce theoretical guarantees on the quality of the reconstructed volume.

\subsection{Acquisition Strategy}
\label{sec:acquisition-strategy-subsec}

SI-FTI allows for the selection of different OPD samples at each specimen location, as opposed to CI-FTI where the selected OPD indices are shared for all locations. We assume in this work that this is achieved by spatially structuring (coding) the illumination of the system, \ie thanks to a spatial light modulator (SLM)~\cite{gehm2007single,duarte2008single} in between of the light source and the biological specimen (see Fig.~\ref{fig:CIFTI_SIFTI_scheme} (bottom)). 

\begin{remark} For the sake of simplicity, we suppose that the SLM and the imaging sensor have identical resolutions (\ie $N_{\rm p}$ pixels for both) and that there are perfect alignment and scaling between the discrete illumination pattern and the pixel grid of the camera, \eg using an appropriate optical magnification system (not represented in Fig.~\ref{fig:FTI_scheme}). Therefore, one coded location on a thin, still, 2D biological specimen can be identified with one pixel location of the camera. In other words, masking one location of the biological specimen over an area $\delta_S$ fixed by the SLM pixel pitch corresponds to blocking the light received by a single pixel of the camera. Additionally, we assume the SLM does not reduce the light intensity when its pixels are activated (\ie the system has 100\,\% light throughput).     
\end{remark}

To understand the potential benefit of SI-FTI let us observe that a Nyquist-FTI scheme can be recovered from such a system simply by activating all SLM pixels for all OPD samples. In this case, if we assume again that the light source delivers constant intensity $\bar{I}_0$ per second and per unit area on any location of the specimen, the total light exposure at every OPD point on the whole biological specimen is constant and equal to $N_{\rm p} \delta_S \tau_\xi\bar{I}_0$, where $\tau_{\xi}$ corresponds to the duration of each OPD sample and $\delta_S$ is the SLM pixel area. Using structured illumination in SI-FTI, if $\tilde{M}_\nell < N_{\rm p}$ spatial locations are exposed on the specimen at the $\nell^{\rm th}$ OPD sample with $\nell \in \range{N_\xi}$, this exposure can be reduced to $\tilde{M}_\nell\delta_S\tau_\xi\bar{I}_0$ for this OPD sample. Therefore, the total light exposure undergone by the specimen during a Nyquist FTI acquisition, \ie $\Nhs \delta_S\tau_\xi\bar{I}_0$ with $\Nhs =N_\xi N_{\rm p}$, is decreased to $M\delta_S\tau_\xi\bar{I}_0$ in SI-FTI, where $M=\sum_{\nell=1}^{N_\xi}\tilde{M}_\nell$. The question is of course to be able to reconstruct the HS volume from such a partial FTI observations. 

\medskip
Let us now turn to SI-FTI sensing model. We follow the general random VDS scheme of Sec.~\ref{sec:compr-uds-vds}, with special care to integrate the spatiospectral geometry of SI-FTI. We denote by $p(j,\nell):=\bb{P}[(\beta_{\rm p},\beta_\xi)=(j,\nell)]$ a bivariate pmf determining the random activation of the $j^{\rm th}$ spatial location at the $\nell^{\rm th}$ OPD point, \ie the pmf of a bivariate \rv $\bs \beta = (\beta_{\rm p}, \beta_\xi)$ for $\beta_{\rm p} \in \range{N_{\rm p}}$ and $\beta_\xi \in \range{N_\xi}$. In this context, we can generate a random set $\bar \Omega$ with $M$ (possibly non-unique) elements from $\bar \Omega = \{\bs \beta^1, \cdots \bs \beta^M\}$ with $\bs \beta^r \sim_{\rm iid} \bs \beta$ and $r\in\range{M}$. 

The sensing model then reads
$$
\ts y_r = \bs F_{\{\beta^r_\xi\}}^* \bs x_{\beta^r_{\rm p}} + n_r, \quad \forall r \in \range{M} ,
$$  
with $n_r$ modeling an additive measurement noise. As for CI-FTI, we assume this noise Gaussian, \ie $n_r \sim_{\rm iid} \cl N(0, \sigma^2_\Nyq)$ with variance $\sigma^2_\Nyq$ that can be estimated, \eg from the Nyquist measurements or by system calibration.   
\medskip

To reach more compact notation, we adopt a purely 1D random sensing model.  We consider the set $\Omega \subset \range{\Nhs}$ generated from $M$ (scalar) \rvs $\beta^r \sim_\iid \beta$ (with $r \in \range{M}$), where the \rv $\beta \in \Nhs$ is defined from the pmf 
\begin{equation}
  \label{eq:Omega-bivar-pmf}
\ts p(k) := \bb P(\beta = k) = \bb P( N_\xi (\beta_{\rm p} - 1) + \beta_\xi = k),\quad \forall k \in \range{\Nhs},  
\end{equation}
with $\bs \beta = (\beta_{\rm p}, \beta_\xi)$ the bivariate \rv defined above from the pmf $p(j,l)$. In this case, for each $j \in \range{N_{\rm p}}$,
$$
\Omega^j\, :=\, \big(\Omega \cap [N_\xi (j - 1) +1 , N_\xi j]\big)\, -\, N_\xi (j - 1),
$$
is the (possibly empty\footnote{With the convention that $\bs u_{\emptyset} = \emptyset$ for any vector $\bs u$.}) set of OPD samples selected at the $j^{\rm th}$ pixel. 

SI-FTI then amounts to
\begin{equation}
\ts  \bs{y}_j = \bs{F}^*_{\Omega^j}\bs{x}_j+\bs{n}_j,\ \forall j \in \range{N_{\rm p}}\quad \Leftrightarrow\quad \bs{y}^{\SI} = [\bs{y}_1^\top,\cdots,\bs{y}_{N_{\rm p}}^\top]^\top = \bs{\Phi}^*_{\Omega} \bs{x}+\bs{n}, \label{eq:sifti acquisition}
\end{equation}
where $\bs{\Phi} := \bs{I}_{N_{\rm p}}\otimes\bs{F}$ and $\bs{n}:=(\bs{n}_1^\top, \cdots, \bs{n}_{N_{\rm p}}^\top)^\top \in \bb{R}^{M}$. 

\medskip

Note that we can go back and forth between the 1D and the 2D index representations $k \in \range{\Nhs}$ and $(j,\nell) \in \range{N_{\rm p}}\times \range{N_\xi}$, respectively, using the relations
\begin{equation}
  \label{eq:transl-1D-2D-ind-repr}
\ts k(j,\nell) = N_\xi (j - 1) + \nell\ \Leftrightarrow\ \nell(k) = ((k-1)\!\!\!\mod N_\xi)+1,\ j(k) = (k - b_\xi)/N_\xi.
\end{equation}

The pmf $p$ defining SI-FTI sensing is considered hereafter as a degree of freedom that we are going to relate to the VDS scheme formulated in Prop.~\ref{prop:variable_density_sampling}. This will allow us to reach an optimized structured illumination strategy, as presented in Thm~\ref{thm:sifti}. 

\subsection{HS Reconstruction Method and Guarantee}\label{sec:sifti reconstruction}
Given the noisy SI-FTI measurements as in~\eqref{eq:sifti acquisition}, an HS volume $\bs{x}$ with $\Nhs$ voxels can be reconstructed via the convex optimization problem
\begin{equation}
  \hat{\bs{x}} = \underset{\bs{u}\in \bb{C}^{\Nhsexp}}{\argmin} \|\bs{\Psi}^\top\bs{u}\|_1\ \st\ \|\bs{D}(\bs{y}^{\SI}-\bs{\Phi}^*_\Omega\bs{u})\|\le \varepsilon_\SI \sqrt{M},\label{eq:sifti main reconstruction}
\end{equation}
where $\Omega = \{\beta^1, \cdots, \beta^M\} \subset \range{\Nhs}$ is randomly generated according to the pmf of \eqref{eq:Omega-bivar-pmf}, $\varepsilon_\SI$ must be such that $\|\bs{D}\bs{n}\|\le\varepsilon_\SI\sqrt{M}$ with high probability, and $\bs{D} = \diag(\bs d) \in \bb R^{M\times M}$ with $d_{r} = 1/(p(\beta^r))^{1/2}$ for $r \in \range{M}$. 

For generality of our model, we regularize Problem~\eqref{eq:sifti main reconstruction} with the joint spatiospectral HS sparsity model described in Sec.~\ref{sec:sparsity models}, \ie  $\bs{\Psi} :=  \bs{\Psi}_{\rm 2D} \otimes\bs{\Psi}_{\rm 1D}$. Contrary to the CI-FTI optimization Problem~\eqref{eq:cifti main reconstruction}, Problem~\eqref{eq:sifti main reconstruction} cannot be decoupled into sub-problems since $\|\bs{\Psi}^\top\bs{u}\|_1$ is not separable in $\{\bs u_j: j\in\range{N_{\rm p}}\}$ and the random set $\Omega$ detailed above is not separable in the spatial and OPD domains. Therefore, we tackle the reconstruction problem of the full HS volume at once. 

Having set the sparsity basis, we can now adjust the pmf \eqref{eq:transl-1D-2D-ind-repr} determining $\Omega$ from Prop.~\ref{prop:variable_density_sampling}. According to this proposition, the preconditioned matrix $\frac{1}{\sqrt{M}}\bs{D}\bs{\Phi}^*_\Omega\bs{\Psi}$ respects the RIP of order $K$ with probability exceeding $1-\Nhs^{-c\log^3(K)}$ if 
$$
\ts
M \gtrsim \delta^{-2} \|\bs \kappa\|^2 K \log^3(K)\log(\Nhs)
$$
SI-FTI measurements are recorded with respect to the pmf 
$p(k) := {\kappa_k^2}/{\|\bs \kappa\|^2}$ ($k\in\range{\Nhs}$), where the vector $\bs \kappa \in \bb R_+^{\Nhs}$ is a bound for the local coherence $\mu_k(\bs \Phi, \bs \Psi)$ with $k \in \range{\Nhs}$, $\bs{\Phi}= \bs{I}_{N_{\rm p}} \otimes \bs{F}$ and $\bs{\Psi}=\bs{\Psi}_{\rm 2D} \otimes \bs{\Psi}_{\rm 1D}$. The next proposition (proved in App.~\ref{sec:proof sifti}) bounds this local coherence, and thus determines the pmf $p$.
\begin{proposition}\label{prop:sifti local coherence}
In the context of Prop.~\ref{prop:variable_density_sampling}, set $N = \Nhs = N_\xi N_{\rm p}$, $\bs{\Phi}= \bs{I}_{N_{\rm p}}\otimes\bs{F}$, and $\bs \Psi = \bs{\Psi}_{\rm 2D} \otimes \bs{\Psi}_{\rm 1D}$. We have then 
\begin{equation}
\label{eq:coh-bound-SI}
\ts \mu_k(\bs \Phi, \bs \Psi) \leq \kappa_k := \frac{\sqrt 2}{2} \min\big\{1, \frac{1}{\sqrt{|\nell - (N_\xi/2)|}}\big\},\ k \in \range{\Nhs}, 
\end{equation}
with $\|\bs \kappa\|^2 \leq N_{\rm p} (2 + \log(N_\xi/2)) \lesssim N_{\rm p} \log N_\xi$ and the relation $k = k(j,\nell)$ defined in \eqref{eq:transl-1D-2D-ind-repr}. In this case, Prop.~\ref{prop:variable_density_sampling} holds for 
  \begin{equation}
\label{eq:p-bound-SI}
\ts p(k) =\frac{C_{N_\xi}}{N_{\rm p}} \min \{1,|\nell-(N_\xi/2)|^{-1}\},
  \end{equation}
where $k = k(j,\nell) \in \range{\Nhs}$, and the normalization constant $C_{N_\xi}$ respects  $2\log(N_{\xi}/2) < C_{N_\xi}^{-1} < 4+2\log(N_{\xi}/2)$.
\end{proposition}

If $\frac{1}{\sqrt{M}}\bs{D}\bs{\Phi}^*_\Omega\bs{\Psi}$ respects the RIP with the pmf specified in the previous proposition, Prop.~\ref{prop:RIP} shows that the estimation error achieved by~\eqref{eq:sifti main reconstruction} can be bounded as
\begin{equation}
  \ts  \|\bs{x}-\hat{\bs{x}}\| \le \frac{2\sigma_{K}(\bs{\Psi}^\top\bs{x})_1}{\sqrt{K}}+\varepsilon_\SI.\label{eq:sifti error}
\end{equation}

Moreover, following Prop.~\ref{prop:sifti local coherence}, Cor.~\ref{cor:noise-level-estim-Gauss-noise} and Rem.~\ref{remark:simple_bound_for_rho}, we can estimate the level $\varepsilon_\SI$ of a Gaussian noise in the SI-FTI model \eqref{eq:sifti acquisition}. 
Indeed, we can compute the bound \eqref{eq:pmf-moment} associated with the pmf \eqref{eq:p-bound-SI} since, for any integer $q\geq 1$, we easily prove that
$$
\ts \tinv{\Nhs} \big(\bb E_\beta\,p(\beta)^{-q}\big)^{1/q} = \tinv{N_\xi} \big (\bb E_\gamma \,p'(\gamma)^{-q}\big)^{1/q},
$$
where $\beta \in \range{\Nhs}$ is a \rv with the pmf $p$ of \eqref{eq:p-bound-SI}, and $\gamma \in \range{N_\xi}$ is a \rv with pmf $p'(l) := C_{N_\xi} \min\{1,|\nell-(N_\xi/2)|^{-1}\}$. From Rem.~\ref{remark:simple_bound_for_rho}, $p'$ is thus a VDS scheme with exponent $\alpha=1$ and offset $\nell_0 = N_\xi/2$. We can thus set $\rho = 2+\log(N_\xi/2) > C_{N_\xi}^{-1}/2 = C_{N_\xi}^{-1} (N_\xi - \nell_0)/N_\xi$, so that 
\begin{equation}
\ts \frac{1}{M} \|\bs{D}\bs n\|^2 = \frac{1}{M} \|\bs{D}(\bs{y}^{\SI}-\bs{\Phi}^*_\Omega\bs{x})\|^2 \leq \varepsilon_\SI^2(s) := \varepsilon^2_{\sigma_\Nyq,s}(\Nhs, M, 2+\log(N_\xi/2)), \label{eq:sifti_final_noise_estimate}
\end{equation}
with probability exceeding $1-3e^{-s/2}$ and $\varepsilon_{\sigma,s}$ defined in \eqref{eq:noise-level-estim-Gauss-noise-ext-noise}.
\medskip

We are now ready to summarize the complete analysis of SI-FTI in the following theorem.
\begin{thm}[Structured Illumination-FTI]\label{thm:sifti}
  Given $s >0$, fix integers $K$, $\Nhs = N_\xi N_{\rm p}$ such that $K \gtrsim \log(\Nhs)$ and 
  \begin{equation}
    M \gtrsim N_{\rm p} K \log^3(K) \log(N_\xi)\log(\Nhs).\label{eq:sifti m}
  \end{equation}
Generate $M$ random (non-unique) indices associated with a (1D) index set $\Omega = \{\beta^1,\cdots,\beta^M\}$ such that $\beta^r \sim_{\iid} \beta$ for $r \in \range{M}$, with $\beta$ a \rv with the pmf \eqref{eq:p-bound-SI}. Then, given the noisy SI-FTI measurements $\bs y^{\SI}$ in \eqref{eq:sifti acquisition}, the HS volume $\bs x$ can be approximated by solving~\eqref{eq:sifti main reconstruction} with the bound $\varepsilon_\SI(s)$ in \eqref{eq:sifti_final_noise_estimate}, up to an error 
  \begin{equation}
    \ts  \|\bs{x}-\hat{\bs{x}}\| \le \frac{2\sigma_{K}(\bs{\Psi}^\top\bs{x})_1}{\sqrt{K}}+\varepsilon_\SI(s),\label{eq:sifti final error}
  \end{equation}
and with probability exceeding $1-\Nhs^{-c\log^3(K)}-3e^{-s/2}$.
\end{thm}
\begin{proof}
  A combination of Prop.~\ref{prop:variable_density_sampling}, Prop.~\ref{prop:RIP}, and Prop.~\ref{prop:sifti local coherence} with \eqref{eq:sifti_final_noise_estimate}, \eqref{eq:sifti error} completes the proof.
\end{proof}

Note that according to~\eqref{eq:p-bound-SI}, the spatial locations are selected (and thus exposed) uniformly at random; while for a fixed spatial location the probability of exposing that location obeys a two-sided power-law decay distribution centered at OPD origin $N_\xi/2$. In addition, since we can ensure $N_{\rm p} K\ll \Nhs = N_{\rm p} N_\xi$ with a low best $K$-term approximation error $\sigma_{K}(\bs{\Psi}^\top\bs{x})_1/\sqrt{K}$ under the hypotheses made on the observed HS volume, the light exposure on the biological specimen can be significantly reduced thanks to the VDS strategy~\eqref{eq:sifti acquisition}; a UDS strategy would require $M \gtrsim \Nhs K \log^3(K) \log(\Nhs)$, which is equivalent to overexposing the specimen.

\section{Constrained-Exposure Coding}    
\label{sec:Sec_05}

As aforementioned, in fluorescence spectroscopy, the type of the injected fluorescent dyes (and therefore their tolerance to light exposure) is known. In this section, assuming that photo-bleaching\footnote{We suppose here that the intensity of the light illuminated by a fluorescent dye is constant during CI/SI-FTI experiments. As a more realistic model, however, the intensity of the fluorescent light would exponentially decrease as a function of the received intensity.} for a fluorescent dye depends only on the total light exposure it has been subject to and supposing that the maximum light exposure~$I_{\max}$ that a fluorescent dye can tolerate is known, we adapt the proposed CI/SI-FTI schemes by ensuring that the total light exposure~$I_{\rm tot}$ on each spatial location of a biological specimen is constant and smaller than~$I_{\max}$, whatever the number of compressive observations.

From the description of Sec.~\ref{sec:FTI continuous model}, in the case of the Nyquist-FTI system, if the total acquisition time is $T>0$, we have $I_{\rm tot} = T \bar I_0 = T \int_0^{+\infty} |\tilde E_0(\nu)|^2\,\ud \nu$. Since $N_\xi$ OPD samples are recorded, each location thus receives an intensity of $I_{\rm opd} = I_{\rm tot}/N_\xi$ per OPD sample. If we fix the value $I_{\rm tot}$, we show below that for the compressive CI/SI-FTI schemes, the light source intensity, and thus the intensity per OPD sample, can be increased. This leads to an improved Measurement-to-Noise Ratio (MNR), \ie 
  \begin{equation}
    \label{eq:MNR-def}
\ts  \textsl{MNR} := 10\log_{10} \frac{\|\bs y\|^2}{\|\bs n\|^2},  
  \end{equation}if the measurement noise is assumed to not depend on the light intensity (\ie under high-photon counting assumption, see Sec.~\ref{sec:noise}).

\subsection{Constrained-exposure CI-FTI}

We want to adjust the intensity $I'_{\rm opd}$ received per second and per unit area by each location of the biological specimen (when illuminated), assuming that the total light exposure of the specimen per unit area is fixed (constrained). 

In CI-FTI, $I_{\rm tot} = M_\xi \tau_\xi I'_{\rm opd}$, with $\tau_{\xi}$ is the constant duration of each time slot fixed by the acquisition. In constrained-exposure CI-FTI, we keep $I_{\rm tot}$ constant so that it matches the total light exposure of the Nyquist FTI where $M_\xi = N_\xi$, \ie $I_{\rm tot} = N_\xi \tau_\xi I_{\rm opd}$. This thus imposes the relation $I'_{\rm opd} = (N_\xi/M_\xi) I_{\rm opd}$. 

The light source intensity can thus be amplified by a factor of $N_\xi/M_\xi$ while still preventing photo-bleaching if $I_{\rm tot} < I_{\max}$; from the assumption of linear intensity scaling made in \eqref{eq:linearity-transfer-intensity}, the intensity of the light outgoing from the biological specimen then undergoes the same amplification. As explained in Sec.~\ref{sec:acquisition-strategy}, we also assume that the additive measurement noise is independent of this increase of intensity. This assumption is experimentally well-verified (see Sec.~\ref{subsec:Simulated experimental data}).

Therefore, under the constant light exposure constraint, the acquisition model of CI-FTI in~\eqref{eq:cifti acquisition} reads
\begin{equation}
\label{eq:cifti ci acquisition}
  \ts \bs{Y}'^{\CI} = \frac{N_\xi}{M_\xi}\bs{F}^*_{\Omega^{\xi}}\bs{X}+ \bs{N}\ =\ \frac{N_\xi}{M_\xi} (\bs{F}^*_{\Omega^{\xi}}\bs{X}+\frac{M_\xi}{N_\xi} \bs N)\ =: \frac{N_\xi}{M_\xi} \tilde{\bs{Y}}^{\CI}.
\end{equation} 
The last equality in \eqref{eq:cifti ci acquisition} thus shows that $\bs{Y}'^{\CI} = N_\xi/M_\xi \tilde{\bs{Y}}^{\CI}$, 
with $\tilde{\bs{Y}}^{\CI}$ being the measurements that would be acquired from~\eqref{eq:cifti acquisition} by attenuating the noise $\bs N$ by a factor of $M_\xi/N_\xi$. 

Therefore, in this constrained-exposure context, we can recover $\bs X$ from \eqref{eq:cifti main reconstruction} by computing $\tilde{\bs{Y}}^{\CI} = (\tilde{\bs y}_1, \cdots, \tilde{\bs y}_{N_{\rm p}}) = (M_\xi/N_\xi) \bs{Y}'^{\CI} $ and replacing $\sigma_\Nyq^2 \leftarrow (M_\xi/N_\xi) \sigma_\Nyq^2$ in the evaluation of $\varepsilon_\CI(s)$ in \eqref{eq:cifti_final_noise_estimate}. In other words, as expected, by fixing the total light exposure, the Measurement-to-Noise Ratio \eqref{eq:MNR-def} is boosted when the number of measurements decreases. This effect is verified numerically in Sec.~\ref{sec:Sec_06}.

Under the conditions of Thm~\ref{thm:cifti}, compared to \eqref{eq:cifti error final}, we also get an improved error bound 
\begin{equation}
    \ts    \|\bs{x}-\hat{\bs{x}}\| \le \frac{2\sqrt{2}}{\sqrt{K_\xi}} \left(\sum_{j=1}^{N_{\rm p}}\left(\sigma_{K_\xi}(\bs{\Psi}_{\rm 1D}^\top\bs{x}_j)_1\right)^{2}\right)^{1/2}+\frac{\sqrt{2N_{\rm p}}M_\xi}{N_\xi}\varepsilon_\CI(s).\label{eq:cifti ci error}
\end{equation}

\subsection{Constrained-exposure SI-FTI}

Compared to CI-FTI, the total light exposure received on each specimen location varies spatially. From the randomness of the structured illumination pattern described in Sec.~\ref{sec:Sec_04}, the $j^{\rm th}$ specimen location is indeed highlighted during $M_j = |\Omega_j|$ time slots, \ie a \rv determined by the pmf $p(k(j,\nell))$ in \eqref{eq:p-bound-SI}.  Despite this variability, we can compute a tight (worst case) upper bound on all $\{M_j: j \in \range{N_{\rm p}}\}$ that holds with arbitrarily high probability. This bound can then be used to constrain the light exposure, \ie to ensure that no specimen location will be over-exposed. 

At each independent random draw of the \rv $\beta$ with pmf $p$, whose $M$ draws populate $\Omega$ in Thm.~\ref{thm:sifti}, the probability of illuminating the $j^\textsl{th}$ specimen location, irrespective of the OPD index, is the marginal pmf $\sum_{\nell \in \range{N_\xi}} p(k(j,l)) = 1/N_{\rm p}$; $M_j$ is thus a Binomial \rv with $M$ trials and success probability $1/N_{\rm p}$. Consequently, $\bb E M_j = M/N_{\rm p}$ and from Bernstein inequality, we have\footnote{Remark that the event $M_j < \frac{M}{N_{\rm p}} - t$, which is useless here, holds with the same probability bound.}
$$
\ts \bb P[ M_j > \frac{M}{N_{\rm p}} + t] \leq \exp(- \frac{1}{2} t^2 (\frac{M}{N_{\rm p}} +\frac{t}{3})^{-1}),
$$
for all $t > 0$. 

Moreover, by union bound and applying the rescaling $t \leftarrow \frac{M}{N_{\rm p}} t$, this concentration is uniform for all spatial locations, \ie
\begin{equation}
\ts \bb P[ M_j > \frac{M}{N_{\rm p}} (1 +t ) ] \leq N_{\rm p}\exp(-\frac{3}{2}\,t^2 \frac{M}{N_{\rm p}} (3 + t)^{-1} ),\quad\textsl{for all}~ j \in \range{N_{\rm p}}. \label{eq:concentration inequality for SIFTI-inter}  
\end{equation}

Note that this holds true despite the dependence of the \rvs $\{M_j: j \in \range{N_{\rm p}}\}$ induced from the relation $\sum_j M_j = M$. Finally, with the change of variable $\zeta = N_{\rm p}\exp(-\frac{3}{2}\,t^2 \frac{M}{N_{\rm p}} (3 + t)^{-1} )$, \ie $t = \tinv{2} (t_0 + \sqrt{t_0^2 + 12 t_0})$ with $t_0 := \frac{2N_{\rm p}}{3M} \log(\frac{N_{\rm p}}{\zeta})$, \eqref{eq:concentration inequality for SIFTI-inter} involves 
\begin{equation}
\ts \bb P[ M_j > \frac{M}{N_{\rm p}} (1 + \frac{t_0 + \sqrt{t_0^2 + 12 t_0}}{2}) ] \leq \zeta,\quad\textsl{for all}~ j \in \range{N_{\rm p}}. \label{eq:concentration inequality for SIFTI}
\end{equation}

Therefore, despite the spatial variability of $M_j$, we can set a low failure probability~$\zeta$ and adjust the light exposure by relying on the fact that 
\begin{equation}
  \label{eq:event-Mj-bound}
\ts M_j < \bar M(\zeta) := \frac{M}{N_{\rm p}} (1 + \frac{t_0 + \sqrt{t_0^2 + 12 t_0}}{2}),  
\end{equation}
for all $j \in \range{N_{\rm p}}$ with probability exceeding $1-\zeta$.

Let $I'_{\rm opd}$ denotes the light intensity received per second and per unit area by each location of the biological specimen when illuminated. The light exposure per unit area on each location $j$ does not exceed $I_{\rm tot} = \bar M I'_{\rm opd} \tau_{\xi}$, and photo-bleaching does not occur if this quantity is smaller than $I_{\max}$. Matching the constrained $I_{\rm tot}$ with the light exposure $N_{\xi} \tau_{\xi} I_{\rm opd}$ of a Nyquist FTI scenario (\ie full OPD sampling), we thus get the light amplification rule
$$
\ts I'_{\rm opd}  = (N_{\xi} / \bar M ) I_{\rm opd} = \frac{\Nhs}{M} (1 + \frac{t_0 + \sqrt{t_0^2 + 12 t_0}}{2})^{-1} I_{\rm opd}.
$$
Thereby, we can increase the intensity of the light source by a factor of $N_{\xi} / \bar M$, for a fixed parameter $\zeta$ independent from other parameters.

From the linear intensity scaling assumption made in \eqref{eq:linearity-transfer-intensity}, and assuming that the measurement noise is independent of this scaling, the acquisition model of constrained-exposure SI-FTI, with respect to~\eqref{eq:sifti acquisition}, then reads
\begin{equation}
  \ts \bs{y}'^{\SI} = (N_{\xi} / \bar M)\,\bs\Phi^*_\Omega\bs{x}+\bs{n} = (N_{\xi} / \bar M ) \big(\bs\Phi^*_\Omega\bs{x}+ (\bar M / N_\xi) \bs{n}\big) := (N_{\xi} / \bar M ) \tilde{\bs{y}}^{\SI}.\label{eq:sifti ci acquisition}
\end{equation}
Similarly to the discussion for constrained-exposure CI-FTI,~\eqref{eq:sifti ci acquisition} means that $\bs{y}'^{\SI} = (N_{\xi} / \bar M )\,\tilde{\bs{y}}^{\SI}$, with $\tilde{\bs{y}}^{\SI}$ being the measurements that would be acquired from~\eqref{eq:sifti acquisition} by attenuating the noise $\bs n$ by a factor of $(\bar M / N_\xi) = (M/\Nhs) (1 + \frac{t_0 + \sqrt{t_0^2 + 12 t_0}}{2})$. 

Finally, by union bound over the events ensuring \eqref{eq:event-Mj-bound} and the statement of Thm~\ref{thm:sifti}, with probability exceeding $1-\Nhs^{-c\log^3(K)}-3e^{-s/2}-\zeta$,
the recovery guarantee for the reconstruction of $\bs x$ from $\tilde{\bs{y}}^{\SI}$ in constrained-exposure SI-FTI model is
\begin{equation}
  \ts  \|\bs{x}-\hat{\bs{x}}\| \le \frac{2\sigma_{K}(\bs{\Psi}^\top\bs{x})_1}{\sqrt{K}}+ \frac{M}{\Nhs} (1 + \frac{t_0 + \sqrt{t_0^2 + 12 t_0}}{2})\,\varepsilon_\SI(s),\label{eq:sifti ci error}
\end{equation}
where $t_0 = t_0(\zeta)$ is defined above, and replacing $\sigma_\Nyq^2 \leftarrow (\bar M / N_\xi) \sigma_\Nyq^2$ in the evaluation of $\varepsilon_\SI(s)$ in \eqref{eq:sifti_final_noise_estimate}. 
\medskip

Note that, the second term of the right-hand side of~\eqref{eq:sifti ci error} in SI-FTI, as well as the second term of the right-hand side of~\eqref{eq:cifti ci error} in CI-FTI, increase linearly with the number of compressive FTI measurements. Thus, given the constrained-exposure budget, recording full FTI measurements does not yield the best recovery quality; the optimum number of measurements is the outcome of a trade-off between the best $K$-term approximation of the HS volume and the noise level.

\section{Numerical Results}                      
\label{sec:Sec_06} 

We conduct three simulations to verify the performance of the proposed compressive FTI methods. In the first scenario we examine the optimality of the proposed VDS strategies for CI-FTI and SI-FTI in~\eqref{eq:pmf-ci-FTI} and~\eqref{eq:p-bound-SI} for any arbitrary sparse HS data by tracing the phase transition curves of successful recovery in a noiseless setting. The second scenario consists in following up the performance of the proposed FTI frameworks with constrained/unconstrained-exposure budget on a simulated biological HS volumes. In the third setup, constrained-exposure CI-FTI is simulated from the real FTI measurements.
For all the experiments we report the Reconstruction SNR (RSNR), \ie 
$$
\ts {\rm RSNR} := -10 \log_{10} \bb{E}_{\rm e} {\|\bs{x}-\hat{\bs{x}}\|^2}/{\|\bs{x}\|^2},
$$ 
where $\bb{E}_{\rm e}$ denotes the empirical mean over several trials of the sensing context (as specified in the text). The HS volume reconstructions \eqref{eq:cifti main reconstruction} and \eqref{eq:sifti main reconstruction} are numerically performed with the SPGL1 solver~\cite{spgl1:2007,vandenBerg:2008}. 

Concerning the value $\varepsilon_\CI$ and $\varepsilon_\SI$ of the noise power in the optimization problems~\eqref{eq:cifti main reconstruction} and~\eqref{eq:sifti main reconstruction}, we observed numerically that the estimations in \eqref{eq:cifti_final_noise_estimate} and \eqref{eq:sifti_final_noise_estimate} are not tight enough at small number of measurements. While tightening those bounds constitutes an interesting future work, we have rather numerically estimated $\varepsilon_\CI$ and $\varepsilon_\SI$ in our simulations. For this, we have computed the empirical $95^\textsl{th}$ percentile curve of the weighted noise power over 100 Monte-Carlo realizations of both a Gaussian random noise (with unit variance) and the index set $\Omega$ (for the selected sensing scenario, CI- or SI-FTI) over the considered interval of measurement number. This experimental curve has then be multiplied by the noise standard deviation, known (for synthetic examples) or estimated (for experimental data).

\subsection{Impact of different VDS strategies in compressive FTI}
\label{subsec:Noiseless synthetic HS volume}

We here challenge the VDS densities defined in Thm.~\ref{thm:cifti} and Thm.~\ref{thm:sifti} by comparing their performances (in terms of reconstruction error) with those reached by other density functions with different decaying power laws in the OPD domain. We also show that the prescribed schemes achieve near-optimal results with respect to a pmf set to the value of the local coherence $\bs \mu^\textsl{loc}$ between the sparsity and the sensing bases. 

In details, following Sec.~\ref{sec:acquisition-strategy} and Sec.~\ref{sec:acquisition-strategy-subsec}, the compressive FTI scenarios are associated with different subsampling of measurement indices, as recorded in the measurement index set $\Omega$. In particular, the (possibly non-unique) random elements of $\Omega$ are \iid according to the following pmfs generalizing \eqref{eq:pmf-ci-FTI} and~\eqref{eq:p-bound-SI} to variable power laws:
\begin{equation}
\label{eq:vds_alpha}
\scalebox{.8}{$\begin{array}{|c|l|}
\hline
&\\[-3mm]
\ts {\rm (CI-FTI)}&\ts p^\alpha(\nell):=C_{N_\xi,\alpha} \min\big\{1,\frac{1}{|\nell-N_\xi/2|^\alpha}\big\},\\[2mm]
&\ts p^\Opt(\nell) := \big(\sum_{\nell'}\mu^\textsl{loc}_{\nell'}(\bs F^* \bs \Psi_\textsl{1D})^2\big)^{-1} (\mu^\textsl{loc}_{\nell}(\bs F^* \bs \Psi_\textsl{1D}))^2, \\[2mm]
\hline
&\\[-3mm]
\ts {\rm (SI-FTI)}& p^\alpha(k(j,\nell)) :=\frac{C_{N_\xi,\alpha}}{N_{\rm p}} \min\big\{1,\frac{1}{|\nell-N_\xi/2|^\alpha}\big\}, \\[2mm]
&\ts p^\Opt(k(j,\nell)) := \big(\sum_{k} \mu^\textsl{loc}_{k}(\bs \Phi^* \bs \Psi)^2\big)^{-1} (\mu^\textsl{loc}_{k}(\bs \Phi^* \bs \Psi))^2,\\[1mm]
\hline
\end{array}$}
\end{equation}
where $C_{N_\xi, \alpha}$ ensures that the probabilities sum to one, and the specific pmf parameterizations are defined in Sec.~\ref{sec:acquisition-strategy} and Sec.~\ref{sec:acquisition-strategy-subsec}. In~\eqref{eq:vds_alpha}, the parameter $\alpha \in \{0,1,1.5,2,8\}$ controls the decaying power of the VDS strategy; for $\alpha=0$ it reduces to a UDS strategy and for $\alpha=1$ it matches the pmfs of~\eqref{eq:pmf-ci-FTI} and~\eqref{eq:p-bound-SI}. We also consider optimal VDS strategies, \ie the two pmfs $p^{\Opt}$ in~\eqref{eq:vds_alpha}, where each pmf --- for CI and SI-FTI --- is set to the local coherence between the sensing and the sparsity bases. While this setting is numerically computable, it is, however, hardly integrable to our estimation of the noise power, and thus to our analysis of the reconstruction error of the HS volume. 

\begin{figure}[tbh]
  \centering
  \input{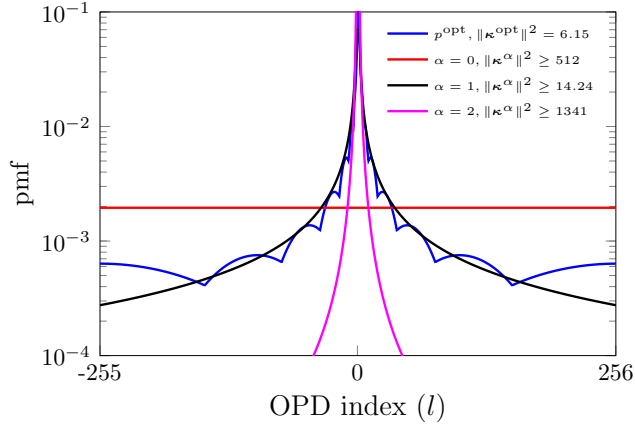}
  \caption{Representation of different CI-FTI sampling strategies for $N_\xi = 512$ associated with pmfs $p^\alpha(\nell)$ and $p^{\rm opt}$ in~\eqref{eq:vds_alpha}. The lower value of $\|\bs \kappa^\alpha\|^2$ is an indicator of sampling optimality. Smaller values amount to a tighter sample complexity bound in~\eqref{eq:rip_for_onb_vds}. The same trend is expected for the phase transition curves as well, which is confirmed in Fig.~\ref{fig:sim1_res1}.}
  \label{fig:pmf}
\end{figure}

Note that in the case of CI-FTI, from Prop.~\ref{prop:variable_density_sampling}, the pmf $p^\alpha$ in \eqref{eq:vds_alpha} is admissible --- \ie it allows for HS volume reconstruction in CI-FTI --- only if it is proportional to a vector $\bs \kappa^\alpha \in \bb R^{N_\xi}_+$ that bounds the local coherence $\mu^{\rm loc}_{\nell}(\bs F^* \bs \Psi_{\rm 1D})$  between $\bs F$ and $\bs \Psi_{\rm 1D}$. Thus, we must have $\|\bs \kappa^\alpha\|^2 p^\alpha(l) = (\kappa^\alpha_\nell)^2 \geq \mu^{\rm loc}_{\nell}(\bs F^* \bs \Psi_{\rm 1D})$ for all $\nell \in \range{N_\xi}$. This involves in particular $\|\bs \kappa^\alpha\|^2 p^\alpha(N_\xi/2) = \|\bs \kappa^\alpha\|^2 C_{N_\xi, \alpha} \geq 1$, since $\mu^{\rm loc}_{N_\xi/2}(\bs F^* \bs \Psi_{\rm 1D}) = 1$ (see App.~\ref{sec:proof sifti}), \ie we necessarily have $\|\bs \kappa^\alpha\|^2 \geq 1/C_{N_\xi, \alpha} = \sum_\nell \min\{1, \frac{1}{|\nell - N_\xi/2|^\alpha}\}$.  In the case of $p^{\rm opt}$, we can directly set $\kappa_l^{\rm opt} = \mu^{\rm loc}_{\nell}(\bs F^* \bs \Psi_{\rm 1D})$ by construction.

We provide an illustration of the corresponding pmfs in Fig.~\ref{fig:pmf} for CI-FTI.
It is clear that the pmf for $\alpha=1$ is close to the optimal pmf curve $p^{\rm opt}$. This is in agreement with the bound $\|\bs\kappa^{1}\|^2 \geq 14.24$, \ie a value closer to $\|\bs\kappa^{\rm opt}\|^2 \simeq 6.15$ than the lower bounds of $\|\bs\kappa^{\alpha}\|^2$ for the other values of $\alpha$ (see Fig.~\ref{fig:pmf}).
\medskip

The CI/SI-FTI measurements are simulated by restricting the Nyquist measurements to a subset $\Omega$ specified by the pmfs in \eqref{eq:vds_alpha} and the procedures defined in Sec.~\ref{sec:acquisition-strategy} and Sec.~\ref{sec:acquisition-strategy-subsec}. At each trial of our experiments, we generate randomly both $\Omega$ and a synthetic HS volume from $\bs{X} = \bs{\Psi}\bs{S}$, where $\bs{\Psi}=\bs{\Psi}_{{\rm 2D}}\otimes\bs{\Psi}_{\rm 1D}$ for SI-FTI, and $\bs{\Psi}=\bs{I}_{N_{\rm p}}\otimes\bs{\Psi}_{\rm 1D}$ for CI-FTI. The matrix $\bs S \in \bb R^{N_\xi \times N_{\rm p}}$ is sparse, its dimensions are $N_\xi = 512$ and $N_{\rm p} = 8^2$ (\ie $\Nhs = 2^{15}$), and its row and column sparsity levels are $K_\xi = 4$ and $K_p = 4$, respectively. The indices of the $K_\xi$ rows and $K_p$ columns are chosen uniformly at random and the non-zero random coefficients follow a normal distribution. We finally simulate noiseless Nyquist-FTI measurements of $\bs x = \ve(\bs X)$ from $\bs{y}^{\Nyq} = \bs{\Phi}^*\bs{x}$, where $\bs{\Phi}:=\bs{I}_{N_{\rm p}}\otimes\bs{F}$. 

For a fixed measurement ratio $M/\Nhs$ this procedure is repeated for 50 independent trials. We trace the phase transition curves in Fig.~\ref{fig:sim1_res1} which shows the probability of successful recovery against the number of measurements by solving \eqref{eq:cifti main reconstruction} and \eqref{eq:sifti main reconstruction} with a zero noise level (\ie $\varepsilon^\CI = \varepsilon^\SI = 0$). We count a recovery as successful if $\|\hat{\bs{x}}-\bs{x}\|^2\le 10^{-10}\|\bs{x}\|^2$. The plot confirms that the value $\alpha=1$ (\ie the proposed VDS strategy) is very close to the optimal sampling strategy $p^{\rm opt}$ and can reach 100\,\% chance of successful recovery from $M/\Nhs>0.2$ and $M/\Nhs>0.5$ for SI-FTI and CI-FTI, respectively; while for UDS strategy ($\alpha=0$) we observe that this cannot be achieved even when $M/N_{\rm hs}=1$. Hereafter, the rest of our experiments are restricted to the case $\alpha = 1$.

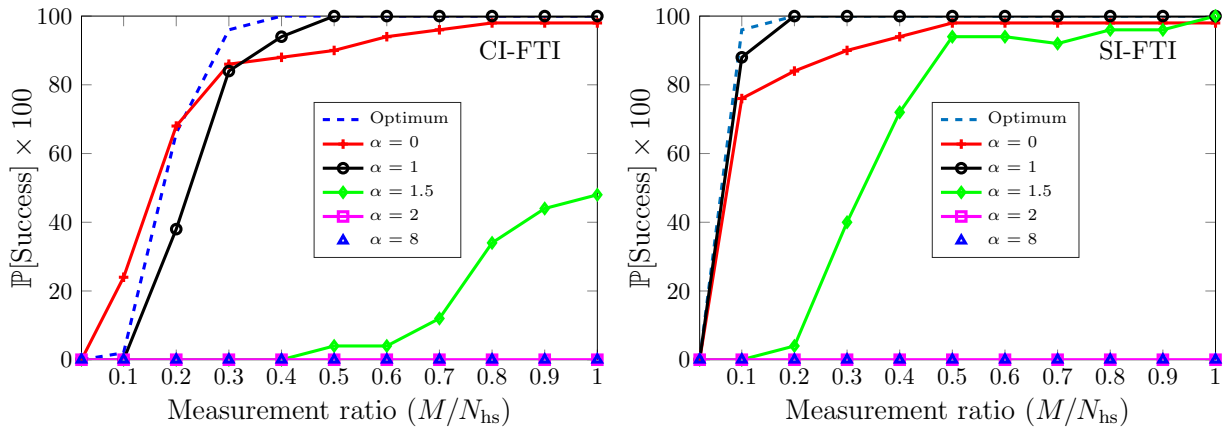
\begin{figure}[tbh]
  \begin{minipage}{0.49\linewidth}
    \centering
%
%
\definecolor{mycolor1}{rgb}{1.00000,0.00000,1.00000}%
\begin{tikzpicture}[scale = 0.9]
\begin{axis}[%
width=3in,
height=2in,
at={(0.762in,0.486in)},
scale only axis,
xmin=0.02,
xmax=1,
xlabel={$\text{Measurement ratio}~(M/N_{\rm hs})$},
xtick={0.1,0.2,0.3,0.4,0.5,0.6,0.7,0.8,0.9,1},
xticklabels={0.1,0.2,0.3,0.4,0.5,0.6,0.7,0.8,0.9,1},
xlabel style={at = {(0.5,0.02)},font=\fontsize{12}{1}\selectfont},
ticklabel style={font=\fontsize{10}{1}\selectfont},
ymin=0,
ymax=100,
ytick={0,20,40,60,80,100},
yticklabels={0,20,40,60,80,100},
ylabel={$\bb{P}[\text{Success}] \times 100$},
ylabel style={at = {(0.05,0.5)},font=\fontsize{12}{1}},
axis background/.style={fill=white},
legend style={at={(0.45,0.75)},anchor=north west,legend cell align=left,align=left,draw=white!15!black,font=\fontsize{7}{1}\selectfont}
]
\node[text=black, draw=none] at (rel axis cs:0.85,0.9) {CI-FTI}; 
\addplot [color=blue,dashed,line width=1.3pt]
  table[row sep=crcr]{%
0.02	0\\
0.1	2\\
0.2	66\\
0.3	96\\
0.4	100\\
0.5	100\\
0.6	100\\
0.7	100\\
0.8	100\\
0.9	100\\
1	100\\
};
\addlegendentry{$\text{Optimum}$};

\addplot [color=red,solid,line width=1.3pt,mark=+,mark options={solid}]
  table[row sep=crcr]{%
0.02	0\\
0.1	24\\
0.2	68\\
0.3	86\\
0.4	88\\
0.5	90\\
0.6	94\\
0.7	96\\
0.8	98\\
0.9	98\\
1	98\\
};
\addlegendentry{$\alpha\text{ = 0}$};

\addplot [color=black,solid,line width=1.3pt,mark=o,mark options={solid}]
  table[row sep=crcr]{%
0.02	0\\
0.1	0\\
0.2	38\\
0.3	84\\
0.4	94\\
0.5	100\\
0.6	100\\
0.7	100\\
0.8	100\\
0.9	100\\
1	100\\
};
\addlegendentry{$\alpha\text{ = 1}$};

\addplot [color=green,solid,line width=1.3pt,mark=diamond,mark options={solid}]
  table[row sep=crcr]{%
0.02	0\\
0.1	0\\
0.2	0\\
0.3	0\\
0.4	0\\
0.5	4\\
0.6	4\\
0.7	12\\
0.8	34\\
0.9	44\\
1	48\\
};
\addlegendentry{$\alpha\text{ = 1.5}$};

\addplot [color=mycolor1,solid,line width=1.3pt,mark=square,mark options={solid}]
  table[row sep=crcr]{%
0.02	0\\
0.1	0\\
0.2	0\\
0.3	0\\
0.4	0\\
0.5	0\\
0.6	0\\
0.7	0\\
0.8	0\\
0.9	0\\
1	0\\
};
\addlegendentry{$\alpha\text{ = 2}$};

\addplot [color=blue,line width=1.3pt,only marks,mark=triangle,mark options={solid}]
  table[row sep=crcr]{%
0.02	0\\
0.1	0\\
0.2	0\\
0.3	0\\
0.4	0\\
0.5	0\\
0.6	0\\
0.7	0\\
0.8	0\\
0.9	0\\
1	0\\
};
\addlegendentry{$\alpha\text{ = 8}$};

\end{axis}
\end{tikzpicture}%

  \end{minipage}
  \begin{minipage}{0.49\linewidth}
    \centering
%
%
\definecolor{mycolor1}{rgb}{0.00000,0.44700,0.74100}%
\definecolor{mycolor2}{rgb}{1.00000,0.00000,1.00000}%
\begin{tikzpicture}[scale = 0.9]
\begin{axis}[%
width=3in,
height=2in,
at={(0.762in,0.486in)},
scale only axis,
xmin=0.02,
xmax=1,
xlabel={$\text{Measurement ratio}~(M/N_{\rm hs})$},
xtick={0.1,0.2,0.3,0.4,0.5,0.6,0.7,0.8,0.9,1},
xticklabels={0.1,0.2,0.3,0.4,0.5,0.6,0.7,0.8,0.9,1},
ticklabel style={font=\fontsize{10}{1}\selectfont},
xlabel style={at = {(0.5,0.02)},font=\fontsize{12}{1}\selectfont},
ymin=0,
ymax=100,
ytick={0,20,40,60,80,100},
yticklabels={0,20,40,60,80,100},
ylabel={$\bb{P}[\text{Success}] \times 100$},
ylabel style={at = {(0.05,0.5)},font=\fontsize{12}{1}},
axis background/.style={fill=white},
legend style={at={(0.45,0.75)},anchor=north west,legend cell align=left,align=left,draw=white!15!black,font=\fontsize{7}{1}\selectfont}
]
\node[text=black, draw=none] at (rel axis cs:0.85,0.9) {SI-FTI}; 
\addplot [color=mycolor1,dashed,line width=1.3pt]
  table[row sep=crcr]{%
0.02	0\\
0.1	96\\
0.2	100\\
0.3	100\\
0.4	100\\
0.5	100\\
0.6	100\\
0.7	100\\
0.8	100\\
0.9	100\\
1	100\\
};
\addlegendentry{$\text{Optimum}$};

\addplot [color=red,solid,line width=1.3pt,mark=+,mark options={solid}]
  table[row sep=crcr]{%
0.02	0\\
0.1	76\\
0.2	84\\
0.3	90\\
0.4	94\\
0.5	98\\
0.6	98\\
0.7	98\\
0.8	98\\
0.9	98\\
1	98\\
};
\addlegendentry{$\alpha\text{ = 0}$};

\addplot [color=black,solid,line width=1.3pt,mark=o,mark options={solid}]
  table[row sep=crcr]{%
0.02	0\\
0.1	88\\
0.2	100\\
0.3	100\\
0.4	100\\
0.5	100\\
0.6	100\\
0.7	100\\
0.8	100\\
0.9	100\\
1	100\\
};
\addlegendentry{$\alpha\text{ = 1}$};

\addplot [color=green,solid,line width=1.3pt,mark=diamond,mark options={solid}]
  table[row sep=crcr]{%
0.02	0\\
0.1	0\\
0.2	4\\
0.3	40\\
0.4	72\\
0.5	94\\
0.6	94\\
0.7	92\\
0.8	96\\
0.9	96\\
1	100\\
};
\addlegendentry{$\alpha\text{ = 1.5}$};

\addplot [color=mycolor2,solid,line width=1.3pt,mark=square,mark options={solid}]
  table[row sep=crcr]{%
0.02	0\\
0.1	0\\
0.2	0\\
0.3	0\\
0.4	0\\
0.5	0\\
0.6	0\\
0.7	0\\
0.8	0\\
0.9	0\\
1	0\\
};
\addlegendentry{$\alpha\text{ = 2}$};

\addplot [color=blue,line width=1.3pt,only marks,mark=triangle,mark options={solid}]
  table[row sep=crcr]{%
0.02	0\\
0.1	0\\
0.2	0\\
0.3	0\\
0.4	0\\
0.5	0\\
0.6	0\\
0.7	0\\
0.8	0\\
0.9	0\\
1	0\\
};
\addlegendentry{$\alpha\text{ = 8}$};

\end{axis}
\end{tikzpicture}%
  \end{minipage}
  \caption{Phase transition comparison of proposed VDS ($\alpha =1$) against UDS ($\alpha =0$) and other VDS strategies in the frameworks of CI-FTI and SI-FTI.}
  \label{fig:sim1_res1}
\end{figure}
\begin{figure}[tbh]
  \centering\noindent
  \centering
  \raisebox{5mm}{
%
%
\begin{tikzpicture}[scale = 0.3]
\centering
\begin{axis}[%
width=2.5in,
height=2.5in,
at={(0in,0in)},
scale only axis,
axis on top,
xmin=0.5,
xmax=64.5,
tick align=outside,
y dir=reverse,
ymin=0.5,
ymax=64.5,
hide axis
]
\addplot [forget plot] graphics [xmin=0.5,xmax=64.5,ymin=0.5,ymax=64.5] {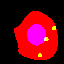};
\node[text=white, draw=none] at (rel axis cs:0.5,0.9) {\fontsize{18}{1}\selectfont Ground truth RGB}; 
\end{axis}
\end{tikzpicture}%
} 
  \raisebox{5mm}{
%
%
\begin{tikzpicture}[scale = 0.3]
\centering
\begin{axis}[%
width=7.5in,
height=2.5in,
at={(0in,0in)},
scale only axis,
point meta min=1,
point meta max=384,
axis on top,
xmin=0.5,
xmax=192.5,
y dir=reverse,
ymin=0.5,
ymax=64.5,
hide axis
]
\addplot [forget plot] graphics [xmin=0.5,xmax=192.5,ymin=0.5,ymax=64.5] {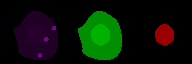};
\node[text=white, draw=none] at (rel axis cs:0.166,0.9) {\fontsize{18}{1}\selectfont $\bs{\nell = 72}$}; 
\node[text=white, draw=none] at (rel axis cs:0.5,0.9) {\fontsize{18}{1}\selectfont $\bs{\nell = 79}$}; 
\node[text=white, draw=none] at (rel axis cs:0.8333,0.9) {\fontsize{18}{1}\selectfont $\bs{\nell = 96}$}; 
\end{axis}
\end{tikzpicture}%
}\hspace{3mm}
%
%
\definecolor{mycolor1}{rgb}{0.00000,0.44700,0.74100}%
\definecolor{mycolor2}{rgb}{0.85000,0.32500,0.09800}%
\definecolor{mycolor3}{rgb}{0.92900,0.69400,0.12500}%
\begin{tikzpicture}[scale = 0.3]
\centering
\begin{axis}[
width=7.5in,
height=2.5in,
at={(0in,0in)},
scale only axis,
xmin=1,
xmax=256,
xlabel={$\text{Wavenumber index (}\nell\text{)}$},
xtick={1,64,128,196,256},
xticklabels={1,64,128,196,256},
ticklabel style={font=\fontsize{22}{1}\selectfont},
ymin=-2,
ymax=102,
ylabel={$\text{Intensity}$},
ytick={0,100},
yticklabels={0,100},
xlabel style={at = {(0.5,-0.05)},font=\fontsize{30}{1}},ylabel style={at = {(0.01,0.5)},font=\fontsize{30}{1}},
axis background/.style={fill=white},
legend style={at={(1,0.98)},anchor=north east,legend cell align=left,align=left,fill=none,draw=none,font=\fontsize{25}{10}\selectfont}
]
\addplot [color=red,solid,line width=2.0pt]
  table[row sep=crcr]{%
1	0\\
2	0\\
3	0\\
4	0\\
5	0\\
6	0\\
7	0\\
8	0\\
9	0\\
10	0\\
11	0\\
12	0\\
13	0\\
14	0\\
15	0\\
16	0\\
17	0\\
18	0\\
19	0\\
20	0\\
21	0\\
22	0\\
23	0\\
24	0\\
25	0\\
26	0\\
27	0\\
28	0\\
29	0\\
30	0\\
31	0\\
32	0\\
33	0\\
34	0\\
35	0\\
36	0\\
37	0\\
38	0\\
39	0\\
40	0\\
41	0\\
42	0\\
43	0\\
44	0\\
45	0\\
46	0\\
47	0\\
48	0\\
49	0\\
50	0\\
51	0\\
52	0\\
53	0\\
54	0\\
55	0\\
56	0.425123840097398\\
57	1.01322132846114\\
58	1.43240469374638\\
59	1.92176290736576\\
60	2.49037897780377\\
61	3.06719923160167\\
62	4.05942633785257\\
63	5.76046766692309\\
64	9.0228210011953\\
65	13.6480579311096\\
66	17.3868053430139\\
67	19.8974910653507\\
68	22.0714950020302\\
69	27.7038877110484\\
70	38.3794873665128\\
71	62.3071575685072\\
72	100\\
73	99.7216816026445\\
74	47.2526528306809\\
75	16.5896717716924\\
76	5.70358052529219\\
77	2.19604585415567\\
78	0.8161982210223\\
79	0.328603355789523\\
80	0.131726016361397\\
81	0\\
82	0\\
83	0\\
84	0\\
85	0\\
86	0\\
87	0\\
88	0\\
89	0\\
90	0\\
91	0\\
92	0\\
93	0\\
94	0\\
95	0\\
96	0\\
97	0\\
98	0\\
99	0\\
100	0\\
101	0\\
102	0\\
103	0\\
104	0\\
105	0\\
106	0\\
107	0\\
108	0\\
109	0\\
110	0\\
111	0\\
112	0\\
113	0\\
114	0\\
115	0\\
116	0\\
117	0\\
118	0\\
119	0\\
120	0\\
121	0\\
122	0\\
123	0\\
124	0\\
125	0\\
126	0\\
127	0\\
128	0\\
129	0\\
130	0\\
131	0\\
132	0\\
133	0\\
134	0\\
135	0\\
136	0\\
137	0\\
138	0\\
139	0\\
140	0\\
141	0\\
142	0\\
143	0\\
144	0\\
145	0\\
146	0\\
147	0\\
148	0\\
149	0\\
150	0\\
151	0\\
152	0\\
153	0\\
154	0\\
155	0\\
156	0\\
157	0\\
158	0\\
159	0\\
160	0\\
161	0\\
162	0\\
163	0\\
164	0\\
165	0\\
166	0\\
167	0\\
168	0\\
169	0\\
170	0\\
171	0\\
172	0\\
173	0\\
174	0\\
175	0\\
176	0\\
177	0\\
178	0\\
179	0\\
180	0\\
181	0\\
182	0\\
183	0\\
184	0\\
185	0\\
186	0\\
187	0\\
188	0\\
189	0\\
190	0\\
191	0\\
192	0\\
193	0\\
194	0\\
195	0\\
196	0\\
197	0\\
198	0\\
199	0\\
200	0\\
201	0\\
202	0\\
203	0\\
204	0\\
205	0\\
206	0\\
207	0\\
208	0\\
209	0\\
210	0\\
211	0\\
212	0\\
213	0\\
214	0\\
215	0\\
216	0\\
217	0\\
218	0\\
219	0\\
220	0\\
221	0\\
222	0\\
223	0\\
224	0\\
225	0\\
226	0\\
227	0\\
228	0\\
229	0\\
230	0\\
231	0\\
232	0\\
233	0\\
234	0\\
235	0\\
236	0\\
237	0\\
238	0\\
239	0\\
240	0\\
241	0\\
242	0\\
243	0\\
244	0\\
245	0\\
246	0\\
247	0\\
248	0\\
249	0\\
250	0\\
251	0\\
252	0\\
253	0\\
254	0\\
255	0\\
256	0\\
};
\addlegendentry{\ R-PE (R-phycoerythrin)};

\addplot [color=green,solid,line width=2.0pt]
  table[row sep=crcr]{%
1	0\\
2	0\\
3	0\\
4	0\\
5	0\\
6	0\\
7	0\\
8	0\\
9	0\\
10	0\\
11	0\\
12	0\\
13	0\\
14	0\\
15	0\\
16	0\\
17	0\\
18	0\\
19	0\\
20	0\\
21	0\\
22	0\\
23	0\\
24	0\\
25	0\\
26	0\\
27	0\\
28	0\\
29	0\\
30	0\\
31	0\\
32	0\\
33	0\\
34	0\\
35	0\\
36	0\\
37	0\\
38	0\\
39	0\\
40	0\\
41	0\\
42	0\\
43	0\\
44	0\\
45	0\\
46	0\\
47	0\\
48	0\\
49	0\\
50	0\\
51	0\\
52	0\\
53	0\\
54	0\\
55	0\\
56	0\\
57	0\\
58	0\\
59	0\\
60	0\\
61	0\\
62	0\\
63	0\\
64	3.56436398700623\\
65	4.71793563336539\\
66	6.21968508989662\\
67	8.10283744740463\\
68	10.3797515327807\\
69	13.6129735250634\\
70	17.8965026542928\\
71	23.1972632121268\\
72	29.8445506581177\\
73	36.6165640659674\\
74	43.6186552430398\\
75	49.8499252645555\\
76	57.187936006264\\
77	67.3658510808578\\
78	80.3651424563954\\
79	93.3376595568273\\
80	100\\
81	92.3602658880212\\
82	68.9427660780757\\
83	41.0519785478685\\
84	18.5163272594056\\
85	7.16101544040108\\
86	2.36212019398192\\
87	0.793372866248487\\
88	0.431054487419222\\
89	0\\
90	0\\
91	0\\
92	0\\
93	0\\
94	0\\
95	0\\
96	0\\
97	0\\
98	0\\
99	0\\
100	0\\
101	0\\
102	0\\
103	0\\
104	0\\
105	0\\
106	0\\
107	0\\
108	0\\
109	0\\
110	0\\
111	0\\
112	0\\
113	0\\
114	0\\
115	0\\
116	0\\
117	0\\
118	0\\
119	0\\
120	0\\
121	0\\
122	0\\
123	0\\
124	0\\
125	0\\
126	0\\
127	0\\
128	0\\
129	0\\
130	0\\
131	0\\
132	0\\
133	0\\
134	0\\
135	0\\
136	0\\
137	0\\
138	0\\
139	0\\
140	0\\
141	0\\
142	0\\
143	0\\
144	0\\
145	0\\
146	0\\
147	0\\
148	0\\
149	0\\
150	0\\
151	0\\
152	0\\
153	0\\
154	0\\
155	0\\
156	0\\
157	0\\
158	0\\
159	0\\
160	0\\
161	0\\
162	0\\
163	0\\
164	0\\
165	0\\
166	0\\
167	0\\
168	0\\
169	0\\
170	0\\
171	0\\
172	0\\
173	0\\
174	0\\
175	0\\
176	0\\
177	0\\
178	0\\
179	0\\
180	0\\
181	0\\
182	0\\
183	0\\
184	0\\
185	0\\
186	0\\
187	0\\
188	0\\
189	0\\
190	0\\
191	0\\
192	0\\
193	0\\
194	0\\
195	0\\
196	0\\
197	0\\
198	0\\
199	0\\
200	0\\
201	0\\
202	0\\
203	0\\
204	0\\
205	0\\
206	0\\
207	0\\
208	0\\
209	0\\
210	0\\
211	0\\
212	0\\
213	0\\
214	0\\
215	0\\
216	0\\
217	0\\
218	0\\
219	0\\
220	0\\
221	0\\
222	0\\
223	0\\
224	0\\
225	0\\
226	0\\
227	0\\
228	0\\
229	0\\
230	0\\
231	0\\
232	0\\
233	0\\
234	0\\
235	0\\
236	0\\
237	0\\
238	0\\
239	0\\
240	0\\
241	0\\
242	0\\
243	0\\
244	0\\
245	0\\
246	0\\
247	0\\
248	0\\
249	0\\
250	0\\
251	0\\
252	0\\
253	0\\
254	0\\
255	0\\
256	0\\
};
\addlegendentry{\ Acridine orange};

\addplot [color=blue,solid,line width=2.0pt]
  table[row sep=crcr]{%
1	0\\
2	0\\
3	0\\
4	0\\
5	0\\
6	0\\
7	0\\
8	0\\
9	0\\
10	0\\
11	0\\
12	0\\
13	0\\
14	0\\
15	0\\
16	0\\
17	0\\
18	0\\
19	0\\
20	0\\
21	0\\
22	0\\
23	0\\
24	0\\
25	0\\
26	0\\
27	0\\
28	0\\
29	0\\
30	0\\
31	0\\
32	0\\
33	0\\
34	0\\
35	0\\
36	0\\
37	0\\
38	0\\
39	0\\
40	0\\
41	0\\
42	0\\
43	0\\
44	0\\
45	0\\
46	0\\
47	0\\
48	0\\
49	0\\
50	0\\
51	0\\
52	0\\
53	0\\
54	0\\
55	0\\
56	0\\
57	0\\
58	0\\
59	0\\
60	0\\
61	0\\
62	0\\
63	0\\
64	0\\
65	0\\
66	0\\
67	0\\
68	0\\
69	0\\
70	4.99148605175311\\
71	5.88630024495692\\
72	6.91390729429339\\
73	8.32812217454069\\
74	9.81896462230744\\
75	11.4466130001154\\
76	13.8979698304278\\
77	16.74794054043\\
78	19.5019547454361\\
79	22.1504297211369\\
80	25.3943703603723\\
81	28.7024739125831\\
82	31.5726640936457\\
83	35.0715848058799\\
84	38.9463545563847\\
85	43.0380526072911\\
86	46.773340110578\\
87	52.7067843754531\\
88	58.8285474053792\\
89	65.2377588631634\\
90	71.8967075062441\\
91	78.8404900055439\\
92	84.924867282287\\
93	90.2651954740089\\
94	94.6641522723048\\
95	97.5664475937293\\
96	99.6443820577769\\
97	100\\
98	99.0297467420538\\
99	96.1936527081008\\
100	92.26017238648\\
101	86.9530558811076\\
102	80.7340933529192\\
103	73.6319358819381\\
104	65.9289892779263\\
105	57.6817709129762\\
106	49.4431809172146\\
107	41.8696023749508\\
108	34.287474275664\\
109	26.5160119942333\\
110	19.2999390460755\\
111	13.3532295164691\\
112	8.47106244123438\\
113	5.19021408463166\\
114	3.17646389826742\\
115	2.00078332682424\\
116	1.37581883938494\\
117	1.08398814707413\\
118	0\\
119	0\\
120	0\\
121	0\\
122	0\\
123	0\\
124	0\\
125	0\\
126	0\\
127	0\\
128	0\\
129	0\\
130	0\\
131	0\\
132	0\\
133	0\\
134	0\\
135	0\\
136	0\\
137	0\\
138	0\\
139	0\\
140	0\\
141	0\\
142	0\\
143	0\\
144	0\\
145	0\\
146	0\\
147	0\\
148	0\\
149	0\\
150	0\\
151	0\\
152	0\\
153	0\\
154	0\\
155	0\\
156	0\\
157	0\\
158	0\\
159	0\\
160	0\\
161	0\\
162	0\\
163	0\\
164	0\\
165	0\\
166	0\\
167	0\\
168	0\\
169	0\\
170	0\\
171	0\\
172	0\\
173	0\\
174	0\\
175	0\\
176	0\\
177	0\\
178	0\\
179	0\\
180	0\\
181	0\\
182	0\\
183	0\\
184	0\\
185	0\\
186	0\\
187	0\\
188	0\\
189	0\\
190	0\\
191	0\\
192	0\\
193	0\\
194	0\\
195	0\\
196	0\\
197	0\\
198	0\\
199	0\\
200	0\\
201	0\\
202	0\\
203	0\\
204	0\\
205	0\\
206	0\\
207	0\\
208	0\\
209	0\\
210	0\\
211	0\\
212	0\\
213	0\\
214	0\\
215	0\\
216	0\\
217	0\\
218	0\\
219	0\\
220	0\\
221	0\\
222	0\\
223	0\\
224	0\\
225	0\\
226	0\\
227	0\\
228	0\\
229	0\\
230	0\\
231	0\\
232	0\\
233	0\\
234	0\\
235	0\\
236	0\\
237	0\\
238	0\\
239	0\\
240	0\\
241	0\\
242	0\\
243	0\\
244	0\\
245	0\\
246	0\\
247	0\\
248	0\\
249	0\\
250	0\\
251	0\\
252	0\\
253	0\\
254	0\\
255	0\\
256	0\\
};
\addlegendentry{\ TetraSpeck blue dye};

\end{axis}
\end{tikzpicture}%
  \caption{A synthetic biological RGB image (left); three spectral bands of the generated ground truth HS volume (middle); the known spectral signatures of three fluorochromes (right).}
  \label{fig:gt}
\end{figure}
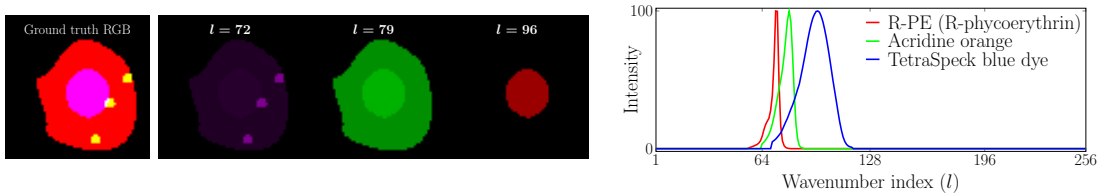
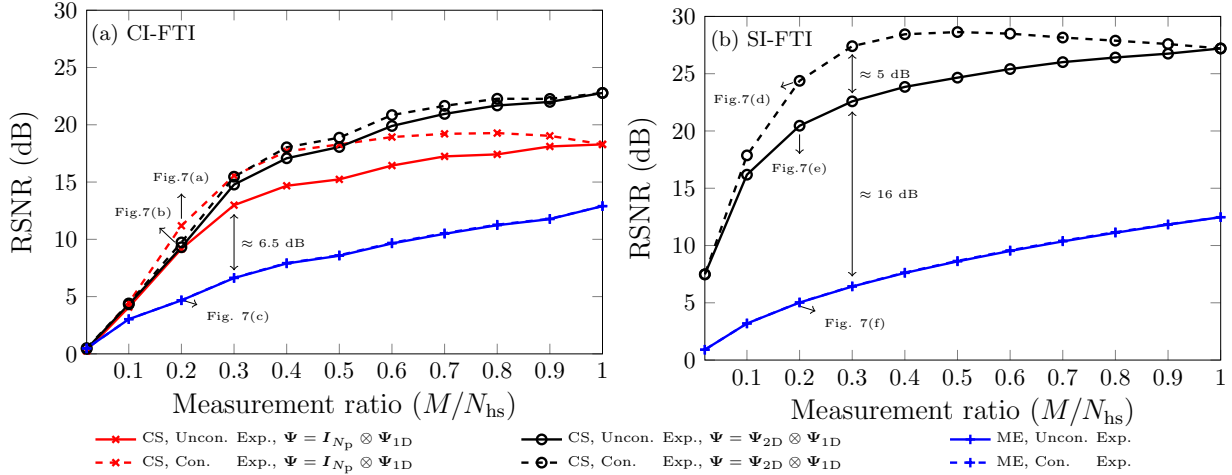
\begin{figure}[tbh]
  \begin{minipage}{0.49\linewidth}
    \centering
    \scalebox{0.9}{
%
%
\begin{tikzpicture}[scale = 1]
\begin{axis}[%
width=3in,
height=2in,
at={(0.762in,0.486in)},
scale only axis,
xmin=0.02,
xmax=1,
xlabel={$\text{Measurement ratio}~(M/N_{\rm hs})$},
xtick={0.1,0.2,0.3,0.4,0.5,0.6,0.7,0.8,0.9,1},
xticklabels={0.1,0.2,0.3,0.4,0.5,0.6,0.7,0.8,0.9,1},
xlabel style={at = {(0.5,0.02)},font=\fontsize{12}{1}\selectfont},
ticklabel style={font=\fontsize{10}{1}\selectfont},
ymin=0,
ymax=30,
ytick={0,5,10,15,20,25,30},
yticklabels={0,5,10,15,20,25,30},
ylabel={RSNR (dB)},
ylabel style={at = {(0.04,0.5)},font=\fontsize{12}{1}\selectfont},
axis background/.style={fill=white},
legend columns=3,
legend style={at={(0,-0.2)},anchor=north west,legend cell align=left,align=left,fill=none,draw=none,font=\fontsize{6}{1}\selectfont,/tikz/column 2/.style={column sep=15mm},/tikz/column 4/.style={column sep=15mm}},
legend entries={{CS, Uncon. Exp., $\bs \Psi = \bs I_{N_{\rm p}} \otimes  \bs \Psi_{\rm 1D}$},
{CS, Uncon. Exp., $\bs \Psi = \bs \Psi_{\rm 2D} \otimes  \bs \Psi_{\rm 1D}$},
{ME, Uncon. Exp.},
{CS, Con. ~~~Exp., $\bs \Psi = \bs I_{N_{\rm p}} \otimes  \bs \Psi_{\rm 1D}$},
{CS, Con. ~~~Exp., $\bs \Psi = \bs \Psi_{\rm 2D} \otimes  \bs \Psi_{\rm 1D}$},
{ME, Con. ~~~Exp.}
}
]
\node[text=black, draw=none] at (rel axis cs:0.11,0.93) {\fontsize{9}{1}\selectfont (a) CI-FTI}; 
\node[align=left, text=black, draw=none,anchor=south] (a) at (180,140) {\fontsize{5}{1}\selectfont Fig.\ref{fig:sim2_res2}(a)}; 
\node (a1) at (180,110){};\draw[->] (a1) -- (a);
\node[align=left, text=black, draw=none,anchor=south east] (b) at (180,110) {\fontsize{5}{1}\selectfont Fig.\ref{fig:sim2_res2}(b)}; 
\node (b1) at (185,90){};\draw[->] (b1) -- (b);
\node[align=left, text=black, draw=none,anchor=north west] (d) at (210,45) {\fontsize{5}{1}\selectfont Fig. \ref{fig:sim2_res2}(c)};
\node (d1) at (160,50) {};\draw[->] (d1) -- (d);
\node (e1) at (280,130) {};\node (e2) at (280,65) {};\draw[<->] (e1) -- (e2);
\node[align=left, text=black, draw=none,anchor=west]  at (265,97) {\fontsize{5}{1}\selectfont $\ \approx\text{6.5 dB}$};

\addlegendimage{mark=x,red,solid,mark options={solid},line width=1.0pt}
\addlegendimage{mark=o,black,solid,mark options={solid},line width=1.0pt}
\addlegendimage{mark=+,blue,solid,mark options={solid},line width=1.0pt}
\addlegendimage{mark=x,red,dashed,mark options={solid},line width=1.0pt}
\addlegendimage{mark=o,black,dashed,mark options={solid},line width=1.0pt}
\addlegendimage{mark=+,blue,dashed,mark options={solid},line width=1.0pt}

\addplot [color=red,solid,line width=1.0pt,mark=x,mark options={solid}]
  table[row sep=crcr]{%
0.02	0.291338287354896\\
0.1	4.0766693160847\\
0.2	9.17050420115845\\
0.3	12.9819871675355\\
0.4	14.6752752056534\\
0.5	15.2330992716186\\
0.6	16.4392313291146\\
0.7	17.2401911462748\\
0.8	17.4161785599611\\
0.9	18.1080122402541\\
1	18.2875312978891\\
};
\addplot [color=red,dashed,line width=1.0pt,mark=x,mark options={solid}]
  table[row sep=crcr]{%
0.02	0.259655310462315\\
0.1	4.47079839958646\\
0.2	11.1842562390267\\
0.3	15.5632656942597\\
0.4	17.6813111327555\\
0.5	18.285\\
0.6	18.9282304287411\\
0.7	19.2093323234999\\
0.8	19.2753879377574\\
0.9	19.0365721773203\\
1	18.2875312978891\\
};

\addplot [color=black,solid,line width=1.0pt,mark=o,mark options={solid}]
  table[row sep=crcr]{%
0.02	0.432724149858494\\
0.1	4.32354218844923\\
0.2	9.30690853154713\\
0.3	14.7797163722504\\
0.4	17.076967821677\\
0.5	18.0599291979456\\
0.6	19.8959461175763\\
0.7	20.9483861951601\\
0.8	21.6837098725132\\
0.9	21.9970174310401\\
1	22.7780946780921\\
};

\addplot [color=black,dashed,line width=1.0pt,mark=o,mark options={solid}]
  table[row sep=crcr]{%
0.02	0.504350797433591\\
0.1	4.40164038421566\\
0.2	9.72870787568183\\
0.3	15.4660863664557\\
0.4	18.0338906012123\\
0.5	18.8751616483886\\
0.6	20.84792813563\\
0.7	21.6536489223939\\
0.8	22.2725116638126\\
0.9	22.2592975759607\\
1	22.7780946780921\\
};

\addplot [color=blue,solid,line width=1.0pt,mark=+,mark options={solid}]
  table[row sep=crcr]{%
1	12.9006876129134\\
0.9	11.7623898993566\\
0.8	11.2310987214288\\
0.7	10.4803738014821\\
0.6	9.64043205294349\\
0.5	8.56700086441659\\
0.4	7.89258766359379\\
0.3	6.61952856916502\\
0.2	4.66914379942768\\
0.1	3.02674633755848\\
0.02	0.47776\\
};

\addplot [color=blue,dashed,line width=1.0pt,mark=+,mark options={solid}]
  table[row sep=crcr]{%
1	12.9006876129134\\
0.9	11.7931147639484\\
0.8	11.2805197062069\\
0.7	10.5349102448967\\
0.6	9.69311441390552\\
0.5	8.6100419011505\\
0.4	7.92894123656159\\
0.3	6.64257616980061\\
0.2	4.68082395366488\\
0.1	3.03234244419695\\
0.02	0.47868\\
};

\end{axis}
\end{tikzpicture}%
}
  \end{minipage}
  \begin{minipage}{0.49\linewidth}
    \centering
    \vspace{-5mm}
    \scalebox{0.9}{
%
%
\begin{tikzpicture}[scale = 1]
\begin{axis}[%
width=3in,
height=2in,
at={(0.762in,0.486in)},
scale only axis,
xmin=0.02,
xmax=1,
xlabel={$\text{Measurement ratio}~(M/N_{\rm hs})$},
xtick={0.1,0.2,0.3,0.4,0.5,0.6,0.7,0.8,0.9,1},
xticklabels={0.1,0.2,0.3,0.4,0.5,0.6,0.7,0.8,0.9,1},
xlabel style={at = {(0.5,0.02)},font=\fontsize{12}{1}\selectfont},
ticklabel style={font=\fontsize{10}{1}\selectfont},
ymin=0,
ymax=30,
ytick={0,5,10,15,20,25,30},
yticklabels={0,5,10,15,20,25,30},
ylabel={RSNR (dB)},
ylabel style={at = {(0.04,0.5)},font=\fontsize{12}{1}\selectfont},
axis background/.style={fill=white}
]
\node[text=black, draw=none] at (rel axis cs:0.11,0.93) {\fontsize{9}{1}\selectfont (b) SI-FTI}; 
\node[align=left, text=black, draw=none,anchor=north east] (a) at (144,240) {\fontsize{5}{1}\selectfont Fig.\ref{fig:sim2_res2}(d)}; 
\node (a1) at (185,245){};\draw[->] (a1) -- (a);
\node[align=left, text=black, draw=none,anchor=north] (b) at (180,180) {\fontsize{5}{1}\selectfont Fig.\ref{fig:sim2_res2}(e)}; 
\node (b1) at (180,205){};\draw[->] (b1) -- (b);
\node[align=left, text=black, draw=none,anchor=north west] (d) at (210,45) {\fontsize{5}{1}\selectfont Fig. \ref{fig:sim2_res2}(f)};
\node (d1) at (160,50) {};\draw[->] (d1) -- (d);
\node (c1) at (280,275) {};\node (c2) at (280,225) {};\draw[<->] (c1) -- (c2);
\node[align=left, text=black, draw=none,anchor=west]  at (265,250) {\fontsize{5}{1}\selectfont $\ \approx\text{5 dB}$};
\node (e1) at (280,225) {};\node (e2) at (280,65) {};\draw[<->] (e1) -- (e2);
\node[align=left, text=black, draw=none,anchor=west]  at (265,145) {\fontsize{5}{1}\selectfont $\ \approx\text{16 dB}$};
\addplot [color=black,solid,line width=1.0pt,mark=o,mark options={solid}]
  table[row sep=crcr]{%
1	27.203153341177\\
0.9	26.7500002976235\\
0.8	26.4151496310173\\
0.7	26.0110467117923\\
0.6	25.4070020474973\\
0.5	24.6544421536395\\
0.4	23.8476027119339\\
0.3	22.5850011061536\\
0.2	20.4578232400665\\
0.1	16.1845860948836\\
0.02	7.48126623536697\\
};

\addplot [color=black,dashed,line width=1.0pt,mark=o,mark options={solid}]
  table[row sep=crcr]{%
1	27.203153341177\\
0.9	27.5900256416671\\
0.8	27.8851493646377\\
0.7	28.1562102551593\\
0.6	28.495314304738\\
0.5	28.6356280367611\\
0.4	28.4431178333461\\
0.3	27.4023721967207\\
0.2	24.377104639457\\
0.1	17.8682787340575\\
0.02	7.48417169771716\\
};

\addplot [color=blue,solid,line width=1.0pt,mark=+,mark options={solid}]
  table[row sep=crcr]{%
1	12.4763628556388\\
0.9	11.8196432252404\\
0.8	11.1163259151755\\
0.7	10.3602855527667\\
0.6	9.52657246084252\\
0.5	8.6144732784481\\
0.4	7.60663511490738\\
0.3	6.41811500203849\\
0.2	5.02206967969799\\
0.1	3.19394806086136\\
0.02	0.904944395084541\\
};

\addplot [color=blue,dashed,line width=1.0pt,mark=+,mark options={solid}]
  table[row sep=crcr]{%
1	12.4763628556388\\
0.9	11.851367002926\\
0.8	11.1647432479917\\
0.7	10.4139520984033\\
0.6	9.57762735028409\\
0.5	8.65790274234788\\
0.4	7.64016779939477\\
0.3	6.44101221564104\\
0.2	5.03533451739425\\
0.1	3.19940616441017\\
0.02	0.905854747117305\\
};

\end{axis}
\end{tikzpicture}%
}
  \end{minipage}
  \caption{The reconstruction performance of CI/SI-FTI systems solved by $\ell_1$ minimization (CS) and ME problem with the constrained/unconstrained light exposure budget. Note that for a fixed measurement ratio, the amount of light exposure in the unconstrained exposure scenario is less than the one in constrained exposure scenario, by a factor $M/N_{\rm hs}$.} 
  \label{fig:sim2_res1}
\end{figure}
\subsection{Reconstruction performances on a synthetic biological HS volume}
\label{subsec:Synthetic biological HS volume}
We now test the performance of the proposed compressive FTI frameworks in a more realistic context including measurement noise. We consider sensing scenarios with both constrained and unconstrained light exposure.
  
As illustrated in Fig.~\ref{fig:gt}, we simulate a biological HS volume of size $(N_\xi, N_{\rm p}) = (512,64^2)$ by mixing the coefficients of three spectral bands of a synthetic biological RGB image (selected from the benchmark images \cite{Ruusuvuori2008}) with the known spectra of three common fluorochromes. 

The Nyquist measurements are formed as $\bs y^{\Nyq} = \bs{\Phi}^*\bs{x}+\bs{n}^{\Nyq}$ where $n_{\nell,k}^\Nyq \sim_{\iid}\mathcal{N}(0,\sigma_\Nyq)$ and $\sigma_\Nyq$ is fixed such that $10\log( {\|\bs{x}\|^2}/{\|\bs{n}^\Nyq\|^2}) \approx 20$ dB. In the un\-con\-strained-exposure context, CI-FTI and SI-FTI observations are formed according to \eqref{eq:cifti acquisition} and \eqref{eq:sifti acquisition}, where the sets $\Omega^\xi$ and $\Omega$ are randomly generated from the pmfs \eqref{eq:pmf-ci-FTI} and \eqref{eq:p-bound-SI}, respectively, and the variance of each component of the associated additive Gaussian measurement noises is also set to $\sigma_\Nyq^2$. In the context of constrained-exposure simulations, the values of ground truth HS volume are multiplied by $N_\xi/ M_\xi$ and $N_\xi/ \bar M$ for CI-FTI and SI-FTI, respectively, according to the linear intensity scaling assumption explained in \eqref{eq:cifti ci acquisition} and \eqref{eq:sifti ci acquisition} (Sec.~\ref{sec:Sec_05}).

This sensing context is repeated over 10 random realizations of the noise, and the sets $\Omega^\xi$ or $\Omega$. Fig.~\ref{fig:sim2_res1} depicts the reconstruction SNR in dB as a function of the number of measurements. The HS volumes are reconstructed by solving two types of problems: \emph{(i)} the CS-based $\ell_1$ minimization problem~\eqref{eq:cifti main reconstruction} or~\eqref{eq:sifti main reconstruction}; and \emph{(ii)} the Minimal Energy (ME) problem~\cite{candes2006robust}, \ie a standard reconstruction method that amounts to applying the pseudo-inverse of the sensing operator on the noisy measurements. 

In both cases, ME reconstructions (the blue curves) do not reach CS reconstruction qualities, even though they benefit of the VDS strategy.
The poor performances of the ME solutions advocate the necessity of promoting sparsity prior in the reconstruction problem. The increased RSNR quality of SI-FTI over CI-FTI  is induced by both the 3D wavelet sparsity model and the greater diversity of the compressive sensing matrix in SI-FTI where spatial information of the HS volume is captured at each OPD index in $\range{N_\xi}$; in CI-FTI, this spatial information is observed only on the $M^\xi$ selected OPD indices. We recall that since there are repeated indices in subsampled sets $\Omega^\xi$ and $\Omega$, even for $M/N_{\rm hs}=1$, we cannot sample all the distinct elements. As a consequence, ME reconstructions does not reach Nyquist quality at $M/N_{\rm hs}=1$.

We now question whether the superior performance of SI-FTI is the consequence of the sensing diversity or of the selected 3D wavelet sparsity model. 
For this, deviating from what is guaranteed by CS literature (see Sec.~\ref{sec:cifti_HS Reconstruction Method and Guarantee}), we test the recovery of HS data from CI-FTI measurements using the recovery scheme \eqref{eq:sifti main reconstruction} regularized by the same prior as in SI-FTI, \ie with the 3D wavelet sparsity basis $\bs \Psi = \bs \Psi_{\rm 2D} \otimes \bs \Psi_{\rm 1D}$. We thus solve this optimization with the changes $\bs D \rightarrow \bs I_{N_{\rm p}} \otimes \sqrt{M_\xi}\bs D^\xi$, $M \rightarrow M_\xi N_{\rm p}$, and $\bs y^{\rm si} \rightarrow \bs y^{\rm ci}$, and assuming that the level of the noise is unaffected between the CI-FTI and the SI-FTI sensing schemes. The results are displayed in Fig.~\ref{fig:sim2_res1}-left (the black curves). We observe that this 3D wavelet sparsity model increases the RSNR, up to 4.5 dB, when $M/N_{\rm hs}$ is close to 1. This improvement, however, does not reach the RSNR of SI-FTI displayed in Fig.~\ref{fig:sim2_res1}-right. We thus conclude that the SI-FTI scheme benefit more from the diversity of its measurements than from its regularization compared to CI-FTI. In practice, HS data in CI-FTI should be reconstructed with this 3D wavelet sparsity basis to reach the highest RSNR when $M/N_{\rm hs}$ is large. Keeping this in mind, the next CI-FTI experiments are, however, run with $\bs \Psi = \bs I_{N_{\rm p}} \otimes \bs \Psi_{\rm 1D}$ for simplicity.

For constrained-exposure scenarios, we observe that --- as opposed to the common belief --- subsampling 40 percent of the FTI measurements yields better HS recovery quality than the one obtained from maximal number of FTI measurements\footnote{We note that in the case of CI-FTI regularized with 3D wavelet basis, there is still a 1~dB gain (around $M/N_{\rm hs} = 0.5$) to use the constrained-exposure scenario over the unconstrained one. However, a deeper analysis of this effect for this setting goes beyond the scope of this work.}. However, this effect vanishes for large values of $M$ and the RSNR decreases to reach the quality performances at the maximal number of measurements, which parallels with the behavior of the second error term in~\eqref{eq:cifti ci error} and~\eqref{eq:sifti ci error}. This phenomenon is due to the fact that the best $K$-term approximation error in the right-hand-side of \eqref{eq:cifti ci error} or \eqref{eq:sifti ci error} is dominated by the noise term, which can be reduced by taking less number of measurements. Besides, ME reconstruction cannot leverage constrained-exposure budget, since the measurement noise is the same as the one in unconstrained-exposure scenario. 
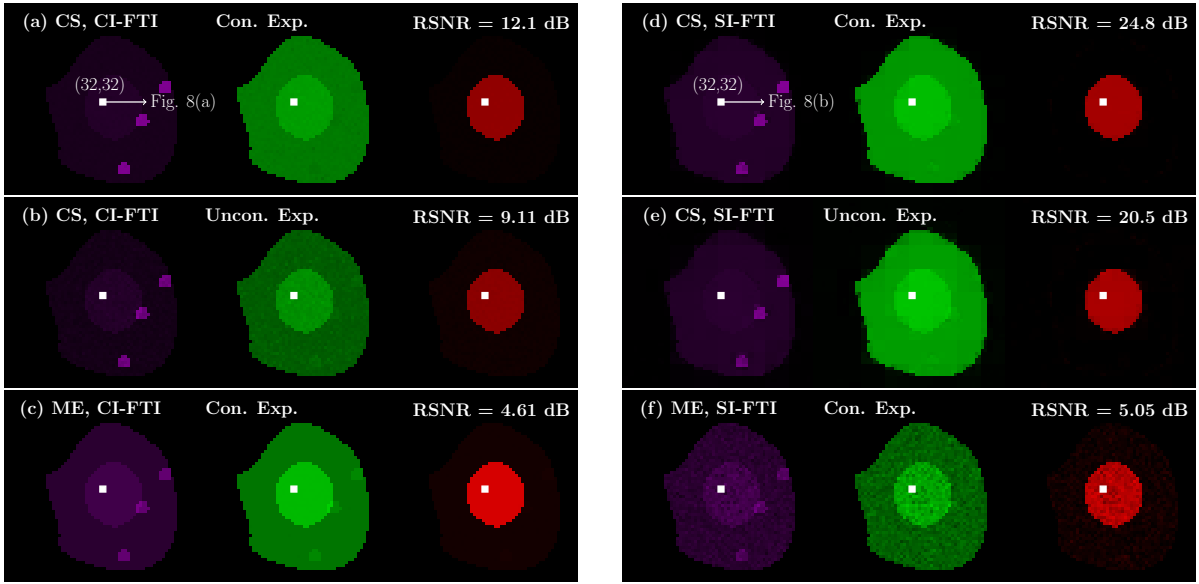
\begin{figure}[tbh]
  \centering
  \begin{minipage}{0.49\linewidth}
    \centering
    \scalebox{0.4}{
%
%
\begin{tikzpicture}[scale = 1]
\centering
\begin{axis}[%
width=7.5in,
height=2.5in,
at={(0in,0in)},
scale only axis,
point meta min=1,
point meta max=384,
axis on top,
xmin=0.5,
xmax=192.5,
y dir=reverse,
ymin=0.5,
ymax=64.5,
hide axis
]
\addplot [forget plot] graphics [xmin=0.5,xmax=192.5,ymin=0.5,ymax=64.5] {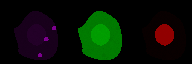};
\node[text=white, draw=none] at (rel axis cs:0.15,0.9) {\fontsize{18}{1}\selectfont \textbf{(a) CS, CI-FTI}}; 
\node[text=white, draw=none] at (rel axis cs:0.45,0.9) {\fontsize{18}{1}\selectfont \textbf{Con. Exp.}}; 
\node[text=white, draw=none] at (rel axis cs:0.85,0.9) {\fontsize{18}{1}\selectfont \textbf{RSNR = 12.1 dB}}; 
\filldraw[color = white, white, very thick] (32,320)rectangle (34,340);
\filldraw[color = white, white, very thick] (96,320)rectangle (98,340);
\filldraw[color = white, white, very thick] (160,320)rectangle (162,340);
\node[align=center, text=white, draw=none] (a) at (60,330) {\fontsize{18}{1}\selectfont \textcolor{white}{Fig. \ref{fig:sim2_res3}(a)}};
\node (a1) at (32,330) {};\draw[color = white,->,very thick] (a1) -- (a);
\node[above, align=center, text=white, draw=none] (b) at (32,320) {\fontsize{18}{1}\selectfont (32,32)};
\end{axis}
\end{tikzpicture}%
}
  \end{minipage}
  \begin{minipage}{0.49\linewidth}
    \centering
    \scalebox{0.4}{
%
%
\begin{tikzpicture}[scale = 1]
\centering
\begin{axis}[%
width=7.5in,
height=2.5in,
at={(0in,0in)},
scale only axis,
point meta min=1,
point meta max=384,
axis on top,
xmin=0.5,
xmax=192.5,
y dir=reverse,
ymin=0.5,
ymax=64.5,
hide axis
]
\addplot [forget plot] graphics [xmin=0.5,xmax=192.5,ymin=0.5,ymax=64.5] {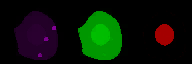};
\node[text=white, draw=none] at (rel axis cs:0.15,0.9) {\fontsize{18}{1}\selectfont \textbf{(d) CS, SI-FTI}}; 
\node[text=white, draw=none] at (rel axis cs:0.45,0.9) {\fontsize{18}{1}\selectfont \textbf{Con. Exp.}}; 
\node[text=white, draw=none] at (rel axis cs:0.85,0.9) {\fontsize{18}{1}\selectfont \textbf{RSNR = 24.8 dB}}; 
\filldraw[color = white, white, very thick] (32,320)rectangle (34,340);
\filldraw[color = white, white, very thick] (96,320)rectangle (98,340);
\filldraw[color = white, white, very thick] (160,320)rectangle (162,340);
\node[align=center, text=white, draw=none] (a) at (60,330) {\fontsize{18}{1}\selectfont \textcolor{white}{Fig. \ref{fig:sim2_res3}(b)}};
\node (a1) at (32,330) {};\draw[color = white,->,very thick] (a1) -- (a);
\node[above, align=center, text=white, draw=none] (b) at (32,320) {\fontsize{18}{1}\selectfont (32,32)};
\end{axis}
\end{tikzpicture}%
}
  \end{minipage}
  \begin{minipage}{0.49\linewidth}
    \centering
    \scalebox{0.4}{
%
%
\begin{tikzpicture}[scale = 1]
\centering
\begin{axis}[%
width=7.5in,
height=2.5in,
at={(0in,0in)},
scale only axis,
point meta min=1,
point meta max=384,
axis on top,
xmin=0.5,
xmax=192.5,
y dir=reverse,
ymin=0.5,
ymax=64.5,
hide axis
]
\addplot [forget plot] graphics [xmin=0.5,xmax=192.5,ymin=0.5,ymax=64.5] {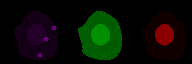};
\node[text=white, draw=none] at (rel axis cs:0.15,0.9) {\fontsize{18}{1}\selectfont \textbf{(b) CS, CI-FTI}}; 
\node[text=white, draw=none] at (rel axis cs:0.45,0.9) {\fontsize{18}{1}\selectfont \textbf{Uncon. Exp.}}; 
\node[text=white, draw=none] at (rel axis cs:0.85,0.9) {\fontsize{18}{1}\selectfont \textbf{RSNR = 9.11 dB}}; 
\filldraw[color = white, white, very thick] (32,320)rectangle (34,340);
\filldraw[color = white, white, very thick] (96,320)rectangle (98,340);
\filldraw[color = white, white, very thick] (160,320)rectangle (162,340);
\end{axis}
\end{tikzpicture}%
}
  \end{minipage}
  \begin{minipage}{0.49\linewidth}
    \centering
    \scalebox{0.4}{
%
%
\begin{tikzpicture}[scale = 1]
\centering
\begin{axis}[%
width=7.5in,
height=2.5in,
at={(0in,0in)},
scale only axis,
point meta min=1,
point meta max=384,
axis on top,
xmin=0.5,
xmax=192.5,
y dir=reverse,
ymin=0.5,
ymax=64.5,
hide axis
]
\addplot [forget plot] graphics [xmin=0.5,xmax=192.5,ymin=0.5,ymax=64.5] {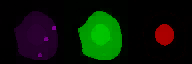};
\node[text=white, draw=none] at (rel axis cs:0.15,0.9) {\fontsize{18}{1}\selectfont \textbf{(e) CS, SI-FTI}}; 
\node[text=white, draw=none] at (rel axis cs:0.45,0.9) {\fontsize{18}{1}\selectfont \textbf{Uncon. Exp.}}; 
\node[text=white, draw=none] at (rel axis cs:0.85,0.9) {\fontsize{18}{1}\selectfont \textbf{RSNR = 20.5 dB}}; 
\filldraw[color = white, white, very thick] (32,320)rectangle (34,340);
\filldraw[color = white, white, very thick] (96,320)rectangle (98,340);
\filldraw[color = white, white, very thick] (160,320)rectangle (162,340);
\end{axis}
\end{tikzpicture}%
}
  \end{minipage}
  \begin{minipage}{0.49\linewidth}
    \centering
    \scalebox{0.4}{
%
%
\begin{tikzpicture}[scale = 1]
\centering
\begin{axis}[%
width=7.5in,
height=2.5in,
at={(0in,0in)},
scale only axis,
point meta min=1,
point meta max=384,
axis on top,
xmin=0.5,
xmax=192.5,
y dir=reverse,
ymin=0.5,
ymax=64.5,
hide axis
]
\addplot [forget plot] graphics [xmin=0.5,xmax=192.5,ymin=0.5,ymax=64.5] {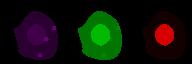};
\node[text=white, draw=none] at (rel axis cs:0.15,0.9) {\fontsize{18}{1}\selectfont \textbf{(c) ME, CI-FTI}}; 
\node[text=white, draw=none] at (rel axis cs:0.45,0.9) {\fontsize{18}{1}\selectfont \textbf{Con. Exp.~~~~}}; 
\node[text=white, draw=none] at (rel axis cs:0.85,0.9) {\fontsize{18}{1}\selectfont \textbf{RSNR = 4.61 dB}}; 
\filldraw[color = white, white, very thick] (32,320)rectangle (34,340);
\filldraw[color = white, white, very thick] (96,320)rectangle (98,340);
\filldraw[color = white, white, very thick] (160,320)rectangle (162,340);
\end{axis}
\end{tikzpicture}%
}
  \end{minipage}
  \begin{minipage}{0.49\linewidth}
    \centering
    \scalebox{0.4}{
%
%
\begin{tikzpicture}[scale = 1]
\centering
\begin{axis}[%
width=7.5in,
height=2.5in,
at={(0in,0in)},
scale only axis,
point meta min=1,
point meta max=384,
axis on top,
xmin=0.5,
xmax=192.5,
y dir=reverse,
ymin=0.5,
ymax=64.5,
hide axis
]
\addplot [forget plot] graphics [xmin=0.5,xmax=192.5,ymin=0.5,ymax=64.5] {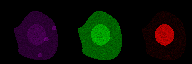};
\node[text=white, draw=none] at (rel axis cs:0.15,0.9) {\fontsize{18}{1}\selectfont \textbf{(f) ME, SI-FTI}}; 
\node[text=white, draw=none] at (rel axis cs:0.45,0.9) {\fontsize{18}{1}\selectfont \textbf{Con. Exp.~~~~}}; 
\node[text=white, draw=none] at (rel axis cs:0.85,0.9) {\fontsize{18}{1}\selectfont \textbf{RSNR = 5.05 dB}}; 
\filldraw[color = white, white, very thick] (32,320)rectangle (34,340);
\filldraw[color = white, white, very thick] (96,320)rectangle (98,340);
\filldraw[color = white, white, very thick] (160,320)rectangle (162,340);
\end{axis}
\end{tikzpicture}%
}
  \end{minipage}
  \caption{An example of the reconstructed HS volumes from 20\,\% of the total measurements (or light exposure). The spectral content at the spatial location indicated by a white square is shown in Fig.~\ref{fig:sim2_res3}.}
  \label{fig:sim2_res2}
\end{figure}
\begin{figure}[tbh]
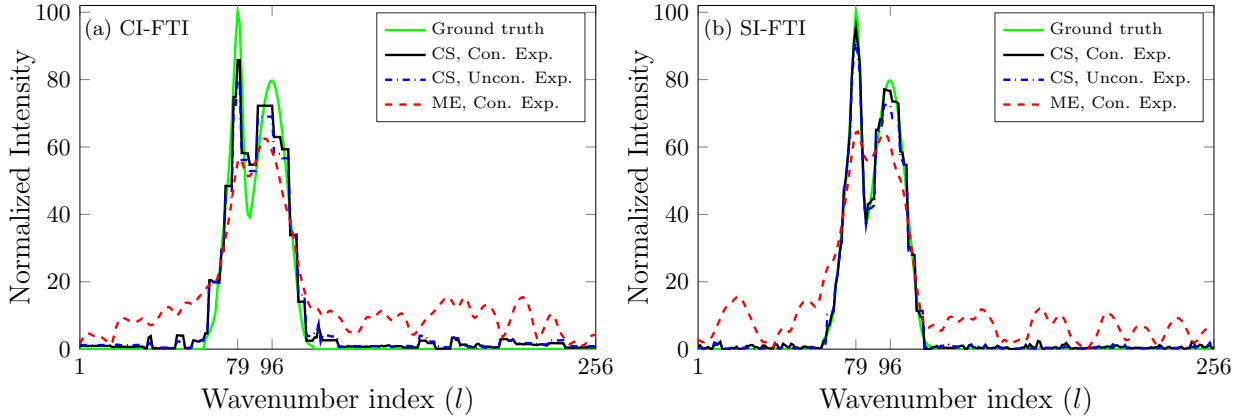

  \begin{minipage}{0.49\linewidth}
    \centering
    \scalebox{0.9}{\input{Figures/Sim2_res3_CIFTI.tex}}
  \end{minipage}
  \begin{minipage}{0.49\linewidth}
    \centering
    \scalebox{0.9}{\input{Figures/Sim2_res3_SIFTI.tex}}
  \end{minipage}
  \caption{The spectrum of the reconstructed HS volumes in Fig.~\ref{fig:sim2_res2} at spatial location $(q_1,q_2)=(32,32)$.}
  \label{fig:sim2_res3}
\end{figure}

Fig.~\ref{fig:sim2_res2} illustrates three spatial maps associated with the wavenumber indices $\nell \in \{72, 79, 96\}$ of the reconstructed HS volumes. This figure is one instance of the Monte-Carlo trials in the experiment of Fig.~\ref{fig:sim2_res1}, for $M/\Nhs = 0.2$. It can be seen that, for both systems, the HS volume recovered with CS-based formulation preserves the spatial configuration of the specimen, although the RSNR for CI-FTI may not be satisfactory. However, as mentioned in Sec.~\ref{sec:intro}, the main motivation for studying FTI is to acquire HS volumes with high spectral resolution, \ie limited by the life-time of the fluorochromes. Note that the spatial resolution is determined by physical characteristics of the 2D imaging sensor in Fig.~\ref{fig:FTI_scheme}. In order to illustrate the ability of the proposed approaches in preserving the spectral information, Fig.~\ref{fig:sim2_res3} also compares the spectral content  of the reconstructed HS volumes at the center spatial location indicated by a white square. The artifacts in the solution of ME problem are obvious and they could be sufficiently prominent for causing false detections in the decompositions of a biological HS volume into its spectral constituents. Despite minor disturbances, the CS-based solution in CI-FTI successfully follows the ground truth spectrum. However, the solution of SI-FTI almost perfectly matches the ground truth signal. 
\subsection{Constrained-exposure simulations from actual experimental data}\label{subsec:Simulated experimental data} 
To prove the concept of CI-FTI, we have simulated a CI-FTI setup in a constrained-exposure context by subsampling the data recorded by an actual (Nyquist) FTI acquisition. We have scanned a thin layer of a biologic cell (\ie Convallaria, lily of the valley, cross section of rhizome with concentric vascular bundles) stained on a glass slide, using a confocal microscope (with a 40$\times$ objective) equipped with an FTI and a LUXEON light source~\cite{LUXEON}, \ie the \textit{HYPE} device used in~\cite{Leonard2015}. By changing the current level of the light source from $I_0=100$ mA to $I_1=700$ mA (with $\Delta I =25$ mA) we acquired 25 sets of Nyquist FTI measurements of size $(N_\xi,\bar N_{\rm p},\bar N_{\rm p}) = (1024,128,128)$, with ${\rm OPD}_{\rm max}=2.537~\mu {\rm m}$, \ie each set corresponds to one light intensity level\footnote{The light source specifications~\cite{LUXEON} specifies that the relative luminous flux is proportional to the forward current.}. This exhaustive acquisition enables us to simulate constrained-exposure compressive sensing-FTI (see below). 

Concerning the level of $\sigma_\Nyq$ of the measurement noise, we have estimated it using an RM estimator~\cite{donoho1994ideal} applied on each of the 25 HS data cubes. We observed that, over $[I_0, I_1]$ with $I_1/I_0 =7$, $\sigma_\Nyq \in [0.1408,0.2272]$, \ie a linear regression provides $\sigma_\Nyq(I) \approx a (I/100) + b$ with $a=1.44~10^{-2}$ and $b=1.26~10^{-1}$. This limited variation of $\sigma$ when $I$ varies confirms that the dependence of the Nyquist noise on the light intensity can be neglected (see Sec.~\ref{sec:noise} and Sec.~\ref{sec:Sec_05}).
\medskip

The simulated CI-FTI sensing model of this section contains one crucial difference compared to the subsampling procedure studied in Sec.~\ref{sec:Sec_03} and tested in the two previous sections. Since the subsampling set $\Omega^\xi$ contains possible repetition of the same OPD indices (\ie according to the pmf \eqref{eq:pmf-ci-FTI} that is used to draw $M_\xi$ \iid OPD indices with repetition), a restriction of the (noisy) Nyquist FTI observations above to $\Omega^\xi$ at each pixel location copies the same realizations of the noise on repeated indices. This differs from the model~\eqref{eq:cifti acquisition} where the noise samples vary on possible multiple instances of the same OPD index. 

Consequently, rather than artificially copying the same observations for each repeated index in $\Omega^\xi$, we have adopted the following more realistic experimental sensing scenario. For a given $M_\xi$, the subsampling multiset $\Omega^\xi$ is randomly generated according to the pmf \eqref{eq:pmf-ci-FTI}. We then construct a set $\bar{\Omega}^\xi:=\{\bar{\Omega}_1,\cdots,\bar{\Omega}_{M_\textsl{eff}}\}$ with cardinality $M_\textsl{eff} < M_\xi$ which is a replica of $\Omega^\xi$ with only unique indices\footnote{For $l\in\range{N_\xi}$, define the Bernoulli random variable $X_l$  being equal to 1 if the $l^{\rm th}$ OPD element is picked at least once over $M_\xi$ trials, and zero otherwise, then $\bb E X_l = \bb P(X_l = 1) = 1-(1-p(l))^{M_\xi}$ and $\bb E M_{\rm eff} = \sum_l 1-(1-p(l))^{M_\xi}$, since $M_{\rm eff} = \sum_l X_l$.}. Finally, CI-FTI measurements are built by restricting the Nyquist FTI measurements $\bs y^\Nyq$ to the subset~$\bar{\Omega}^\xi$. 

By following this procedure, the fidelity term in~\eqref{eq:cifti sub reconstruction} must be adapted; the weighting matrix $\bs D^\xi$ must stay constant across its diagonal elements. Indeed, we can easily show that each index~$\nell$ presents in $\Omega^\xi$ is repeated according to a binomial \rv $\omega_\nell$ of $M_\xi$ trials and success probability $p(\nell)$, with $p$ the pmf defined in \eqref{eq:pmf-ci-FTI}. Therefore, $\bb E \omega_\nell = M_\xi p(\nell)$, which is exactly the inverse of each entry $\tinv{M_\xi} d^2_{kk} = \frac{1}{M_\xi p(\Omega_k)}$ of $\inv{M_\xi} (\bs D^\xi)^2$ for which $\Omega_k = l$. In other words, in expectation, \emph{the matrix $\bs D^\xi$ accounts for a normalization of the multiplicity of each index present in $\Omega^\xi$}. While a careful mathematical analysis of this effect is postponed to a future work, we conclude that, for each optimization vector $\bs u_j$ and each $j \in \range{N_{\rm p}}$, $\|\bs I^*_{\bar{\Omega}^\xi} (\bs y^\Nyq_j - \bs F^* \bs u_j)\|^2 \approx \tinv{M_\xi} \|\bs D^\xi \bs I^*_{\Omega^\xi} (\bs y^\Nyq_j - \bs F^* \bs u_j)\|^2 \leq \varepsilon_\CI^2$ is an appropriate fidelity term in~\eqref{eq:cifti sub reconstruction} in the replacement of $\Omega^\xi$ by~$\bar \Omega^\xi$.
\medskip

Back to our simulation setup, from the recorded Nyquist FTI measurements, a CI-FTI system with constrained-exposure budget is simulated as follows. Given a reference current level of the light source, \eg $I_{\rm ref}=100$ mA as for the black curve in Fig.~\ref{fig:sim3_res1}, without loss of generality, we assume\footnote{Remind that the flux of the light source is proportional to its current level.} $I_\textsl{opd} = I_\textsl{ref}/\tau_\xi$ and thus a constrained-exposure budget is fixed as $I_\textsl{tot} = I_{\rm ref} N_\xi \le I_\textsl{max}$. Therefore, for every other current levels, $I({\rm mA})$, we are allowed to set $M_\xi(I,I_{\rm ref}) = \frac{I_{\rm ref}}{I} N_\xi$, \eg
$M_\xi(700,100)/N_\xi = 14.28\,\%$. Besides, in view of the above-mentioned sampling approach, considering $M_\textsl{eff} \le M$ unique subsampled indices does not violate the requirements of the constrained-exposure budget scenario. Hereafter, we report effective measurement ratio ($M_\textsl{eff}/N_\xi$) as it will be the actual exposure reduction ratio.

Since the Nyquist FTI measurements recorded at $I({\rm mA})=700$ mA has the highest MNR (see \eqref{eq:MNR-def}), the reconstructed HS volume\footnote{To do so, in \eqref{eq:cifti sub reconstruction}, we set $\Omega^\xi = \range{N_\xi}, \bs D^\xi = \sqrt{N_\xi}\bs I_{N_\xi}$, and $M_\xi=N_\xi$.} from this data is taken as the ground truth volume, \eg for the purpose of computing RSNR. For the sake of obtaining fair RSNR values, all the spectra are normalized with respect to their $\ell_2$-norm, meaning that in RSNR formula we replace $\bs x = [\bs x^\top_1,\cdots,\bs x^\top_{N_{\rm p}}]^\top$ with $\bs x^\circ := [\bs x^\top_1/\|\bs x_1\|,\cdots,\bs x^\top_{N_{\rm p}}/\|\bs x_{N_{\rm p}}\|]^\top$ and replace $\hat{\bs x} = [\hat{\bs x}^\top_1,\cdots,\hat{\bs x}^\top_{N_{\rm p}}]^\top$ with $\hat{\bs x}^\circ := [\hat{\bs x}^\top_1/\|\hat{\bs x}_1\|,\cdots,\hat{\bs x}^\top_{N_{\rm p}}/\|\hat{\bs x}_{N_{\rm p}}\|]^\top$. This allows for the computation of the RSNR independently of the light intensity --- the intensity of the reconstructed spectrum being proportional to the light intensity --- and it reduces the effect of various data corruptions, \eg saturated pixels, light diffraction due to the biological specimen, or motion artifact during the experiment.

\begin{figure}[tb]
  \null\hfill\begin{minipage}[t]{0.47\linewidth}
    \centering
%
%
\definecolor{mycolor1}{rgb}{0.00000,0.44700,0.74100}%
\definecolor{mycolor2}{rgb}{1.00000,0.00000,1.00000}%
\begin{tikzpicture}[scale = 0.8]
\begin{axis}[%
width=3in,
height=2in,
at={(0.762in,0in)},
scale only axis,
xmin=0.05,
xmax=0.4,
xlabel={$\text{Effective measurement ratio}~(M_\textsl{eff}/N_\xi)$},
xtick={0.1,0.2,0.3,0.4},
xticklabels={0.1,0.2,0.3,0.4},
xlabel style={at = {(0.5,0.02)},font=\fontsize{12}{1}\selectfont},
ticklabel style={font=\fontsize{10}{1}\selectfont},
ymin=7,
ymax=15,
ytick={7,8,9,10,11,12,13,14,15},
yticklabels={7,8,9,10,11,12,13,14,15},
ylabel={RSNR (dB)},
ylabel style={at = {(0.04,0.5)},font=\fontsize{12}{1}\selectfont},
axis background/.style={fill=white},
legend style={at={(0.2,0.35)},anchor=north west,legend cell align=left,align=left,draw=white!15!black,font=\fontsize{11}{1}\selectfont}
]
\addplot [color=black,line width=0.5pt,only marks,mark=o,mark options={solid}]
  table[row sep=crcr]{%
0.3369140625	7.86972922966686\\
0.291015625	10.1447730149467\\
0.2705078125	11.1385333901642\\
0.2392578125	11.7595070242504\\
0.2099609375	12.2093109924338\\
0.193359375	12.965326729165\\
0.189453125	12.4746587787917\\
0.1689453125	13.2168100711228\\
0.154296875	13.3217838505492\\
0.1572265625	13.3369831806143\\
0.1474609375	13.4506739252724\\
0.138671875	13.3734920835185\\
0.123046875	13.3429051591467\\
0.1044921875	13.0854716590969\\
0.1171875	13.4359634928532\\
0.1201171875	13.6886829835256\\
0.12109375	12.8822738264134\\
0.1044921875	13.6537843492127\\
0.107421875	13.312922704166\\
0.091796875	12.6889583132442\\
0.09765625	13.1203889576301\\
0.09375	13.4994883634587\\
0.0947265625	13.1341071112126\\
0.0849609375	12.8460462932671\\
0.0888671875	11.956\\
0.3349609375	7.93253468194773\\
0.2880859375	10.1338473608626\\
0.265625	11.1639733263069\\
0.2470703125	11.7018541115881\\
0.2158203125	12.3266260672327\\
0.197265625	12.7434204562867\\
0.1630859375	12.4524682557328\\
0.15234375	12.9653636059724\\
0.1708984375	13.4085459708775\\
0.1552734375	13.2723701306902\\
0.1533203125	13.2072863204855\\
0.150390625	13.5009063949929\\
0.1259765625	13.1240117970319\\
0.12890625	13.4369163862956\\
0.1298828125	13.4719002515492\\
0.1220703125	13.4225947052765\\
0.1044921875	13.0886340694984\\
0.109375	13.655188277225\\
0.1103515625	13.6696332279027\\
0.1044921875	13.4101079229551\\
0.09765625	13.410694540719\\
0.1015625	13.5755372076125\\
0.09375	13.302594925488\\
0.0888671875	13.2650452506815\\
0.0810546875	12.9332409369664\\
};
\addlegendentry{\ $I_\text{ref}$ = 100 mA};

\addplot [color=blue,line width=0.5pt,only marks,mark=+,mark options={solid}]
  table[row sep=crcr]{%
0.3271484375	12.2811971657973\\
0.3154296875	13.0568733782696\\
0.3017578125	12.5498140777105\\
0.2666015625	13.6228317536017\\
0.267578125	13.9466828189503\\
0.2509765625	14.0341925747196\\
0.244140625	14.1007082611159\\
0.2138671875	13.9865708546928\\
0.208984375	14.2525572292072\\
0.21484375	14.3331473621376\\
0.193359375	14.4088831841383\\
0.185546875	14.2137052730312\\
0.1875	14.3881254016406\\
0.1787109375	14.1843905450534\\
0.1484375	14.0976156739064\\
0.169921875	14.5753562660512\\
0.1640625	14.4371402879874\\
0.158203125	14.4594805695597\\
0.1533203125	14.2870971634491\\
0.1484375	14.0763383486134\\
0.1484375	14.4475050217045\\
0.3544921875	12.2902388853537\\
0.322265625	12.8871458008573\\
0.2900390625	12.6456811155148\\
0.2744140625	13.6970598425434\\
0.255859375	13.6880634196732\\
0.2587890625	14.0332658047914\\
0.2412109375	13.9063832737883\\
0.2275390625	14.1740077532809\\
0.212890625	14.1924958802397\\
0.1904296875	14.1509013691941\\
0.189453125	14.335227241462\\
0.197265625	14.4156891583857\\
0.193359375	14.4525845443376\\
0.189453125	14.554460452798\\
0.1728515625	14.3164957780285\\
0.1787109375	14.4171656699666\\
0.1474609375	13.9611968690325\\
0.1630859375	14.6875833257798\\
0.1533203125	14.1795519098173\\
0.1572265625	14.3839819169756\\
0.140625	14.5360873413881\\
};
\addlegendentry{\ $I_\text{ref}$ = 200 mA};

\end{axis}
\draw[red,line width=.5pt,anchor=center] (2.7,3.7) node [left] {\fontsize{5}{1}\selectfont P\textsuperscript{1}} ellipse (4pt and 8pt) ;
\draw[red,line width=.5pt,anchor=center] (8.17,0.6) node [right] {\fontsize{5}{1}\selectfont P\textsuperscript{2}} ellipse (4pt and 8pt) ;
\draw[red,line width=.5pt,anchor=center] (4.1,4.7) node [left] {\fontsize{5}{1}\selectfont P\textsuperscript{3}} ellipse (4pt and 8pt) ;
\draw[red,line width=.5pt,anchor=center] (8.56,3.35) node [right] {\fontsize{5}{1}\selectfont P\textsuperscript{4}} ellipse (4pt and 8pt) ;

\end{tikzpicture}%

    \caption{The reconstruction quality of CI-FTI system with constrained-exposure budget. For each curve, the (relative) light exposure budget is constrained to $I_{\rm ref}N_\xi$; and thus $M_\xi(I,I_{\rm ref})/N_\xi = I_{\rm ref}/I$ for $I \in \{I_{\rm ref},I_{\rm ref}+25,\cdots,700\}$ mA. Subsampling the OPD axis at higher light intensity results always in superior reconstruction.}
    \label{fig:sim3_res1}
  \end{minipage}
  \hfill
  \begin{minipage}[t]{0.47\linewidth}
    \centering
    \input{Figures/Sim3_res3_nu_afternormalization.tex}
    \caption{The spectra of the reconstructed HS volumes, associated with four marked instances in Fig.~\ref{fig:sim3_res1}, at the center spatial location. Performing the proposed VDS strategy, even at the maximum compression ratio, significantly improves the quality of the reconstructed spectrum, especially in terms of denoising.}
    \label{fig:sim3_res3}
  \end{minipage}
  \hfill
  \null
\end{figure}
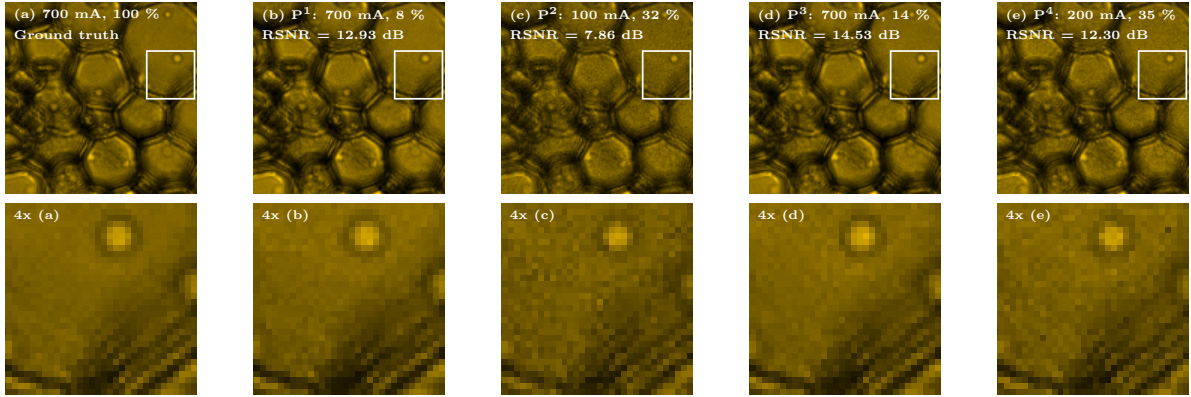
\begin{figure}[tb]
  \begin{minipage}{1.01\linewidth}
    \begin{minipage}{0.19\linewidth}
      \centering
%
%
\begin{tikzpicture}[scale = 0.8]

\begin{axis}[%
width=1.25in,
height=1.25in,
at={(0.729in,0in)},
scale only axis,
axis on top,
xmin=0.5,
xmax=128.5,
y dir=reverse,
ymin=0.5,
ymax=128.5,
hide axis
]
\addplot [forget plot] graphics [xmin=0.5,xmax=128.5,ymin=0.5,ymax=128.5] {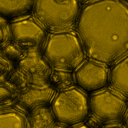};
\node[fill = none,text=white, draw=none,anchor = north west] at (rel axis cs:0,1.01) {\fontsize{6}{1}\selectfont \textbf{(a) 700 mA, 100 \%}}; 
\node[fill = none,text=white, draw=none,anchor = north west] at (rel axis cs:0,0.9)  {\fontsize{6}{1}\selectfont \textbf{Ground truth}}; 
\filldraw[color = white, fill = none, thick] (94.5,32.5)rectangle (126.5,64.5);
\end{axis}
\end{tikzpicture}%
    \end{minipage}
    \begin{minipage}{0.19\linewidth}
      \centering
%
%
\begin{tikzpicture}[scale = 0.8]

\begin{axis}[%
width=1.25in,
height=1.25in,
at={(0.729in,0in)},
scale only axis,
axis on top,
xmin=0.5,
xmax=128.5,
y dir=reverse,
ymin=0.5,
ymax=128.5,
hide axis
]
\addplot [forget plot] graphics [xmin=0.5,xmax=128.5,ymin=0.5,ymax=128.5] {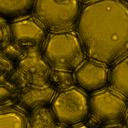};
\node[fill = none,text=white, draw=none,anchor = north west] at (rel axis cs:0,1.02) {\fontsize{6}{1}\selectfont \textbf{(b) P\textsuperscript{1}: 700 mA, 8 \%}}; 
\node[fill = none,text=white, draw=none,anchor = north west] at (rel axis cs:0,0.9)  {\fontsize{6}{1}\selectfont \textbf{RSNR = 12.93 dB}}; 
\filldraw[color = white, fill = none, thick] (94.5,32.5)rectangle (126.5,64.5);
\end{axis}
\end{tikzpicture}%
    \end{minipage}
    \begin{minipage}{0.19\linewidth}
      \centering
%
%
\begin{tikzpicture}[scale = 0.8]

\begin{axis}[%
width=1.25in,
height=1.25in,
at={(0.729in,0in)},
scale only axis,
axis on top,
xmin=0.5,
xmax=128.5,
y dir=reverse,
ymin=0.5,
ymax=128.5,
hide axis
]
\addplot [forget plot] graphics [xmin=0.5,xmax=128.5,ymin=0.5,ymax=128.5] {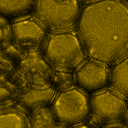};
\node[fill = none,text=white, draw=none,anchor = north west] at (rel axis cs:0,1.02) {\fontsize{6}{1}\selectfont \textbf{(c) P\textsuperscript{2}: 100 mA, 32 \%}}; 
\node[fill = none,text=white, draw=none,anchor = north west] at (rel axis cs:0,0.9) {\fontsize{6}{1}\selectfont \textbf{RSNR = 7.86 dB}};
\filldraw[color = white, fill = none, thick] (94.5,32.5)rectangle (126.5,64.5); 
\end{axis}
\end{tikzpicture}%
    \end{minipage}
    \begin{minipage}{0.19\linewidth}
      \centering
%
%
\begin{tikzpicture}[scale = 0.8]

\begin{axis}[%
width=1.25in,
height=1.25in,
at={(0.729in,0in)},
scale only axis,
axis on top,
xmin=0.5,
xmax=128.5,
y dir=reverse,
ymin=0.5,
ymax=128.5,
hide axis
]
\addplot [forget plot] graphics [xmin=0.5,xmax=128.5,ymin=0.5,ymax=128.5] {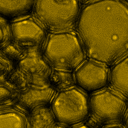};
\node[fill = none,text=white, draw=none,anchor = north west] at (rel axis cs:0,1.02) {\fontsize{6}{1}\selectfont \textbf{(d) P\textsuperscript{3}: 700 mA, 14 \%}}; 
\node[fill = none,text=white, draw=none,anchor = north west] at (rel axis cs:0,0.9)  {\fontsize{6}{1}\selectfont \textbf{RSNR = 14.53 dB}}; 
\filldraw[color = white, fill = none, thick] (94.5,32.5)rectangle (126.5,64.5);
\end{axis}
\end{tikzpicture}%
    \end{minipage}
    \begin{minipage}{0.19\linewidth}
      \centering
%
%
\begin{tikzpicture}[scale = 0.8]

\begin{axis}[%
width=1.25in,
height=1.25in,
at={(0.729in,0in)},
scale only axis,
axis on top,
xmin=0.5,
xmax=128.5,
y dir=reverse,
ymin=0.5,
ymax=128.5,
hide axis
]
\addplot [forget plot] graphics [xmin=0.5,xmax=128.5,ymin=0.5,ymax=128.5] {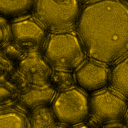};
\node[fill = none,text=white, draw=none,anchor = north west] at (rel axis cs:0,1.02) {\fontsize{6}{1}\selectfont \textbf{(e) P\textsuperscript{4}: 200 mA, 35 \%}}; 
\node[fill = none,text=white, draw=none,anchor = north west] at (rel axis cs:0,0.9)  {\fontsize{6}{1}\selectfont \textbf{RSNR = 12.30 dB}}; 
\filldraw[color = white, fill = none, thick] (94.5,32.5)rectangle (126.5,64.5);
\end{axis}
\end{tikzpicture}%
    \end{minipage}
  \end{minipage}
  \begin{minipage}{1.01\linewidth}
    \begin{minipage}{0.19\linewidth}
      \centering
      \vspace{1mm}
%
%
\begin{tikzpicture}[scale = 0.8]

\begin{axis}[%
width=1.25in,
height=1.25in,
at={(0.729in,0in)},
scale only axis,
axis on top,
xmin=96.5,
xmax=128.5,
y dir=reverse,
ymin=32.5,
ymax=64.5,
hide axis
]
\addplot [forget plot] graphics [xmin=0.5,xmax=128.5,ymin=0.5,ymax=128.5] {Sim3_res2_gt-1.png};
\node[fill = none,text=white, draw=none,anchor = north west] at (rel axis cs:0,1.01) {\fontsize{6}{1}\selectfont \textbf{4x (a)}}; 
\end{axis}
\end{tikzpicture}%
    \end{minipage}
    \begin{minipage}{0.19\linewidth}
      \centering
      \vspace{1mm}
%
%
\begin{tikzpicture}[scale = 0.8]

\begin{axis}[%
width=1.25in,
height=1.25in,
at={(0.729in,0in)},
scale only axis,
axis on top,
xmin=96.5,
xmax=128.5,
y dir=reverse,
ymin=32.5,
ymax=64.5,
hide axis
]
\addplot [forget plot] graphics [xmin=0.5,xmax=128.5,ymin=0.5,ymax=128.5] {Sim3_res2_P1_700mA_8Percent-1.png};
\node[fill = none,text=white, draw=none,anchor = north west] at (rel axis cs:0,1.01) {\fontsize{6}{1}\selectfont \textbf{4x (b)}}; 
\end{axis}
\end{tikzpicture}%
    \end{minipage}
    \begin{minipage}{0.19\linewidth}
      \centering
      \vspace{1mm}
%
%
\begin{tikzpicture}[scale = 0.8]

\begin{axis}[%
width=1.25in,
height=1.25in,
at={(0.729in,0in)},
scale only axis,
axis on top,
xmin=96.5,
xmax=128.5,
y dir=reverse,
ymin=32.5,
ymax=64.5,
hide axis
]
\addplot [forget plot] graphics [xmin=0.5,xmax=128.5,ymin=0.5,ymax=128.5] {Sim3_res2_P2_100mA_32Percent-1.png};
\node[fill = none,text=white, draw=none,anchor = north west] at (rel axis cs:0,1.01) {\fontsize{6}{1}\selectfont \textbf{4x (c)}}; 
\end{axis}
\end{tikzpicture}%
    \end{minipage}
    \begin{minipage}{0.19\linewidth}
      \centering
      \vspace{1mm}
%
%
\begin{tikzpicture}[scale = 0.8]

\begin{axis}[%
width=1.25in,
height=1.25in,
at={(0.729in,0in)},
scale only axis,
axis on top,
xmin=96.5,
xmax=128.5,
y dir=reverse,
ymin=32.5,
ymax=64.5,
hide axis
]
\addplot [forget plot] graphics [xmin=0.5,xmax=128.5,ymin=0.5,ymax=128.5] {Sim3_res2_P3_700mA_14Percent-1.png};
\node[fill = none,text=white, draw=none,anchor = north west] at (rel axis cs:0,1.01) {\fontsize{6}{1}\selectfont \textbf{4x (d)}}; 
\end{axis}
\end{tikzpicture}%
    \end{minipage}
    \begin{minipage}{0.19\linewidth}
      \centering
      \vspace{1mm}
%
%
\begin{tikzpicture}[scale = 0.8]

\begin{axis}[%
width=1.25in,
height=1.25in,
at={(0.729in,0in)},
scale only axis,
axis on top,
xmin=96.5,
xmax=128.5,
y dir=reverse,
ymin=32.5,
ymax=64.5,
hide axis
]
\addplot [forget plot] graphics [xmin=0.5,xmax=128.5,ymin=0.5,ymax=128.5] {Sim3_res2_P4_200mA_32Percent-1.png};
\node[fill = none,text=white, draw=none,anchor = north west] at (rel axis cs:0,1.01) {\fontsize{6}{1}\selectfont \textbf{4x (e)}}; 
\end{axis}
\end{tikzpicture}%
    \end{minipage}
  \end{minipage}
  \caption{The spatial maps of the reconstructed HS volumes at $\nell = 70$ (equivalent to 594\,nm wavelength).}
  \label{fig:sim3_res2}
\end{figure}
The RSNR values of CI-FTI system for two constrained-exposure budgets, \ie $I_\textsl{tot} \in \{100N_\xi, 200N_\xi\}$ mA, and for two random generations of $\Omega^\xi$ is illustrated in Fig.~\ref{fig:sim3_res1}. These results stress that for a fixed light exposure budget, an application of the proposed CI-FTI method always results in a superior quality of the HS volume reconstruction.  For instance, for $100N_\xi$ mA exposure budget, subsampling 8\,\% of the OPD axis, approximately yields a 5\,dB  gain in the RSNR, in comparison to subsampling 34\,\% of the OPD axis. The reconstructed HS volumes corresponding to four circled points in Fig.~\ref{fig:sim3_res1} are shown in Fig.~\ref{fig:sim3_res3} (for the spectra observed at the center spatial location) and Fig.~\ref{fig:sim3_res2} (for the spatial maps at wavenumber index 70, \ie 594\,nm, at which the spectra maximum occurs). In Fig.~\ref{fig:sim3_res3}, the noise amplitudes in the spectra recovered from P$^2$ (34\,\% of FTI measurements) and P$^{4}$ (35\%) are significantly larger than in the spectra of P$^{1}$ (8\,\%) and P$^3$ (14\%), respectively.  In Fig.~\ref{fig:sim3_res2}, as expected from the low light condition (\eg at 100 mA and 200 mA), the spatial map quality is significantly degraded due to noise (Fig.~\ref{fig:sim3_res2}-c and Fig.~\ref{fig:sim3_res2}-e). Lower subsampling rates allows for improved spatial map qualities (Fig.~\ref{fig:sim3_res2}-b and Fig.~\ref{fig:sim3_res2}-d) associated with an increased light intensity of 700 mA where the MNR peaks. Note that the parallel frontiers of the biological cells in Fig.~\ref{fig:sim3_res2} are an effect of the 3D specimen transparency (also observed in the panchromatic microscope).

We conclude this section by mentioning that an SI-FTI framework could have also been simulated from the recorded data, \ie by randomly subsampling spatiotemporally the volume of recorded (Nyquist) interferograms. However, such a simulation would be too ideal with respect to an actual SI-FTI implementation integrating a spatial light modulator (SLM), \eg a semi-transparent Liquid Crystal Display (LCD), Liquid Crystal on Silicon (LCoS) \cite{nagahara2010programmable}, or Digital Micro-mirror Devices (DMD), as used in the single-pixel camera \cite{duarte2008single}. While their inclusion in a compressive imaging procedure is often beneficial, these devices are also known to induce non-negligible light diffraction associated with the small pixel pitch of the SLMs's elements. We therefore postpone the analysis of an actual SI-FTI to future study where the impact of light diffraction (\eg modeled by a spatial convolution with a calibrated point-spread-function~\cite{degraux2018multispectral}) must be carefully integrated in the sensing model.

\section{Conclusion}                             
\label{sec:Sec_07}
We have proposed two versions of compressive sensing Fourier transform interferometry (\ie CI-FTI and SI-FTI) where the light exposure can be compressed temporally and spatially in order to maximize the spectral  resolution of the resulting hyperspectral volumes, \ie in a process that minimizes the photo-bleaching phenomenon. These methods are practically plausible without any modification of interferometer but only of the optical setup to which it is associated (\eg the light system of a confocal microscope). The first proposed system, called Coded illumination-FTI, consists in binary modulation of illuminating light before exposing the biological specimen; while the second system, referred as  Structured Illumination-FTI, involves a spatial light-modulation that allows spatiotemporal light coding. By invoking the theory of compressive sensing and by deriving a variable density sampling strategy guided by~\cite{krahmer2014stable}, the two proposed frameworks are proved to reach optimum uniform recovery guarantees. Furthermore, the impact of promoting fluorochrome life-time as a constraint has been analyzed. Our theoretical analyses were verified via exhaustive numerical tests for three cases: \emph{(i)} noiseless synthetic arbitrary sparse HS data, \emph{(ii)} simulated biological HS volume, and \emph{(iii)} Experimental FTI measurement. However, future investigations can be followed in several directions. For instance, assuming other discrepancy sources, \eg Poisson noise, quantization, instrumental response, or taking into account the structure of the fluorochrome spectra in order to improve the HS recovery. Tracing other low-complexity prior models, \eg low-rankness, group sparsity, dictionary-based sparsity, shearlet, or TV sparsity models will be the scope of a future work.  
\appendix

\section{Proof of Theorem~\ref{prop:gen-bound-noise-vds}}
\label{sec:proof_prop_gen-bound-noise-vds}
In the context defined by Thm~\ref{prop:gen-bound-noise-vds}, we define $M$ independent, positive \rvs $X_j \sim |n_j|^2 \,\eta(\beta)^{-1}$ (for $1\leq j \leq M$), such that $\bb E X_j = |n_j|^2 \,\bb E \eta(\beta)^{-1} = |n_j|^2 N$, and $\bb E \sum_j X_j = N \|\bs n\|^2$. Therefore, the left-hand side of~\eqref{eq:gen-bound-noise-vds-ext-noise} reads 
  \begin{equation}
    \label{eq:appA-tmp1}
    \ts \tinv{M} \|\bs D \bs n\|^2 = \tinv{M}  \sum_{j=1}^M X_j = \frac{N}{M} \|\bs n\|^2 + \tinv{M} \sum_{j=1}^M (X_j - \bb E X_j).  
  \end{equation}
Moreover, for all $j \in \range{M}$ and any integer $p\geq 1$, $\bb E X_j^p \leq \|\bs n\|_\infty^{2p}\,\bb E_\beta \eta(\beta)^{-p}$, so that, for $|s| < 1/(4e\rho N \|\bs n\|_\infty^{2})$ and using $p! \geq e^{1-p} p^p$ (Stirling bound~\cite{rigollet2015high}),
\begin{align*}
\ts \bb E e^{s (X_j - \bb E X_j)}&\ts =\ \sum_{p=0}^{+\infty} \frac{1}{p!} s^p\, \bb E (X_j - \bb E X_j)^p = 1 + \sum_{p=2}^{+\infty} \frac{1}{p!} s^p\, \bb E (X_j - \bb E X_j)^p\\
&\ts \leq 1 + \sum_{p=2}^{+\infty} \frac{1}{p!} 2^p |s|^p\, \bb E X_j^p\\
&\ts \leq 1 + \sum_{p=2}^{+\infty} 2^p |s|^p\, \|\bs n\|_{\infty}^{2p}\ \frac{1}{p!} \bb E_\beta \eta(\beta)^{-p}\\
&\ts \leq 1 + \sum_{p=2}^{+\infty} 2^p |s|^p\, \|\bs n\|_{\infty}^{2p}\ e^{p-1} \big(\frac{1}{p} (\bb E_\beta \eta(\beta)^{-p})^{1/p}\big)^p\\
&\ts \leq 1 + \sum_{p=2}^{+\infty} 2^p |s|^p\, \|\bs n\|_{\infty}^{2p}\, e^{p-1}\rho^p N^p\\
&\ts = 1 + 4 e s^2\, \|\bs n\|_{\infty}^{4}\, \rho^2 N^2 \sum_{p=0}^{+\infty} 2^p e^{p} |s|^p\, \|\bs n\|_{\infty}^{2p}\, \rho^p N^p\\
&\ts \leq 1 + 8 e s^2\, \|\bs n\|_{\infty}^{4}\, \rho^2 N^2\ \leq\ \exp( 16 e^2 s^2\, \|\bs n\|_{\infty}^{4}\, \rho^2 N^2 /2),
\end{align*}
where the second line uses $\bb E_X |X - \bb E X|^p = \bb E_X |\bb E_{X'}(X - X')|^p \leq \bb E_X \bb E_{X'}(|X| + |X'|)^p \leq 2^{p-1} \bb E_X \bb E_{X'}(|X|^p + |X'|^p) \leq 2^p \bb E |X|^p$, for two \iid \rvs $X$ and $X'$. 

This shows that the \rvs $\{X_j - \bb E X_j: 1\leq j \leq M\}$ are positive sub-exponential \rvs with parameter $\lambda = 4 e\|\bs n\|_\infty^{2} \rho N$ \cite{rigollet2015high}, since $\bb E e^{s (X_j - \bb E X_j)} \leq \exp (s^2\lambda^2/2)$ for $|s|<1/\lambda$. Therefore, Bernstein inequality~\cite[Thm. 1.13]{rigollet2015high} shows that their sum concentrates around their mean, \ie 
\begin{equation*}
  \ts \bb P[\frac{1}{M} \sum_{j=1}^M (X_j -\bb E X_j) > t \lambda] \leq \exp(-\frac{M}{2}\min\{t^2,t\}).  
\end{equation*}
With the change of variable $\min\{t^2,t\} = s$, \ie $t = \max\{\sqrt s, s\}$, this gives $\ts \bb P[\frac{1}{M} \sum_{j=1}^M (X_j -\bb E X_j) > \max \{s , \sqrt s\} \lambda] \leq \exp(-\frac{M}{2} s)$, or equivalently, with $s \leftarrow s/M$, 
$$
\ts \bb P[\frac{1}{M} \sum_{j=1}^M (X_j -\bb E X_j) >  \max\big\{\frac{s}{M}, \frac{\sqrt s}{\sqrt M}\big\} \lambda\,] \leq \exp(-\frac{s}{2}).
$$
Therefore, with probability exceeding $1-\exp(-\frac{s}{2})$, \eqref{eq:appA-tmp1} provides 
\begin{align*}
\ts \frac{1}{M}\,\|\bs D \bs n\|^2 \leq \frac{N}{M} \|\bs n\|^2 + 4 e\max\big\{\frac{s}{M}, \frac{\sqrt s}{\sqrt M}\big\}\|\bs n\|_\infty^{2} \rho N,
\end{align*}
which gives the result. 

\section{Proof of Corollary~\ref{cor:noise-level-estim-Gauss-noise}}\label{sec:proof_cor_noise-level-estim-Gauss-noise}
By union bound, we first know that $\bb P[\exists i \in \range{M}: | n_i| \geq t \sigma] \leq M \exp(-t^2/2)$. 
Therefore, $\|\bs n\|_\infty \leq \sigma \sqrt{2\log M +s}$ with probability exceeding $1-\exp(-s/2)$. Moreover, since $n_i^2$ is a $\chi^2$ distribution, we have $\|\bs n\|^2 \leq \sigma^2 (M + \sqrt{M}\sqrt{s/2} + s)$ with probability exceeding $1 - \exp(-s/2)$ (see, \eg \cite[Lem. 1]{laurent2000adaptive}).

Therefore, by union bound over the failure of these two events and over the one covered by Thm.~\ref{prop:gen-bound-noise-vds} in \eqref{eq:gen-bound-noise-vds-ext-noise}, we have 
\begin{align*}
\ts \frac{1}{M}\,\|\bs D \bs n\|^2&\ts \leq \frac{N}{M}\|\bs n\|^2 + 4e \max\big\{\frac{s}{M}, \frac{\sqrt s}{\sqrt M}\big\}\|\bs n\|_\infty^{2} \rho\,N\\
&\ts \leq \sigma^2 \big[ \frac{N}{M} (M + \frac{1}{\sqrt 2}\sqrt{M}\sqrt{s} + s) + 4e (2\log M +s) \max\big\{\frac{s}{M}, \frac{\sqrt s}{\sqrt M}\big\} \rho\,N \big]\\
&\ts = \sigma^2 N\,\big[ (1 + \frac{1}{\sqrt 2}\frac{\sqrt{s}}{\sqrt{M}} + \frac{s}{M}) + 4e (2\log M +s) \max\big\{\frac{s}{M}, \frac{\sqrt s}{\sqrt M}\big\} \rho \big],
\end{align*}
with probability exceeding $1-3\exp(-\frac{s}{2})$, which provides the result. Note that in the case of a UDS, we just need to consider the $\chi^2$-bound on $\|\bs n\|^2$ above since $\bs D^2 = N \Id_N$ is deterministic.

\section{Proof of Prop.~\ref{prop:cifti local coherence}}\label{sec:proof cifti}
The proof of Prop.~\ref{prop:cifti local coherence} requires us to compute here
the local coherence between the 1D discrete Fourier and Haar wavelet bases. We tightly follow the developments of \cite[Cor. 6.4]{krahmer2014stable}, improving them by a factor of $9 \pi$, which leads to a smaller hidden constant in the sample complexity bound. 
Let us first recall the definitions of the 1D discrete Fourier and Haar wavelet bases of~$\bb C^{N}$, for some dimension $N \in \bb N_0$.
\begin{definition}[1D Haar wavelet basis]\label{def:one-dimensional haar}
  Fix ${N} = 2^{\bar n}$ for some $\bar n \in \bb{N}$. The 1D discrete Haar wavelet basis of $\bb C^{N}$ consists of the functions 
$$
\ts \{\tilde{\psi}_j\}_{j=1}^{{N}}:=\{\psi\} \cup \{h_{n,\nell} : 0 \leq n \leq \bar n-1,\, 0\leq \nell \leq 2^n-1\},
$$
where, for $t \in \bb N \cap [0, {N}-1]$, $\psi(t) = 2^{-{\bar n}/2}$ is the constant (scaling) function and $h_{n,\nell} (t) := 2^{\frac{n-{\bar n}}{2}} h(2^{n-{\bar n}}\, t - \nell)$ is the wavelet function at resolution $n$ and position $\nell$, with $h(t)$ equals $1$, $-1$ and $0$ over $[0,1/2)$, $[1/2, 1)$ and $\bb R \setminus [0,1)$, respectively, \ie 
  \begin{equation*}
    h_{n,\nell}(t) =
    \begin{cases}
      2^{\frac{n-{\bar n}}{2}} &{\rm if}\ \nell 2^{{\bar n}-n} \le t < (\nell +\frac{1}{2})2^{{\bar n}-n},\\ 
      -2^{\frac{n-{\bar n}}{2}} &{\rm if}\ (\nell +\frac{1}{2}) 2^{{\bar n}-n} \le t < (\nell +1)2^{{\bar n}-n},\\ 
      0&{\rm \normalfont otherwise}.
    \end{cases}
  \end{equation*}
\end{definition}

\begin{definition}[1D discrete Fourier basis]\label{def:one-dimensional Fourier}
  Fix ${N} = 2^{\bar n}$ for some ${\bar n} \in \bb{N}$. The 1D discrete Fourier basis of $\bb C^{N}$ consists of the functions
  \begin{equation*}
    \ts \big\{\phi_k(t) := \frac{1}{\sqrt{{N}}}e^{2 \pi i \frac{k}{{N}}t} :\ -\frac{{N}}{2}+1 \leq k \leq \frac{{N}}{2},\ t \in [0, {N}-1]\big\},
\end{equation*}
with $k$ and $t$ integers.
\end{definition}

Let us now bound the local coherence 
$$
\ts \mu_{k}^{\rm loc}(\bs F, \bs{\Psi}_{\rm 1D}) := \mu_{k}^{\rm loc}(\bs{A}:=\bs F^* \bs{\Psi}_{\rm 1D}) = \max_{1\le j\le {N}} \lvert\langle \phi_{k-(N/2)},\tilde{\psi}_j \rangle \rvert,
$$
for $k \in \range{N}$. Given $k' = k - N/2$ with $-\frac{{N}}{2}+1 \leq k' \leq \frac{{N}}{2}$, we have to compute three cases: \emph{(i)} $\lvert\langle\phi_{k'}, \psi\rangle\rvert$ for all ${k'}$, \emph{(ii)} $\lvert\langle\phi_0, h_{n,\nell}\rangle\rvert$ for all $n,\nell$, and \emph{(iii)} $\lvert\langle\phi_{k'}, h_{n,\nell}\rangle\rvert$ for all non-zero ${k'}$ and all $n,\nell$. For the first two cases, we observe that $\lvert\langle\phi_{k'}, \psi\rangle\rvert$ equals one if ${k'}=0$ and zero otherwise, while $\lvert\langle\phi_0, h_{n,\nell}\rangle\rvert = 0$, for all $n$ and $\nell$. For the third case, if ${k'}\neq 0$, a direct computation provides $s := \langle\phi_{k'}, h_{n,\nell}\rangle = s_1 - s_2$ with 
\begin{equation}
  s_1 := \sum_{j = \nell 2^{{\bar n}-n}}^{(\nell+\frac{1}{2})2^{{\bar n}-n}-1} 2^{\frac{n-{\bar n}}{2}}2^{-\frac{{\bar n}}{2}}e^{-2 \pi i 2^{-{\bar n}} {k'} j},\ s_2 := \sum_{j = (\nell+\frac{1}{2})2^{{\bar n}-n}}^{(\nell+1)2^{{\bar n}-n}-1} 2^{\frac{n-{\bar n}}{2}}2^{-\frac{{\bar n}}{2}}e^{-2 \pi i 2^{-{\bar n}} {k'} j}.
  \label{eq:s1-s2}
\end{equation}
Separate computations of $s_1$ and $s_2$ yields
\begin{align}
  s_1 &\ts = 2^{\frac{n}{2}-{\bar n}} e^{-2\pi i {k'} \nell 2^{-n}} \sum_{j=0}^{2^{{\bar n}-n-1}-1} e^{-2\pi i 2^{-{\bar n}}{k'} j},
	\label{eq:s1}\\
  s_2 &\ts = 2^{\frac{n}{2}-{\bar n}} e^{-2\pi i {k'} (\nell+\frac{1}{2}) 2^{-n}} \sum_{j=0}^{2^{{\bar n}-n-1}-1} e^{-2\pi i 2^{-{\bar n}}{k'} j}.
	\label{eq:s2}
\end{align} 
Therefore, we get

\begin{align*}
 s &\ts = s_1 - s_2 = 2^{\frac{n}{2}-{\bar n}}\ \big| \sum_{j=0}^{2^{({\bar n}-n-1)}-1} e^{-2\pi i 2^{-{\bar n}}{k'} j}\big|\ \big|1-e^{-2\pi i {k'} 2^{-(n+1)}}\big|\\
                                             &\ts = 2^{\frac{n}{2}-{\bar n}}\ \frac{2\sin^2 (\pi {k'} 2^{-n-1})}{|\sin(\pi {k'} 2^{-{\bar n}})|} \le 2^{\frac{n}{2}} \frac{\sin^2(\pi {k'} 2^{-n-1})}{|{k'}|}.
\end{align*}
In the last inequality we used the fact that $|\sin(\pi x)| \ge 2|x|$ for $|x| \le 1/2$.
Let us consider two cases. First, if $\pi |{k'}| /2^{-n} \ge \pi/2$, then $2^{n/2} < \sqrt{2|{k'}|}$ and $s\le \sqrt{2/|{k'}|}$. Second, if $\pi |{k'}| /2^{-n} < \pi/2$, then using the fact that $\cos |x| \ge 1-2/\pi |x|$ for $|x|<\pi/2$ we get $\sin^2 \pi {k'} 2^{-n-1} = 1 - \cos \pi {k'} 2^{-n} \le |{k'}| 2^{-n}$ which leads to $s \le 2^{-n/2} \le 1/\sqrt{2|{k'}|}$. Combining the two cases leads to 
\begin{align*}
  \ts |\langle\phi_{{k'}},
  h_{n,\nell}\rangle| \le\ \frac{\sqrt{2}}{\sqrt{|{k'}|}}.
\end{align*}

Gathering all results, we thus find $\underset{j}{\max}\lvert\langle \phi_{k'},\tilde{\psi}_j\rangle\rvert \le \min\big\{1,\frac{\sqrt 2}{\sqrt{|{k'}|}}\big\}$, and 
\begin{equation}
  \ts \mu^{\rm loc}_{k}(\bs{A}) \le \kappa_{k} := \sqrt 2 \min \Big \{1,\frac{1}{\sqrt{|k-{N}/2|}} \Big\}.\label{eq:app local coherence 1D}
\end{equation}

Therefore, \eqref{eq:vds pmf} gives $p(k) := C_{{N}} \min \big\{1,\frac{1}{|k-{N}/2|}\big\}$ for $k \in \range{{N}}$, which implies~\eqref{eq:pmf-ci-FTI} for $N = N_\xi$, with $C_{N}$ ensuring $\sum_{k=1}^{{N}} p(k)=1$. Concerning this constant, we find 
\begin{align*}
\ts C_{N}^{-1}&\ts = \sum_k \min\{1,\frac{1}{|k-\frac{{N}}{2}|}\} = 1 + \sum_{k=1}^{\frac{{N}}{2}-1} \frac{1}{\frac{{N}}{2}-k} + \sum_{k=\frac{{N}}{2}+1}^{{N}} \frac{1}{k-\frac{{N}}{2}}\\
&\ts = 1 + \sum_{k=1}^{\frac{{N}}{2}-1} \frac{1}{k} + \sum_{k=1}^{\frac{{N}}{2}} \frac{1}{k} = 1 + \frac{2}{{N-2}} + 2 \sum_{k=1}^{\frac{{N}}{2}-1} \frac{1}{k}.
\end{align*}
However, for any integer $D\geq 2$, $\sum_{k=1}^{D-1} \frac{1}{k} \leq 1+\int_2^D \frac{1}{s-1} \ud s =1+ \log(D-1)$ and $\sum_{k=1}^{D-1} \frac{1}{k} \geq \int_1^D \frac{1}{s} \ud s =\log(D)$. 
Therefore,  
\begin{equation}
  \label{eq:bound-CN}
\ts 2\log(\frac{{N}}{2}) \le 1 + \frac{2}{{N-2}} + 2\log(\frac{{N}}{2}) \le \ts C_{N}^{-1} \leq 3 + \frac{2}{{N-2}} + 2\log(\frac{{N}}{2}-1) <4 + 2\log(\frac{{N}}{2}) ,  
\end{equation}
which is true from ${N}\geq 2$.  Finally, from the definition of $\bs \kappa$, we have 
$\|\bs \kappa\|^2/2 \leq C_{N}^{-1}$, so that $\|\bs \kappa\|^2 \leq 8 + 4\log(\frac{{N}}{2}) \lesssim \log {N}$, and \eqref{eq:rip_for_onb_vds} then explains the sufficient condition~\eqref{eq:cifti m} for $N = N_\xi$.

\section{Proof of Prop.~\ref{prop:sifti local coherence}}
\label{sec:proof sifti}
This proof requires us to compute a bound on the local coherence $\mu_k^{\rm loc}(\bs \Phi, \bs \Psi) := \mu_k^{\rm loc}(\bs A := \bs \Phi^* \bs \Psi)$ for $k \in \range{N}$, when $\bs \Phi = \bs{I}_{N_{\rm p}}\otimes \bs{F}$ and $\bs \Psi = \bs{\Psi}_{\rm 2D}\otimes \bs{\Psi}_{\rm 1D}$. 

We note first that, from the Kronecker product properties, $\bs{A} = (\bs{I}_{N_{\rm p}}\bs{\Psi}_{\rm 2D})\otimes (\bs{F}^*\bs{\Psi}_{\rm 1D})$. Using the relation \eqref{eq:transl-1D-2D-ind-repr} between the 1D ``$k$'' and the 2D ``$(j,\nell)$'' index representations of interferometric data, the definition of local coherence \eqref{eq:loc-coh} gives
\begin{equation}
  \mu_k^{\rm loc}(\bs{A}) =
  \underbrace{\underset{\nell'}{\max}|(\bs{F}^*\bs{\Psi}_{\rm 1D})_{\nell,\nell'}|}_{=:~
    \mu_{\nell}^{\rm loc}(\bs{F}^*\bs{\Psi}_{\rm 1D})}\cdot\underbrace{\underset{j'}{\max}|(\bs{I}_{N_{\rm p}}\bs{\Psi}_{\rm 2D})_{j,j'}|}_{=:~
    \mu_{j}^{\rm loc}(\bs{I}_{N_{\rm p}}\bs{\Psi}_{\rm 2D})},\label{eq:app local coherence 3d}
\end{equation}
since 
\begin{equation}
  \label{eq:max-kron-prod}
\ts \max_{s} |(\bs B \otimes \bs C)_{r,s}| = \max_{s'} |B_{r',s'}| \max_{s''} |C_{r'',s''}|,  
\end{equation}
for any matrices $\bs B \in \bb C^{b_1 \times b_2}, \bs C \in \bb C^{c_1 \times c_2}$, and $r = (r'-1)b_1 + r''$.

Let us now bound the two terms of right-hand side of \eqref{eq:app local coherence 3d}. For the first, \eqref{eq:app local coherence 1D} provides
\begin{equation}
  \ts \mu_{\nell}^{\rm loc}(\bs{F}^*\bs{\Psi}_{\rm 1D}) \le \sqrt 2\, \min\Big\{1,\frac{1}{\sqrt{|\nell-N_\xi/2|}}\Big\}.\label{eq:app local coherence 2D first part}
\end{equation}

Concerning the second, we are going to show that $\mu_{j}^{\rm loc}(\bs{I}_{N_{\rm p}}\bs{\Psi}_{\rm 2D}) = 1/2$ for the two possible types of 2D wavelet constructions, \ie the isotropic~\cite{mallat2008wavelet} and anisotropic schemes~\cite{nowak1999wavelet}.  Let us first recall the definition of these 2D bases.

\begin{definition}[2D Haar wavelet basis with isotropic levels]\label{def:two-dimensional isotropic haar}
  Fix $N = 2^{\bar n}$ for some $\bar n \in \bb{N}$. For the resolution $0 \leq n \leq \bar n-1$, and the position parameter $0\leq \nell \leq 2^n-1$, we define
  \begin{equation*}
      h^0_{n,\nell}(t) =
      \begin{cases}
        2^{\frac{n-{\bar n}}{2}} &{\rm if}\ \nell 2^{{\bar n}-n} \le t < (\nell +1)2^{{\bar n}-n},\\ 
        0&{\rm \normalfont otherwise}.
      \end{cases}
  \end{equation*}
  The 2D Haar wavelet basis with isotropic levels $\ts \{\tilde{\psi}_j\}_{j=1}^{N^2}$ of $\bb C^{N \times N}$ consists of the functions 
  \begin{align*}
    \phi^0(t_1,t_2)&= \psi(t_1)\, \psi(t_2),\\
    \phi^1_{n, (\nell_1,\nell_2)}(t_1,t_2)&= h^0_{n,\nell_1}(t_1)\, h_{n,\nell_2}(t_2),\\
    \phi^2_{n, (\nell_1,\nell_2)}(t_1,t_2)&= h_{n,\nell_1}(t_1)\, h^0_{n,\nell_2}(t_2),\\
    \phi^3_{n,(\nell_1,\nell_2)}(t_1,t_2)&= h_{n,\nell_1}(t_1)\, h_{n,\nell_2}(t_2),
  \end{align*}
  with $0 \leq n \leq \bar n-1$ and $0\leq \nell_1, \nell_2 \leq 2^n-1$, which provides $N^2$ possible functions. Above, $h_{n,\nell}(t)$ and $\psi(t)$ are the elements of the 1D Haar wavelet basis (see Def.~\ref{def:one-dimensional haar}).
\end{definition}

\begin{definition}[2D Haar wavelet basis with anisotropic levels]
\label{def:two-dimensional anisotropic haar}
A 2D wavelet basis of $\bb C^{N \times N}$ with anisotropic levels can be decomposed into the Kronecker product of two replicates of an 1D wavelet basis of $\bb C^N$. In the case of the Haar wavelet transform, we have $\bs{\Psi}_{\rm 2D} = \bs{\Psi}_{\rm 1D} \otimes \bs{\Psi}_{\rm 1D}$. In other words, if $\bs x = {\rm vec}(\bs X) \in \bb C^{N^2}$ is the vectorization of an $(N\times N)$ image $\bs X \in \bb C^{N \times N}$, we have 
$\ts \bs{\Psi}_{\rm 2D} \bs x = {\rm vec}(\bs{\Psi}_{\rm 1D}^* \bs X \bs{\Psi}_{\rm 1D})$. 
\end{definition}

Identifying in these definitions $N$ with $\bar N_{\rm p}$ (and thus $N^2$ with $N_{\rm p} = \bar N^2_{\rm p}$), we can now proceed and bound the local coherence of these two bases when the sensing basis is the Dirac basis $\bs I_{N_{\rm p}}$. 
\medskip
 
\paragraph{(i) Anisotropic case} In this case, from Def.~\ref{def:two-dimensional anisotropic haar} and since $\bs{I}_{N_{\rm p}} = \bs{I}_{\bar N_{\rm p}} \otimes \bs{I}_{\bar N_{\rm p}}$, we find 
\begin{align*}
  &\mu_{j}^{\rm loc}(\bs{I}_{N_{\rm p}}\bs{\Psi}_{\rm 2D}) = \mu_{j}^{\rm loc}\big((\bs{I}_{\bar N_{\rm p}}\bs{\Psi}_{\rm 1D})\otimes (\bs{I}_{\bar N_{\rm p}}\bs{\Psi}_{\rm 1D})\big)\\
&=\underbrace{\underset{j'_1}{\max}|(\bs{I}_{\bar N_{\rm p}}\bs{\Psi}_{\rm 1D})_{j_1,j'_1}|}_{=:\mu_{j_1}^{\rm loc}(\bs{I}_{\bar N_{\rm p}}\bs{\Psi}_{\rm 1D})} . \underbrace{\underset{j'_2}{\max}|(\bs{I}_{\bar N_{\rm p}}\bs{\Psi}_{\rm 1D})_{j_2,j'_2}|}_{=:\mu_{j_2}^{\rm loc}(\bs{I}_{\bar N_{\rm p}}\bs{\Psi}_{\rm 1D})}, \label{eq:app coeherence 2d wavelet}
\end{align*}
where we invoked again \eqref{eq:max-kron-prod} using the 1D pixel indexing $j := j_1+\bar N_{\rm p} (j_2-1)$ associated with the 2D pixel indexing $j_1, j_2 \in \range{\bar N_{\rm p}}$. Let $e_k$ be the $k^{\rm th}$ element of the Dirac basis of $\bb C^{\bar N_{\rm p}}$. Setting $N = \bar N_{\rm p}$ in Def.~\ref{def:one-dimensional haar}, we find $|(\bs{I}_{\bar N_{\rm p}}\bs{\Psi}_{\rm 1D})_{j_1,j'_1}|\ =\ \lvert \langle
  e_{j_1},\tilde{\psi}_{j'_1}\rangle \rvert$ and 
\begin{equation*}
  \lvert \langle
  e_{j_1},\tilde{\psi}_{j'_1}\rangle \rvert \  =\  
  \begin{cases}
    2^{\frac{-{\bar n}}{2}} &{\rm if}\ \tilde{\psi}_{j'_1}=\psi,\\ 
    2^{\frac{n-{\bar n}}{2}} &{\rm if}\ \tilde{\psi}_{j'_1}=h_{n,\nell},~0\leq n \leq {\bar n}-1,~0 \leq \nell \leq 2^{n}-1.  
  \end{cases}
\end{equation*}
Therefore, the maximum of the expression above is reached for $n = \bar n -1$ and $\mu_{j_1}^{\rm loc}(\bs{I}_{\bar N_{\rm p}}\bs{\Psi}_{\rm 1D})=\mu_{j_2}^{\rm loc}(\bs{I}_{\bar N_{\rm p}}\bs{\Psi}_{\rm 1D}) = 1/\sqrt{2}$ and $\mu_{j}^{\rm loc}(\bs{I}_{N_{\rm p}}\bs{\Psi}_{\rm 2D}) = 1/2$. 
\medskip

\paragraph{(ii) Isotropic case} Similarly to the developments above, by setting $N = \bar N_{\rm p}$ in Def.~\ref{def:two-dimensional isotropic haar}, we find
\begin{equation*}
  |(\bs{I}_{N_{\rm p}}\bs{\Psi}_{\rm 2D})_{j,j'}|\ =\ \lvert \langle
  e_{j},\tilde{\psi}_{j'}\rangle \rvert \  =\  
  \begin{cases}
    2^{-{\bar n}} &{\rm if}\ \tilde{\psi}_{j'}=\phi^0,\\ 
    2^{n-{\bar n}} &{\rm if}\ \tilde{\psi}_{j'}=\phi^s_{n, (\nell_1, \nell_2)},  
  \end{cases}
\end{equation*}
for some $s \in \{1,2,3\}$, $0\leq n \leq {\bar n}-1$, and $0 \leq \nell_1,\nell_2 \leq 2^{n}-1$. Therefore, its maximum is reached for $n = \bar n - 1$ and $\mu_{j}^{\rm loc}(\bs{I}_{N_{\rm p}}\bs{\Psi}_{\rm 2D}) = 1/2$.

\bigskip

Combining the last two cases with~\eqref{eq:app local coherence 2D first part}, and~\eqref{eq:app local coherence 3d} results in
\begin{equation}
  \ts \mu_{k}^{\rm loc}(\bs{A}) \le \kappa_k := \frac{\sqrt{2}}{2}\min\Big\lbrace 1,\frac{1}{\sqrt{|k_\xi-N_\xi/2|}}\Big\rbrace, \label{eq:app kappa function}
\end{equation}
which provides \eqref{eq:coh-bound-SI}.

The pmf associated with SI-FTI can be then formulated from~\eqref{eq:app kappa function} and~\eqref{eq:app kappa norm} as
\begin{align*}
  p(k) = p(k(j,\nell))&=\ts \frac{C_{N_\xi}}{N_{\rm p}}\min\Big\lbrace 1,\frac{1}{|\nell-N_\xi/2|}\Big\rbrace,
\end{align*}
where the normalizing constant $C_{N_\xi}$, \ie such that $\sum_{k} p(k) = 1$, is formulated in App.~\ref{sec:proof cifti}. Moreover, $\|\bs \kappa\|^2 = \frac{1}{2}N_{\rm p}\sum_{\nell=1}^{N_\xi}\min\Big\lbrace
  1,\frac{1}{|\nell-N_\xi/2|}\Big\rbrace = \frac{1}{2}N_{\rm p} C_{N_\xi}^{-1}$ so that, from~\eqref{eq:bound-CN},  
\begin{equation}
 \ts  \|\bs \kappa\|^2 \leq N_{\rm p} (2 + \log(N_\xi/2)) \lesssim N_{\rm p} \log N_\xi.\label{eq:app kappa norm}
\end{equation}

\section*{Acknowledgment}                        
\label{sec:acknowledgment}
We would like to thank Gilles Puy for his enlightening advices on estimation of the power of the noise affected by VDS. We thank Ben Adcock for the valuable discussions we had with him during the iTWIST'18 workshop (Marseille, France). We also thank the anonymous reviewers for their thorough reading and useful comments during the revision of this work.

\bibliographystyle{IEEEtran}        
\bibliography{IEEEabrv,refs}

\end{document}